\PassOptionsToPackage{dvipsnames}{xcolor}
\documentclass{lmcs}
\pdfoutput=1

\usepackage{lastpage}
\lmcsdoi{17}{2}{9}
\lmcsheading{}{\pageref{LastPage}}{}{}%
{Aug.~22,~2019}{Apr.~22,~2021}{}


\usepackage[utf8]{inputenc}
\usepackage[T1]{fontenc}
\SetLabelAlign{parright}{\parbox[t]{\labelwidth}{\raggedleft#1}}

\usepackage{bussproofs}
\usepackage{tikz}
\usepackage{tikz-cd}
\usepackage{mathrsfs}
\usepackage{amssymb}
\usepackage{mathtools}
\usepackage{stmaryrd}
\usepackage{array}
\usepackage{extarrows}
\usepackage[b]{esvect}

\makeatletter
\newsavebox{\@brx}
\newcommand{\llangle}[1][]{\savebox{\@brx}{\(\m@th{#1\langle}\)}%
\mathopen{\copy\@brx\kern-0.5\wd\@brx\usebox{\@brx}}}
\newcommand{\rrangle}[1][]{\savebox{\@brx}{\(\m@th{#1\rangle}\)}%
\mathclose{\copy\@brx\kern-0.5\wd\@brx\usebox{\@brx}}}
\makeatother
 
\newenvironment{bprooftree}
  {\leavevmode\hbox\bgroup}
  {\DisplayProof\egroup}

\newcommand{\secref}[1]{\S\ref{#1}}

\newcommand{\swap}{\emph{swap}}
\newcommand{\drop}{\emph{drop}}

\newcommand{\calc}{LNL-FPC}%

\newcommand{\Ob}{\text{Ob}}%
\newcommand{\Mor}{\text{Mor}}%

\newcommand{\id}{\text{id}}
\newcommand{\Id}{\text{Id}}

\newcommand{\op}{\ensuremath{\mathrm{op}}}

\newcommand{\AAA}{\ensuremath{\mathbf{A}}}

\newcommand{\BB}{\ensuremath{\mathbf{B}}}
\newcommand{\BBB}{\ensuremath{\mathbf{B}}}

\newcommand{\CC}{\ensuremath{\mathbf{C}}}
\newcommand{\CCe}{\ensuremath{\mathbf{C}_e}}

\newcommand{\DD}{\ensuremath{\mathbf{D}}}

\newcommand{\EE}{\ensuremath{\mathbf{E}}}

\newcommand{\LL}{\ensuremath{\mathbf{L}}}
\newcommand{\LLe}{\ensuremath{\mathbf{L}_e}}

\newcommand{\PE}{\ensuremath{\mathbf{C}_{pe}}}

\newcommand{\Poset}{\ensuremath{\mathbf{Poset}}}
\newcommand{\DCPO}{\ensuremath{\mathbf{CPO}}}

\newcommand{\dcpo}{\DCPO}
\newcommand{\dcpobs}{\ensuremath{\mathbf{CPO}_{\perp!}}}

\newcommand{\cpobs}{\dcpobs}
\newcommand{\cpo}{\dcpo}

\newcommand{\Rel}{\ensuremath{\mathbf{Rel}}}
\newcommand{\RR}{\ensuremath{\mathbf{R}}}

\newcommand{\M}{\ensuremath{\mathbf{M}}}

\newcommand{\colim}{\ensuremath{\mathrm{colim}}}
\newcommand{\tr}{\ensuremath{\ensuremath{\triangleright}}}
\newcommand{\sfold}{\ensuremath{\mathrm{fold}}}
\newcommand{\sunfold}{\ensuremath{\mathrm{unfold}}}

\newcommand{\overcup}[1]{\ensuremath{\stackrel{\smallsmile}{\smash{#1}\rule{0pt}{1.05ex}}}}

\newcommand{\lift}{\textnormal{\texttt{lift}}}
\newcommand{\lleft}{\textnormal{\texttt{left}}}
\newcommand{\rright}{\textnormal{\texttt{right}}}
\newcommand{\force}{\textnormal{\texttt{force}}}
\newcommand{\ccase}{\textnormal{\texttt{case}}}
\newcommand{\llet}{\textnormal{\texttt{let}}}

\newcommand{\eval}{\mathrm{eval}}
\newcommand{\rec}{\textnormal{\texttt{rec}}}

\newcommand{\fold}{\textnormal{\texttt{fold}}}
\newcommand{\ifold}{\ensuremath{\mathpzc{fold}}}
\newcommand{\unfold}{\textnormal{\texttt{unfold}}}
\newcommand{\iunfold}{\ensuremath{\mathpzc{unfold}}}

\newcommand{\FOLD}[1]{\mathbb{I}^{#1}}
\newcommand{\UNFOLD}[1]{\mathbb{E}^{#1}}

\newcommand{\Nat}{\textnormal{\texttt{Nat}}}
\newcommand{\List}{\textnormal{\texttt{List}}}
\newcommand{\Stream}{\textnormal{\texttt{Stream}}}

\newcommand{\REC}{\ensuremath{\mathcal R}}

\newcommand{\lrb}[1]{{\llbracket #1 \rrbracket}}
\newcommand{\sem}[1]{{\llbracket #1 \rrbracket}}
\newcommand{\flrb}[1]{{\llparenthesis #1 \rrparenthesis}}
\newcommand{\elrb}[2]{{\lVert #1 \rVert^{#2}}}
\newcommand{\elrbc}[1]{{\elrb{#1}{\vec C}}}
\newcommand{\elrbs}[1]{{\lVert #1 \rVert}}
\newcommand{\lrbtz}[1]{\lrb{#1}_{\vec Z}^{\Theta}}
\newcommand{\flrbtz}[1]{\flrb{#1}_{\vec Z}^{\Theta}}

\newcommand{\sleq}{\sqsubseteq}
\newcommand{\tleq}{\trianglelefteq}

\newcommand{\naturalto}{\ensuremath{\Rightarrow}}

\DeclareMathAlphabet{\mathpzc}{OT1}{pzc}{m}{it}

\usetikzlibrary{decorations.pathmorphing}
\usetikzlibrary{decorations.markings}
\usetikzlibrary{decorations.pathreplacing}
\usetikzlibrary{arrows}
\usetikzlibrary{shapes}

\pgfdeclarelayer{edgelayer}
\pgfdeclarelayer{nodelayer}
\pgfsetlayers{edgelayer,nodelayer,main}

\tikzstyle{braceedge}=[decorate,decoration={brace,amplitude=10pt}]
\tikzstyle{square box}=[rectangle,fill=white,draw=black,minimum height=6mm,minimum width=6mm,yshift=0.7mm]
\tikzstyle{wire label}=[font=\footnotesize, auto,swap]

\tikzstyle{none}=[inner sep=0pt]
\tikzstyle{empty}=[rectangle,fill=none,draw=none]
\tikzstyle{scaled}=[rectangle,fill=none,draw=none, font=\small]

\tikzstyle{to}=[->,draw=Black]
\tikzstyle{naturalto}=[-{Implies},double distance=1.5pt]

\tikzstyle{equal-arrow}=[double equal sign distance]

\tikzstyle{every picture}=[baseline=-0.25em]

\newcommand{\stikz}[2][1]{\scalebox{#1}{
\InputIfFileExists{#2}{}{\input{./tikz/#2}}
}}
\newcommand{\cstikz}[2][1]{\begin{center}\stikz[#1]{#2}\end{center}}

\begin{document}

\title[LNL-FPC: The Linear/Non-linear Fixpoint Calculus]{LNL-FPC: The Linear/Non-linear Fixpoint Calculus\rsuper{*}}
\titlecomment{{\lsuper*}This is an extended version of the ICFP paper \cite{icfp19}.}

\author[Bert Lindenhovius]{Bert Lindenhovius\rsuper{a}}
\address{\lsuper{a}Tulane University, New Orleans, USA}

\author[Michael Mislove]{Michael Mislove\rsuper{a}}

\author[Vladimir Zamdzhiev]{Vladimir Zamdzhiev\rsuper{b}}
\address{\lsuper{b}Universit\'e de Lorraine, CNRS, Inria, LORIA, F 54000 Nancy, France}




\begin{abstract}
\noindent 
We describe a type system with mixed linear and non-linear recursive types
called LNL-FPC (the linear/non-linear fixpoint calculus). The type system
supports linear typing, which enhances the safety properties of programs, but
also supports non-linear typing as well, which makes the type system more
convenient for programming. Just as in FPC, we show that LNL-FPC supports
type-level recursion, which in turn induces term-level recursion. We also
provide sound and computationally adequate categorical models for LNL-FPC that
describe the categorical structure of the substructural operations of
Intuitionistic Linear Logic at all non-linear types, including the
recursive ones. In order to do so, we describe a new technique for solving
recursive domain equations within cartesian categories by constructing the
solutions over pre-embeddings. The type system also enjoys implicit
weakening and contraction rules that we are able to model by identifying the
canonical comonoid structure of all non-linear types. We also show that the
requirements of our abstract model are reasonable by constructing a large class
of concrete models that have found applications not only in classical
functional programming, but also in emerging programming paradigms that
incorporate linear types, such as quantum programming and circuit description
programming languages.
\end{abstract}

\maketitle



\section{Introduction}

The \emph{Fixpoint Calculus} (FPC) is a type system that has been extensively
studied as a foundation for functional programming languages with recursive
types. Originally proposed by Plotkin~\cite{plotkin-85}, FPC was the focus of
Fiore's celebrated PhD thesis~\cite{fiore-thesis}. The seminal
paper~\cite{fiore-plotkin} gives a summary account of how axiomatic domain
theory can be used to characterize sound and computationally adequate
denotational models of FPC.

Girard's introduction of \emph{linear logic}~\cite{girard} initiated a parallel
line of research into the logics underpinning functional programming
languages~\cite{samson-linear-logic}, focusing on the analysis of
intuitionistic logic in terms of how hypotheses are consumed as resources for
proofs. Attempts to bring linear types into functional programming languages
soon followed~\cite{wadler:world}. Recently, there has been a concrete proposal
for extending Haskell with linear types~\cite{Bernardy:linear_haskell}, and new
languages are proposed that provide both linear and non-linear types with
convenience for programming a design goal~\cite{Tov:practical,Morris:best}.
One of the main benefits of linear types is the enhanced safety properties of
programs, which results from the more fine-grained control of resources, 
e.g., safe in-place updates of mutable data and safe access control to external 
resources such as files, sockets, etc. (see \cite{Bernardy:linear_haskell}).
Linear types also appear in concurrent settings where session types can be
used to model the $\pi$-calculus ~\cite{caires:session-types}; in quantum
programming languages where they are used to ensure error-freeness by enforcing
compliance with the laws of quantum mechanics~\cite{quant-semantics}; and in circuit
description programming languages where they are used to ensure that wires
do not split and do not remain unconnected~\cite{ewire-mfps,eclnl}.

In this paper we present a foundational treatment of mixed linear and
non-linear recursive types. We formulate a denotational semantics that shows
how solutions to a large class of domain equations may be interpreted in both a
cartesian category and a linear one and that the solutions are strongly
related.

\subsection*{Overview and summary of results}

To illustrate our results, we present a type system called \emph{LNL-FPC} (the
linear/non-linear fixpoint calculus). The syntax of our language
(\secref{sec:syntax}) can roughly be understood as adding linear
features to FPC, or alternatively, as adding recursive types to a
linear/non-linear lambda calculus (such as DILL~\cite{dill} or
LNL~\cite{benton-small}). More precisely, it is an extension with recursive
types of the circuit-free fragment of Proto-Quipper-M~\cite{pqm-small} that is
referred to as the CLNL calculus in~\cite{eclnl}.  The
type system has implicit weakening and contraction rules:
non-linear variables are automatically copied and discarded whenever
necessary by the language (without requiring user input). 
The non-linear terms,
types and contexts form \emph{subsets} of the terms, types and contexts of
LNL-FPC. There is no strict separation between linear and non-linear
primitives in LNL-FPC. Instead, we view the non-linear primitives as having
the additional property of being compatible with respect to the substructural
operations of contraction and weakening (copying and discarding), and the remaining primitives
are necessarily treated linearly by the type system. This leads to 
greater convenience for programming, because our type system reduces the need for applying promotion (lifting) and dereliction (forcing) operations, essentially only to linear function types.
We also equip LNL-FPC with a call-by-value big-step operational semantics and show
that type-level recursion induces term-level recursion
(\secref{sec:operational}), thus recreating a well-known result from FPC.

The primary difficulties in designing LNL-FPC are on the denotational side.
Our categorical model (\secref{sec:categorical-model}) is a $\cpo$-enriched
linear/non-linear (LNL) model~\cite{benton-small} that has suitable
$\omega$-colimits. It is given by a $\cpo$-enriched symmetric monoidal
adjunction $\stikz{lnl-fpc-model.tikz}.$ We show that the category $\LL$ (representing the linear world), is
$\cpo$-algebraically compact in the sense of~\cite{fiore-plotkin} and thus one
can solve recursive domain equations in $\LL$ by constructing limits and
colimits over embedding-projection pairs using the famous limit-colimit
coincidence theorem~\cite{smyth-plotkin:domain-equations}. The same paper shows
any mixed-variance $\cpo$-functor $T : \LL^\op \times \LL \to \LL$ may be
seen as a \emph{covariant} functor $T_e : \LL_e \times \LL_e \to \LL_e$ on the
subcategory $\LL_e$ of embeddings and we use this as a basis for the
interpretation of types. 

However, since we work in a mixed linear/non-linear setting, we also have
to explain how to solve certain recursive domain equations involving mixed-variance functors within the category $\CC$,
which is a more challenging problem. We do so by \emph{reflecting}
the solutions from $\LL_e$ via the left adjoint $F$ into the subcategory of \emph{pre-embeddings}
$\PE,$ where a pre-embedding $f \in \CC$ is a morphism, s.t. $Ff$ is an embedding in $\LL_e$.
Unlike the subcategory of embeddings $\CC_e$, the subcategory of pre-embeddings $\PE$ has
sufficient structure for constructing (parameterised) initial algebras, which moreover satisfy
important coherence properties with respect to the (parameterised) initial algebras constructed in $\LL_e$ (\secref{sec:w-categories}).

Then, the (standard) interpretation of an (arbitrary) type $\Theta \vdash A$ is a covariant functor denoted
$\lrb{\Theta \vdash A} : \LL_e^{|\Theta|} \to \LL_e$. A non-linear type $\Theta \vdash P$ admits an
\emph{additional non-linear interpretation} as a covariant functor $\flrb{\Theta \vdash P} : \PE^{|\Theta|} \to \PE$
that is strongly related to its standard interpretation via a natural isomorphism
\[ \alpha^{\Theta \vdash P}: \lrb{\Theta \vdash P} \circ F_{pe}^{\times |\Theta|} \cong F_{pe} \circ \flrb{\Theta \vdash P} : \PE^{|\Theta|} \to \LL_e , \]
where $F_{pe}$ is the restriction of $F$ to $\PE$. By exploiting this natural
isomorphism and the strong coherence properties that it enjoys, we provide a
coherent interpretation of the substructural operations of ILL at all
\emph{non-linear} types, including the recursive ones.
This then allows us to characterise the canonical comonoid structure of non-linear (recursive) types and we prove our semantics sound (\secref{sec:semantics}).
We emphasize that
even though our recursive type expressions allow use of
$\multimap$, the reason we can interpret these types as \emph{covariant}
functors is because we have identified suitable well-behaved subcategories of
$\CC$ and $\LL$ in our model.

We also show that the requirements of our abstract model are reasonable by
constructing a large class of concrete models that have been used in different
programming paradigms ranging from classical to quantum
(\secref{sub:concrete}). In addition, we present a computational adequacy
result for a (non-empty) class of models that satisfy certain additional conditions
(\secref{sec:adequate}).

Our semantic contribution can be split into two main parts. The first one is
the two-tier semantics (one linear and one non-linear), together with the
coherence properties relating them, that show how to rigorously construct the
comonoids we need for non-linear types in the linear category. However, in
order to achieve this in the presence of recursive types, our second main
contribution is fundamental. It shows how the solutions to recursive domain
equations on the linear side, which are constructed over embeddings (and that
are well-understood), can be reflected onto the cartesian side in a coherent way
by constructing them over pre-embeddings (which is novel).

Our results presented here also carry over to
two-judgement calculi, such as LNL \cite{benton-small,benton-wadler}, and it is
not hard to see that our two-tier semantics presented here allows one to extend
such two-judgement calculi with recursive types on both the linear side and the
non-linear side. More specifically, if we adopt this view, then our semantics
allows for a full range of recursive types on the linear side and for all
recursive types on the non-linear side that do not involve non-linear function
space (but which do allow linear function space followed by the right-adjoint
type constructor, making the composite type non-linear).

\subsection*{Publication History}

This article is an extended version of the ICFP paper~\cite{icfp19}. Compared to that version, extensive additions and improvements have been made, including:
\begin{itemize}
\item A more general notion of model where we no longer assume $\CC = \cpo$.
\item Adding an entire subsection devoted to showing that the formal approximation relations (needed for the adequacy proof) exist (\secref{sub:existence}).
\item Improving the presentation of parameterised initial algebras (\secref{sec:w-categories}).
\item Adding omitted proofs and lemmas throughout the entire paper.
\item Adding an example program for the factorial function (Example~\ref{ex:factorial}) that illustrates the presence of non-linear types other than those of the form $! A$. 
\end{itemize}

\section{Syntax of LNL-FPC}\label{sec:syntax}

We begin by describing the syntax of LNL-FPC. The type-level syntax is very similar to that of FPC~\cite{fpc-syntax}
and the term-level syntax can be seen as an extension with recursive types of either
the circuit-free fragment of Proto-Quipper-M~\cite{pqm-small}, or the CLNL calculus in~\cite{eclnl,eclnl2}.

\subsection{LNL-FPC Types}

We use $X,Y,Z$ to range over \emph{type variables}, which are needed in order to form recursive types.
We use $\Theta$ to range over \emph{type contexts} (see Figure~\ref{fig:syntax}).
A type context $\Theta = X_1, \ldots, X_n$ is simply a finite list of type variables.
A type context $\Theta$ is \emph{well-formed}, denoted $\vdash \Theta$, if it can be derived from the rules:
\[
    \begin{bprooftree}
    \AxiomC{{\color{white}$\Theta$}}
    \UnaryInfC{$\vdash \cdot$}
    \end{bprooftree}
    \qquad
    \begin{bprooftree}
    \AxiomC{$\vdash \Theta$}
    \RightLabel{$X \not \in \Theta$,}
    \UnaryInfC{$\vdash \Theta, X$}
    \end{bprooftree}
\]
that is, if all the variables of $\Theta$ are distinct.

We use $A,B,C$ to range over (arbitrary) types of our language (see Figure~\ref{fig:syntax} for their grammar).
A type $A$ is well-formed in type context $\Theta$, denoted $\Theta \vdash A,$ if the judgement can be derived according to the following rules:
\[
    \begin{bprooftree}
    \AxiomC{$\vdash \Theta$}
    \RightLabel{$1 \leq i \leq |\Theta|$}
    \UnaryInfC{$\Theta \vdash \Theta_i$}
    \end{bprooftree}
    \quad
    \begin{bprooftree}
    \AxiomC{$\Theta \vdash A$}
    \UnaryInfC{$\Theta \vdash !A$}
    \end{bprooftree}
    \quad
    \begin{bprooftree}
    \AxiomC{$\Theta \vdash A$}
    \AxiomC{$\Theta \vdash B$}
    \RightLabel{$\star \in \{+, \otimes, \multimap\}$}
    \BinaryInfC{$\Theta \vdash A \star B$}
    \end{bprooftree}
    \quad
    \begin{bprooftree}
    \AxiomC{$\Theta, X \vdash A$}
    \UnaryInfC{$\Theta \vdash \mu X. A$}
    \end{bprooftree}
\]
Note that if $\Theta \vdash A$ is derivable, then so is $\vdash \Theta.$
A well-formed type $A$ is \emph{closed} if $\cdot \vdash A.$


The \emph{non-linear types} form a subset of our types and we use $P,R$ to range over them (see Figure~\ref{fig:syntax} for their grammar).
Of course, they are governed by the same formation rules as those for (arbitrary) types.
A type that is not non-linear is called \emph{linear}.
Note that we do not provide a separate grammar for linear types.
This is because we will mostly be working with arbitrary types (which we do not assume we may copy/discard/promote) and with non-linear types (which we may always copy/discard/promote).
In particular, it is possible for a non-linear type to be treated as arbitrary, but not vice versa.
Thus, when we write $\Theta \vdash A$, the type $A$ may or may not be non-linear, but when we write $\Theta \vdash P$, then $P$ can be only non-linear.
In particular, we do not introduce specific notation for linear types.

\begin{exa}
\label{ex:types}
We list some important closed types that are definable in LNL-FPC. The \emph{empty} type is $0 \equiv \mu X.X,$ which is non-linear.
Another non-linear type is the \emph{unit} type, which is defined by $I \equiv\ !(0 \multimap 0).$ The type of natural numbers is $\Nat \equiv \mu X. I + X,$ which is also non-linear.
Given a closed type $A$, then lists of type $A$ are defined by $\List\ A \equiv \mu X. I + A \otimes X.$ The type $\List\ A$ is non-linear iff $A$ is non-linear.
\emph{Lazy} datatypes can be defined by making use of the $!$ and the $\multimap$ connectives. For instance,
streams of type $A$ can be defined by $\Stream\ A \equiv \mu X. A \otimes !X,$ which is a non-linear type iff $A$ is non-linear.
The ability to form recursive types using both $\multimap$ and $!$ is a
powerful feature and it allows us to \emph{derive} a term-level recursion
operator (see \secref{sub:term-recursion}). Notice that we do not allow non-linear function space $A \to B$, but by using both $\multimap$ and $!$ we may still define recursion for non-linear functions.
\end{exa}

For any types $A$ and $B$ and type variable $X$, we denote with $A[B/X]$ the
type where all free occurrences of $X$ in $A$ are replaced by $B$ (which is
defined in the standard way).

\begin{lem}
If $ \Theta, X \vdash A$ and $\Theta \vdash B,$ then $\Theta \vdash A[B/X].$ Moreover, if $A$ and $B$ are non-linear, then so is $A[B/X]$.
\end{lem}

\begin{rem}
Like FPC, our language allows \emph{nested type recursion} (e.g. $\mu X. \mu Y. I + (X \otimes Y)).$
\end{rem}

\begin{figure}
\centering
\begin{tabular}{l  l  l  l}
  Type variables & $X,Y,Z$ & & \\
  Term variables & $x,y,z$ & & \\
	Types & $A, B, C$ &                                       ::= & $X$  | $A+B$ | $A\otimes B$ | $A \multimap B$ | $!A$ | $\mu X.A$ \\
	Non-linear types & $P, R$ &                           ::= & $X$  | $P+R$ | $P\otimes R$ | $!A$ | $\mu X.P$ \\
  Type contexts & $\Theta$ &                                ::= & $X_1, X_2, \ldots, X_n$\\
  Term contexts & $\Gamma , \Sigma $ &                      ::= & $x_1: A_1, x_2: A_2, \ldots, x_n : A_n$\\
  Non-linear term contexts & $\Phi$ &                   ::= & $x_1: P_1, x_2: P_2, \ldots, x_n : P_n$\\
  Terms & $m, n, p$ & ::= & $x$ | \lleft$_{A,B} m$ | \rright$_{A,B} m$ \\
  & & &| \ccase{} $m$ \texttt{of} $\{$\lleft{} $x\to n\ $\rright{} $y \to p\}$ \\ 
  & & &| $\langle m, n \rangle$ | \llet{} $\langle x, y \rangle = m$ \texttt{in} $n$ | $\lambda x^A.m$ | $mn$  \\
  & & &| \lift{} $m$ | \force{} $m$ | $\fold{}_{\mu X.A} m$ | $\unfold{}\ m$ \\
  Values & $v,w$ & ::= & $x$ | \lleft$_{A,B} v$ | \rright$_{A,B} v$ | $\langle v, w \rangle$ | $\lambda x^A.m$ \\
  & & &| \lift{} $m$ | $\fold{}_{\mu X.A} v$ \\
\end{tabular}
\caption{Syntax of the \calc{} Calculus.}\label{fig:syntax}
\end{figure}

\subsection{LNL-FPC Terms}\label{sub:lnl-terms}

We use $x,y,z$ to range over \emph{term variables}.
We use $\Gamma, \Sigma$ to range over (arbitrary) \emph{term contexts}. A
term context is a list of term variables with types,
written as $\Gamma = x_1: A_1, \ldots, x_n:A_n.$
A term context is \emph{well-formed} in type context $\Theta$, denoted $\Theta \vdash \Gamma$, if the
judgement can be derived from the rules:
\[
    \begin{bprooftree}
    \AxiomC{$\vdash \Theta$}
    \UnaryInfC{$\Theta \vdash \cdot$}
    \end{bprooftree}
    \begin{bprooftree}
    \AxiomC{$\Theta \vdash \Gamma$}
    \AxiomC{$\Theta \vdash A$}
    \RightLabel{$x \not \in \Gamma$,}
    \BinaryInfC{$\Theta \vdash \Gamma, x: A$}
    \end{bprooftree}
\]
that is, if $\Gamma$ is a list of distinct variables with well-formed types.
A \emph{non-linear term context} is a term context whose types are all
non-linear.  We use $\Phi$ to range over non-linear term contexts.
Just as with types, we do not introduce specific notation for purely linear contexts.

The terms and values of our language are defined in Figure~\ref{fig:syntax}.
A \emph{term judgement} has the form $\Theta; \Gamma \vdash m :A$
and indicates that $m$ is a well-formed term of type $A$ in type context $\Theta$ and
term context $\Gamma$. The formation rules are shown in Figure~\ref{fig:typing-term},
under the condition that the $\Gamma$ and $\Sigma$ contexts do not have any
variables in common (one can deduce $\Phi \cap \Gamma =
\varnothing = \Phi \cap \Sigma$).
The formation rules for $\fold{}$ and $\unfold{}$ are the same as in FPC (cf.~\cite{fpc-syntax}).
Observe that if $\Theta; \Gamma \vdash m :A,$ then also $\Theta \vdash \Gamma$ and $\Theta \vdash A.$
Note that there is only one kind of term context and we do not have explicit notation to separate the linear variables from the non-linear ones.
Thus, when we write $\Theta; \Gamma \vdash m : A$, then the context $\Gamma$ could contain both linear and non-linear variables.
However, when we write $\Theta; \Phi \vdash m : A$, then the context $\Phi$ contains only non-linear variables.
The type system enforces that a linear variable is used exactly
once, whereas a non-linear variable may be used any number of times, including
zero.

\begin{figure}
{%
\small{%
  \[
    \begin{bprooftree}
    \AxiomC{$\Theta \vdash \Phi, x:A$}
    \UnaryInfC{$\Theta; \Phi, x:A \vdash x: A$}
    \end{bprooftree}
    \quad
    \begin{bprooftree}
    \def\ScoreOverhang{0.5pt}
    \AxiomC{$\Theta; \Gamma \vdash m : A$}
    \AxiomC{$\Theta \vdash B$}
    \BinaryInfC{$\Theta; \Gamma \vdash \lleft_{A,B} m : A+B$}
    \end{bprooftree}
    \quad
    \begin{bprooftree}
    \def\ScoreOverhang{0.5pt}
    \AxiomC{$\Theta; \Gamma \vdash m : B$}
    \AxiomC{$\Theta \vdash A$}
    \BinaryInfC{$\Theta; \Gamma \vdash \rright_{A,B} m : A+B$}
    \end{bprooftree}
  \]

  \[
    \begin{bprooftree}
    \def\ScoreOverhang{0.5pt}
    \AxiomC{$\Theta; \Phi, \Gamma \vdash m : A+B$}
    \AxiomC{$\Theta; \Phi, \Sigma, x : A \vdash n : C$}
    \AxiomC{$\Theta; \Phi, \Sigma, y : B \vdash p : C$}
    \TrinaryInfC{$\Theta; \Phi, \Gamma, \Sigma \vdash \ccase\ m\ \texttt{of}\ \{\lleft\ x \to n\ |\ \rright\ y \to p\} : C$}
    \end{bprooftree}
  \]

  \[
    \begin{bprooftree}
    \def\ScoreOverhang{0.5pt}
    \AxiomC{$\Theta; \Phi, \Gamma \vdash m : A$}
    \AxiomC{$\Theta; \Phi, \Sigma \vdash n : B$}
    \BinaryInfC{$\Theta; \Phi, \Gamma, \Sigma \vdash \langle m, n \rangle : A \otimes B$}
    \end{bprooftree}
    \ 
    \begin{bprooftree}
    \def\ScoreOverhang{0.5pt}
    \AxiomC{$\Theta; \Phi,\Gamma \vdash m : A\otimes B$}
    \AxiomC{$\Theta; \Phi, \Sigma,x:A,y:B \vdash n : C$}
    \BinaryInfC{$\Theta; \Phi, \Gamma, \Sigma \vdash\llet\ \langle x,y \rangle=m\ \texttt{in}\ n:C$}
    \end{bprooftree}
  \]
  
  \[
    \begin{bprooftree}
    \def\ScoreOverhang{0.5pt}
    \AxiomC{$\Theta; \Gamma, x: A \vdash m : B$}
    \UnaryInfC{$\Theta; \Gamma \vdash \lambda x^A . m : A \multimap B$}
    \end{bprooftree}
    \quad
    \begin{bprooftree}
    \def\ScoreOverhang{0.5pt}
    \AxiomC{$\Theta; \Phi, \Gamma \vdash m : A \multimap B$}
    \AxiomC{$\Theta; \Phi, \Sigma \vdash n : A$}
    \BinaryInfC{$\Theta; \Phi, \Gamma, \Sigma \vdash mn : B$}
    \end{bprooftree}
    \quad
    \begin{bprooftree}
    \def\ScoreOverhang{0.5pt}
    \AxiomC{$\Theta; \Phi \vdash m : A$}
    \UnaryInfC{$\Theta; \Phi \vdash \lift\ m :\ !A$}
    \end{bprooftree}
  \]

  \[
    \begin{bprooftree}
    \def\ScoreOverhang{0.5pt}
    \AxiomC{$\Theta; \Gamma \vdash m :\ !A$}
    \UnaryInfC{$\Theta; \Gamma \vdash \force\ m : A$}
    \end{bprooftree}
    \ 
    \begin{bprooftree}
    \def\ScoreOverhang{0.5pt}
    \AxiomC{$\Theta; \Gamma \vdash m : A[\mu X. A / X]$}
    \AxiomC{$\Theta, X \vdash A$}
    \BinaryInfC{$\Theta; \Gamma \vdash \fold{}_{\mu X.A} m: \mu X. A$}
    \end{bprooftree}
    \ 
    \begin{bprooftree}
    \def\ScoreOverhang{0.5pt}
    \AxiomC{$\Theta; \Gamma \vdash m : \mu X. A$}
    \UnaryInfC{$\Theta; \Gamma \vdash \unfold\ m : A[\mu X. A / X]$}
    \end{bprooftree}
  \]
  
  \[ \text{where } \Gamma \cap \Sigma = \varnothing.  \]
}%
}%
\caption{Formation rules for \calc{} terms.}\label{fig:typing-term}
\end{figure}

Type assignment to terms is unique, but derivations of term judgements in LNL-FPC are in general not unique,
because non-linear variables may be part of an arbitrary context $\Gamma$, that is,
our type system allows for non-linear variables to be treated as if they were linear (but not vice versa).
For example, if $\Theta \vdash P_1$ and $\Theta \vdash P_2$ are non-linear types, then:
\[
\begin{bprooftree}
\AxiomC{$\Theta; x:P_1 \vdash x: P_1$}
\AxiomC{$\Theta; y:P_2 \vdash y: P_2$}
\BinaryInfC{$\Theta; x:P_1, y:P_2 \vdash \langle x, y \rangle: P_1 \otimes P_2$}
\end{bprooftree}
\ 
\begin{bprooftree}
\AxiomC{$\Theta; x:P_1 \vdash x: P_1$}
\AxiomC{$\Theta; x:P_1, y:P_2 \vdash y: P_2$}
\BinaryInfC{$\Theta; x:P_1, y:P_2 \vdash \langle x, y \rangle: P_1 \otimes P_2$}
\end{bprooftree}
\]
are two different derivations of the same judgement. In both of these
derivations, the variable $y$ is treated as if it were linear and it is
propagated up into only one judgement. The variable $x$ we treat non-linearly
in one case (propagated up twice) and in other case we treat it as if it were
linear (propagated up once). This is because non-linear variables are allowed
to be part of arbitrary contexts (see formation rule for pairing).

This non-uniqueness is a result of the design choice to have only one kind of pairing, one kind of conditional branching, etc. (instead of having separate linear and non-linear ones), which we think
results in a more convenient syntax for programming. 
We note that the interpretation of any two derivations of the same judgement are equal (see Theorem~\ref{thm:derivations}).

\begin{exa}
The term $\cdot; \cdot \vdash \lambda x. \langle x, x \rangle : A \multimap A \otimes A$ is well-formed iff $A$ is a non-linear type.
Indeed, if $A$ is a linear type, then the term $\cdot; x: A \vdash \langle x, x \rangle : A \otimes A$ is not well-formed, because
the variable $x$ is part of a linear context and contraction is not admissible.
\end{exa}

A  \emph{term of closed type} is a term $m$, such that $\cdot\ ; \Gamma \vdash m :
A$ for some type $A$ and context $\Gamma.$ In such a situation we simply write
$\Gamma \vdash m : A.$
Naturally, we are primarily interested in these terms (observe that the term formation rules are invariant with respect to the type context).
A \emph{program} is a term $p$, such that $\cdot\ ; \cdot \vdash p : A$ for some type $A$
and we simply write $p : A$ to indicate this.

\begin{exa}
We list some important programs. We define $* \equiv \lift\ \lambda x^0. x : I,$
which is the canonical value of unit type. The zero natural number is defined by the program $\texttt{zero} \equiv \fold_{\Nat}\ \lleft_{I, \Nat}\ * : \Nat$.
The successor function can be defined by the program $\texttt{succ} \equiv \lambda n^\Nat. \fold_{\Nat}\ \rright_{I,\Nat} n : \Nat{} \multimap \Nat{}.$
\end{exa}

Given terms $m$, $n$ and a variable $x$, we denote with $m[n/x]$ the term
obtained from $m$ by replacing all free occurrences of $x$ with $n$ (which is
defined in the standard way).

\begin{lem}[Substitution]\label{lem:syntax-sub}
  If $\Theta; \Phi, \Gamma, x:A \vdash m: B$ and $\Theta; \Phi, \Sigma \vdash n : A$ and $\Gamma \cap \Sigma = \varnothing,$ then ${\Theta; \Phi, \Gamma, \Sigma \vdash m[n/x] :B.}$
\end{lem}

We say a value $\Theta; \Gamma \vdash v : P$ is \emph{non-linear} whenever the type $P$ is non-linear.

\begin{lem}\label{lem:values-syntax}
  If $\Theta; \Gamma \vdash v : P$ is a non-linear value, then $\Gamma$ is also non-linear.
\end{lem}

\subsection{Term Recursion in LNL-FPC}\label{sub:term-recursion}
Recall that in FPC, general recursion on terms may be implemented using the
\fold{} and \unfold{} terms.
The same is also true for LNL-FPC. Moreover, the derived term for recursion has
exactly the same syntax as the one in~\cite{eclnl}, where the authors showed
how to extend a mixed linear/non-linear type system with recursion.
In the next section we show it has the same operational behaviour as well.
We also note that our recursion term is very similar to one of the recursive terms in \cite[Chapter 8]{eppendahl-thesis}.

\begin{thm}\label{thm:recursion operator}
The rule
\begin{bprooftree}
\AxiomC{$\Theta; \Phi, z:!A \vdash m: A$}
\UnaryInfC{$\Theta; \Phi \vdash \rec\ z^{!A}. m: A$}
\end{bprooftree}
is derivable in \calc{},
where $\Phi$ is non-linear,
\begin{align*}
\rec\ z^{!A}. m &\equiv (\unfold\ \force\ \alpha_m^z)\alpha_m^z \qquad \text{ and} \\
\alpha_m^z &\equiv \lift\ \fold\ \lambda x^{!\mu X. (!X \multimap A)}. (\lambda z^{!A}. m)(\lift\ (\unfold\ \force\ x) x),
\end{align*}
such that $X \not \in \Theta$ and $x \not \in \Phi.$
\end{thm}
\begin{proof}
Since $ \Theta; \Phi, z:\ !A \vdash m: A $, then $\Theta \vdash A$ (see~\secref{sub:lnl-terms}) and since $X \not \in \Theta$, then
\begin{equation}\label{eq:type-ok}
\begin{bprooftree}
\AxiomC{$\Theta, X \vdash\ !X \multimap A$}
\end{bprooftree}
\quad
\text{and}
\quad
\Theta \vdash x:\ !\mu X.(!X\multimap A)
\end{equation}

We will first show that:
\begin{equation}\label{eq:x-derivation}
\begin{bprooftree}
\AxiomC{$\Theta; x:\ !\mu X.(!X\multimap A)\vdash\lift\ (\unfold\ \force\ x)x:\ !A$}
\end{bprooftree}
\end{equation}

For brevity, we write $\REC \equiv \mu X.(!X\multimap A).$ Then, we have:
\\
  \begin{minipage}{\textwidth}
  \begin{prooftree}
        \AxiomC{}
        \RightLabel{\eqref{eq:type-ok}}
        \UnaryInfC{$\Theta \vdash x:\ !\REC$}
      \UnaryInfC{$\Theta; x:\ !\REC\vdash x:\ !\REC$}
      \UnaryInfC{$\Theta; x:\ !\REC\vdash\force\ x:\REC$}
      \UnaryInfC{$\Theta; x:\ !\REC\vdash\unfold\ \force\ x:\big(!\REC\big)\multimap A$}
        \AxiomC{}
        \RightLabel{\eqref{eq:type-ok}}
        \UnaryInfC{$\Theta \vdash x:\ !\REC$}
      \UnaryInfC{$\Theta; x:\ !\REC\vdash x:\ !\REC$}
      \BinaryInfC{$\Theta; x:\ !\REC\vdash(\unfold\ \force\ x)x:A$}
      \UnaryInfC{$\Theta; x:\ !\REC\vdash\lift\ (\unfold\ \force\ x)x:\ !A$}
  \end{prooftree}
  \end{minipage}

\mbox{}

Next, we show that:
\begin{equation}\label{eq:alpha-z-m}
\begin{bprooftree}
\AxiomC{$ \Theta; \Phi, z:\ !A \vdash m: A $}
\UnaryInfC{$ \Theta; \Phi\vdash \alpha_m^z:\ !\REC $}
\end{bprooftree}
\end{equation}

Indeed, we have:

\mbox{}

\scalebox{0.89}{
\begin{minipage}{\textwidth}
\begin{prooftree}
    \def\ScoreOverhang{0.5pt}
    \def\defaultHypSeparation{\hskip .1in}
\AxiomC{$\Theta; \Phi,z:\ !A\vdash m:A$}
\UnaryInfC{$\Theta; \Phi\vdash(\lambda z^{!A}.m):\ !A\multimap A$}
\AxiomC{}
\RightLabel{\eqref{eq:x-derivation}}
\UnaryInfC{$\Theta; x:\ !\REC\vdash\lift\ (\unfold\ \force\ x)x:\ !A$}
\BinaryInfC{$\Theta; \Phi,x:\ !\REC\vdash(\lambda z^{!A}.m)(\lift\ (\unfold\ \force\ x)x):A$}
\UnaryInfC{$\Theta; \Phi\vdash\lambda x^{!\REC}.(\lambda z^{!A}.m)(\lift\ (\unfold\ \force\ x)x):(!\REC) \multimap A$}
\AxiomC{}
\RightLabel{\eqref{eq:type-ok}}
\UnaryInfC{$\Theta , X \vdash\ !X \multimap A$}
\BinaryInfC{$\Theta; \Phi\vdash\fold\ \lambda x^{!\REC}.(\lambda z^{!A}.m)(\lift\ (\unfold\ \force\ x)x):\REC$}
\UnaryInfC{$\Theta; \Phi \vdash \lift\ \fold\ \lambda x^{!\REC}. (\lambda z^{!A}. m)(\lift\ (\unfold\ \force\ x) x):\ !\REC$}
\end{prooftree}
\end{minipage}
}

\mbox{}

Finally:
\begin{prooftree}
  \AxiomC{$ \Theta; \Phi, z:\ !A \vdash m: A $}
  \RightLabel{\eqref{eq:alpha-z-m}}
	\UnaryInfC{$\Theta; \Phi\vdash \alpha_m^z:\ ! \REC$}
	\UnaryInfC{$\Theta; \Phi\vdash\force\ \alpha_m^z:\REC$}
	\UnaryInfC{$\Theta; \Phi\vdash\unfold\ \force\ \alpha_m^z: (!\REC) \multimap A$}
  \AxiomC{$ \Theta; \Phi, z:\ !A \vdash m: A $}
  \RightLabel{\eqref{eq:alpha-z-m}}
	\UnaryInfC{$\Theta; \Phi\vdash \alpha_m^z:\ ! \REC$}
	\BinaryInfC{$\Theta; \Phi\vdash(\unfold\ \force\ \alpha_m^z)\alpha_m^z : A$}
\end{prooftree}	
and therefore:
\[
\begin{bprooftree}
  \AxiomC{$ \Theta; \Phi, z:\ !A \vdash m: A $}
	\UnaryInfC{$\Theta; \Phi \vdash \rec\ z^{!A}. m: A$}
\end{bprooftree}
\qedhere
\]
\end{proof}


We now consider an example that highlights some of the advantages of the
ambivalence of our typing system.

\begin{exa}
\label{ex:factorial}
The factorial function on natural numbers may be defined by:

\noindent
\texttt{rec fact.} $\lambda$ \texttt{n.}\\
\texttt{
  case unfold n of \\
    \qquad left u --> succ zero \\
    \qquad right n' --> mult(n, (force fact) n')
\texttt}

\noindent
where \texttt{mult} is the multiplication function (which also can be easily 
defined). Note that copying and discarding of the non-linear variables is
implicit. In more traditional linear typing systems, the variables for
which contraction and weakening (copying and discarding) are admissible are
those of type $!A$. But, notice that the variable \texttt{n} above is of type
\texttt{Nat}, which is \emph{not} of the form $!A$. We are able to implicitly copy
it above, because we have extended the non-linear types to include more
types than just those of the form $!A$. This means that \texttt{lift} and \texttt{force} are
mostly needed for promotion and dereliction of terms of function types (as
can be seen above).
\end{exa}

\section{Operational Semantics}\label{sec:operational}

The operational semantics of LNL-FPC is standard.  We use a big-step
call-by-value evaluation relation whose rules are shown in
Figure~\ref{fig:operational}.

\begin{figure}
\[
\begin{bprooftree}
  \def\ScoreOverhang{0.5pt}
\AxiomC{{\color{white} $\Downarrow$}}
\UnaryInfC{$x \Downarrow x$}
\end{bprooftree}
\quad
\begin{bprooftree}
  \def\ScoreOverhang{0.5pt}
\AxiomC{$m \Downarrow v$}
\UnaryInfC{$\lleft\ m \Downarrow \lleft\ v$}
\end{bprooftree}
\quad
\begin{bprooftree}
  \def\ScoreOverhang{0.5pt}
\AxiomC{$m \Downarrow v$}
\UnaryInfC{$\rright\ m \Downarrow \rright\ v$}
\end{bprooftree}
\]

\[
\begin{bprooftree}
  \def\ScoreOverhang{0.5pt}
\AxiomC{$m \Downarrow \lleft\ v$}
\AxiomC{$n[v/x] \Downarrow w$}
\BinaryInfC{$\ccase\ m\ \texttt{of}\ \{\lleft\ x \to n\ |\ \rright\ y \to p\} \Downarrow w$}
\end{bprooftree}
\quad
\begin{bprooftree}
  \def\ScoreOverhang{0.5pt}
\AxiomC{$m \Downarrow v$}
\AxiomC{$n \Downarrow w$}
\BinaryInfC{$\langle m, n \rangle \Downarrow \langle v, w \rangle$}
\end{bprooftree}
\]

\[
\begin{bprooftree}
  \def\ScoreOverhang{0.5pt}
\AxiomC{$m \Downarrow \rright\ v$}
\AxiomC{$p[v/y] \Downarrow w$}
\BinaryInfC{$\ccase\ m\ \texttt{of}\ \{\lleft\ x \to n\ |\ \rright\ y \to p\} \Downarrow w$}
\end{bprooftree}
\quad
\begin{bprooftree}
  \def\ScoreOverhang{0.5pt}
\AxiomC{$m \Downarrow \langle v, v' \rangle$}
\AxiomC{$n[v/x, v'/y] \Downarrow w$}
\BinaryInfC{$\llet\ \langle x, y \rangle = m\ \text{in}\ n \Downarrow w$}
\end{bprooftree}
\]

\[
\begin{bprooftree}
  \def\ScoreOverhang{0.5pt}
  \AxiomC{{\color{white} $\Downarrow$}}
\UnaryInfC{$\lambda x.m \Downarrow \lambda x.m$}
\end{bprooftree}
\quad
\begin{bprooftree}
	\def\ScoreOverhang{0.5pt}
	\AxiomC{$m \Downarrow \lambda x.m'$}
	\AxiomC{$n \Downarrow v$}
	\AxiomC{$m'[v/x]\Downarrow w$}
	\TrinaryInfC{$mn \Downarrow w$}
\end{bprooftree}
\]

\[
\begin{bprooftree}
  \def\ScoreOverhang{0.5pt}
\AxiomC{{\color{white} $\Downarrow$}}
\UnaryInfC{$\lift\ m \Downarrow \lift\ m$}
\end{bprooftree}
\quad
\begin{bprooftree}
	\def\ScoreOverhang{0.5pt}
	\AxiomC{$m \Downarrow \lift\ m'$}
	\AxiomC{$m' \Downarrow v$}
	\BinaryInfC{$\force\ m \Downarrow v$}
\end{bprooftree}
\quad
\begin{bprooftree}
  \def\ScoreOverhang{0.5pt}
\AxiomC{$m \Downarrow v$}
\UnaryInfC{$\fold{}\ m \Downarrow \fold{}\ v$}
\end{bprooftree}
\quad
\begin{bprooftree}
  \def\ScoreOverhang{0.5pt}
\AxiomC{$m \Downarrow \fold{}\ v$}
\UnaryInfC{$\unfold{}\ m \Downarrow v$}
\end{bprooftree}
\]
\caption{Operational semantics of the \calc{} calculus.}\label{fig:operational}
\end{figure}

As usual, the \emph{values} are terms $v$ such that $v \Downarrow v$ (see Figure~\ref{fig:syntax}).
The evaluation relation $(- \Downarrow -)$ is, in fact, a \emph{partial} function from terms to values.
Of course, it is not a total function, because the language supports a general recursion operator, as we show next.

\begin{thm}\label{thm:operational-recursion}
The evaluation rule 
    \begin{bprooftree}
    \AxiomC{$m[\lift\ \rec\ z^{!A}. m\ /\ z] \Downarrow v$}
    \UnaryInfC{$\rec\ z^{!A}. m \Downarrow v$}
    \end{bprooftree}
is derivable within LNL-FPC, where $\rec\ z^{!A}. m$ is defined as in Theorem~\ref{thm:recursion operator}.
\end{thm}
\begin{proof}

First, we introduce the term
\[m'\equiv(\lambda z^{!A}.m)(\lift\ (\unfold\ \force\ x)x),\]
so that 
\[\alpha_m^z=\lift\ \fold\ \lambda x^{!\mu X. (!X \multimap A)}.m',\]
and
\[  m'[\alpha_m^z\ /\ x]=(\lambda z^{!A}.m)(\lift\ (\unfold\ \force\ \alpha_m^z)\alpha_m^z)=(\lambda z^{!A}.m)(\lift\ \rec\ z^{!A}.m).\]
Then
\[
\begin{bprooftree}
\AxiomC{}
\UnaryInfC{$\lambda z^{!A}.m\Downarrow \lambda z^{!A}.m$}
\AxiomC{}
\UnaryInfC{$\lift\ \rec\ z^{!A}.m\Downarrow \lift\ \rec\ z^{!A}.m$}
\AxiomC{$m[\lift\ \rec\ z^{!A}. m\ /\ z] \Downarrow v$}
\TrinaryInfC{$(\lambda z^{!A}.m)(\lift\ \rec\ z^{!A}.m)\Downarrow v$}
\end{bprooftree}
\]
hence
\begin{equation}\label{eq:oper-whatever}
\begin{bprooftree}
	\AxiomC{$m[\lift\ \rec\ z^{!A}. m\ /\ z] \Downarrow v$}
	\UnaryInfC{$m'[\alpha_m^z\ /\ x] \Downarrow v$}
\end{bprooftree}
\end{equation}
Moreover,
\[
\begin{bprooftree}
\AxiomC{}
\UnaryInfC{$\lift\ \fold\ \lambda x.m'\Downarrow \lift\ \fold\ \lambda x.m'$}
\AxiomC{}
\UnaryInfC{$\lambda x.m'\Downarrow \lambda x.m'$}	
\UnaryInfC{$\fold\ \lambda x.m'\Downarrow\fold\ \lambda x.m' $}
\BinaryInfC{$\force\ \lift\ \fold{}\ \lambda x.m'\Downarrow \fold\ \lambda x.m'$}
\UnaryInfC{$\unfold\ \force\ \lift\ \fold\ \lambda x.m' \Downarrow \lambda x. m'$}
\end{bprooftree}
\]
hence we obtain
\begin{equation}\label{eq:oper-again}
\begin{bprooftree}
\AxiomC{}
\UnaryInfC{$\unfold\ \force\ \alpha_m^z\Downarrow\lambda x.m'$}
\end{bprooftree}
\end{equation}
Then
\[
\begin{bprooftree}
  \AxiomC{}
  \RightLabel{\eqref{eq:oper-again}}
	\UnaryInfC{$\unfold\ \force\ \alpha_m^z\Downarrow\lambda x.m'$}
  \AxiomC{}
	\UnaryInfC{$\alpha_m^z\Downarrow\alpha_m^z$}
  \RightLabel{\eqref{eq:oper-whatever}}
		\AxiomC{$m[\lift\ \rec\ z^{!A}. m\ /\ z] \Downarrow v$}
	\UnaryInfC{$m'[\alpha_m^z\ /\ x] \Downarrow v$}
	\TrinaryInfC{ $(\unfold\ \force\ \alpha_m^z)\alpha_m^z\Downarrow v$,}
	\end{bprooftree}
  \]
and therefore:
\[
   \begin{bprooftree}
	\AxiomC{$m[\lift\ \rec\ z^{!A}. m\ /\ z] \Downarrow v$}
	\UnaryInfC{$\rec\ z^{!A}. m \Downarrow v$}
\end{bprooftree}
\qedhere
\]
\end{proof}

This theorem shows that our \emph{derived} general recursion operator is exactly the same as the one in~\cite{eclnl},
where it was added as an axiom.

\begin{nota}
A term $m$ is said to \emph{terminate}, denoted by $m \Downarrow$, if there exists
a value $v$, such that $m \Downarrow v.$ 
\end{nota}
The simplest
non-terminating program of type $A$ is given by $\rec\ z^{!A}.\ \force\ z : A$.

\begin{exa}
\label{ex:streams}
The constant stream of zero natural numbers can be defined by
\[ \texttt{const}_0 \equiv \rec\ s^{!(\Stream\ \Nat)} \fold_{\Stream\ \Nat}\ \langle \texttt{zero}, s \rangle : \Stream\ \Nat. \]
\end{exa}

\begin{rem}
Streams of type $P$ should not be defined as $\mu X. P \otimes X$, because there are no \emph{closed values} of this type. This
is a consequence of Theorem \ref{thm:adequacy} and the fact that $\lrb{\mu X. P \otimes X} = 0$ (see \secref{sec:semantics}).
\end{rem}


\begin{thm}[Subject reduction]\label{thm:subject-reduction}
If $\Theta; \Gamma \vdash m :A$ and $m \Downarrow v,$ then
$\Theta; \Gamma \vdash v :A$.
\end{thm}

\begin{asm}
Throughout the remainder of the paper, we implicitly assume that all types, contexts and terms are well-formed.
\end{asm}

\section{$\omega$-categories and (Parameterised) Initial Algebras}\label{sec:w-categories}

In this section we recall the theory of $\omega$-categories introduced
in~\cite{lehman-smyth} and develop new results of our own.
In~\secref{sub:notation}, we introduce some notation for operations on
natural transformations that we use throughout the paper.
In~\secref{sub:initial}, we recall how initial algebras are constructed in $\omega$-categories.
In~\secref{sub:par-initial}, we show how to construct parameterised initial algebras in $\omega$-categories that we use to
model recursive types that may potentially be defined by nested recursion.
In~\secref{sub:relating}, we present new results that show how
(parameterised) initial algebras of functors acting on different categories may
be related to one another, provided there exist suitable mediating functors
between the two categories. The type-level semantics makes heavy use of this
relationship in order to present \emph{coherent} non-linear
type interpretations that are strongly related to the standard type interpretations.

\subsection{Operations on natural transformations}\label{sub:notation}

Given natural transformations $\tau$ and $\sigma$ and a functor $F$, we denote:
vertical composition by $\sigma \circ \tau$; horizontal composition by $\sigma
* \tau$; whiskering by $F \tau$ and $\tau F$, whenever these
operations are admissible.

If ${\tau:T \naturalto H : \CC \to \CC}$ is a natural transformation between two endofunctors, then we define $\tau^{*n} \coloneqq
\tau * \cdots *\tau : T^n \naturalto H^n : \CC \to \CC$ to be the $n$-fold horizontal composition of $\tau$ with itself (if $n=0$, then $\tau^{*0}$ is the identity  $\id : \Id \naturalto \Id : \CC \to \CC$).

If ${\tau: T \naturalto T': \AAA \to \BB}$ and ${\sigma: H \naturalto H': \CC \to \DD}$ are natural transformations, we define
a natural transformation ${\tau \times \sigma : T \times H \naturalto T' \times H' : \AAA \times \CC \to \BB \times \DD}$
via the assignment ${(\tau \times \sigma)_{(A,C)} \coloneqq (\tau_A, \sigma_C).}$

If $\tau: T \naturalto T' : \AAA \to \BB$ and $\sigma: H \naturalto H': \AAA \to \CC$ are natural transformations, we define a natural transformation
${\langle \tau, \sigma \rangle : \langle T, H \rangle \naturalto \langle T', H' \rangle : \AAA \to \BB \times \CC}$ by
${\langle \tau, \sigma \rangle_A \coloneqq (\tau_A, \sigma_A).}$

We denote with $\omega$ the poset of natural numbers when viewed as a category. A functor $D : \omega \to \CC$ is then an $\omega$-diagram.
We let $[\AAA, \BB]$ denote the functor category from $\AAA$ to $\BB$.
Given a functor $M: \AAA \to \BB$, we define a functor ${M \triangleright - : [\omega, \AAA] \to [\omega, \BB]}$ by:
\[
M \tr D = M \circ D\qquad\text{and}\qquad
M \tr \tau = M \tau.
\]
So, the functor $M \tr -$ is just whiskering with $M$. 

\subsection{Initial algebras in $\omega$-categories}\label{sub:initial}

We now recall some definitions and facts about $\omega$-categories and $\omega$-functors that are stated in~\cite{lehman-smyth}.
A functor $F: \AAA \to \CC$ is an \emph{$\omega$-functor} if $F$ preserves all existing colimits of $\omega$-diagrams.
If, in addition, $\AAA$ has an initial object and $F$ preserves it, then we say that $F$ is a \emph{strict} $\omega$-functor.
Of course, $\omega$-functors are closed under composition and pairing, that is, if $F$ and $G$ are $\omega$-functors, then
so are $F \circ G$ and $\langle F, G \rangle$ whenever composition and pairing are admissible.

A category $\CC$ is an \emph{$\omega$-category} if it has an initial object and all $\omega$-colimits.
We denote with $[\AAA \to_\omega \CC]$ the full subcategory of $[\AAA,\CC]$ consisting of $\omega$-functors.
If $\CC$ is an $\omega$-category, then so are $[\AAA, \CC]$ and $[\AAA \to_\omega \CC]$.

\begin{defi}[\cite{lehman-smyth}]
Given an $\omega$-category $\CC$ with an initial object $\varnothing$, we define a functor $S: [\CC \to_\omega \CC] \to [\omega, \CC]$ in the following way:
\begin{itemize}
\item Given an $\omega$-functor $T: \CC \to \CC,$ then $S(T)$ is the $\omega$-diagram $\varnothing \xrightarrow{\iota} T\varnothing \xrightarrow{T\iota} T^2\varnothing \xrightarrow{T^2\iota} \cdots$.
More specifically, $S(T)(n) \coloneqq T^{n}\varnothing$ and $S(T)(n \leq n+1) \coloneqq T^n\iota$, where $\iota : \varnothing \to T\varnothing$ is the initial map.
\item Given $\tau: T \naturalto H:\CC\to \CC$, then $S(\tau) : S(T) \naturalto S(H) : \omega \to \CC$ is given by
\[ S(\tau)_n \coloneqq (\tau^{*n})_\varnothing: T^n\varnothing \to H^n\varnothing. \]
\end{itemize}
We also define a functor ${Y \coloneqq \colim \circ S : [\CC \to_\omega \CC] \to \CC},$ where
$\colim: [\omega, \CC] \to \CC$ is the colimiting
functor, which is the left-adjoint of the $\omega$-ary diagonal functor $\Delta: \CC \to [\omega, \CC]$.
\end{defi}

\begin{thm}\label{thm:Y is omega functor}\cite{lehman-smyth}
Let $\CC$ be an arbitrary $\omega$-category. Then both ${S: [\CC \to_\omega \CC] \to [\omega, \CC]}$ and ${Y: [\CC \to_\omega \CC] \to \CC}$ are $\omega$-functors.
\end{thm}

Therefore, for an $\omega$-functor $T: \CC \to \CC$ on an $\omega$-category
$\CC$, its initial sequence is given by $S(T)$ and the \emph{carrier}
of its initial algebra is given by $Y(T)$, thanks to a famous result
in~\cite{adamek-original}.
To describe its initial algebra structure we need an additional definition.


\begin{thm}[\cite{fixed-points-of-functors}] \label{thm:initial-algebra-simple}
Let $T: \CC \to \CC$ be an $\omega$-endofunctor on an $\omega$-category $\CC$.
We define a natural transformation
$s^T : S(T) \naturalto T \circ S(T) : \omega \to \CC$ by $(s^T)_n \coloneqq T^n(\iota_{T\varnothing})$,
where
${\iota_{T\varnothing}:\varnothing \to T\varnothing}$
is the initial morphism.
Then one can define an isomorphism:
\begin{align*}
y^T &: T(Y(T)) \to Y(T) \\
y^T &\coloneqq \left( TY(T) = T \colim(S(T)) = \colim(TS(T)) \xrightarrow{(\colim(s^T))^{-1}} \colim(S(T)) = Y(T) \right)
\end{align*}
and the pair $(Y(T), y^T)$ forms the initial $T$-algebra.
\end{thm}

\subsection{Parameterised initial algebras}\label{sub:par-initial}

Initial algebras can be used to model recursive types where the type recursion
is done over a single type variable. In order to model recursive types defined
by \emph{nested recursion}, one has to allow recursive types to
depend on several type variables. The interpretation of these more general recursive types requires
a more general notion, namely \emph{parameterised initial algebras}, which we introduce next.

\begin{defi}[cf. {\cite[\S 6.1]{fiore-thesis}}]
\label{def:param-initial}
  Given categories $\AAA$ and $\BB$ and a functor $T : \AAA \times \BB \to \BB,$ a \emph{parameterised initial algebra}
  for $T$ is a pair $(T^\dagger, \phi^T),$ such that:
  \begin{itemize}
    \item $T^\dagger : \AAA \to \BB$ is a functor;
    \item $\phi^T : T \circ \langle \Id, T^\dagger \rangle \naturalto T^\dagger : \AAA \to \BB$ is a natural isomorphism;
    \item For every $A \in \Ob(\AAA)$, the pair $(T^\dagger A, \phi^T_A)$ is an initial $T(A, -)$-algebra.
  \end{itemize}
    
\end{defi}

\begin{rem}
Notice that by trivialising the category $\AAA$, we recover the usual notion of initial algebra.
Because of this, parameterised initial algebras are a more general notion.
\end{rem}

\begin{rem}
The naturality of $\phi$ \eqref{eq:par-naturality} in fact determines the action of $T^\dagger$ on morphisms. Indeed,
for every $f : A_1 \to A_2$ in $\AAA$, $T^\dagger f$ is the unique $T(A_1,-)$-algebra morphism making \eqref{eq:par-initial} commute. To see this, notice that \eqref{eq:par-naturality} and \eqref{eq:par-initial} are equivalent diagrams.
\begin{equation}
\label{eq:par-naturality}
\stikz{parameterised-initial-naturality.tikz}
\end{equation}
\begin{equation}
\label{eq:par-initial}
\stikz{parameterised-initial-algebra.tikz}
\end{equation}
\end{rem}


Next, we aim to show that the class of $\omega$-functors on an $\omega$-category is closed under formation of parameterised initial algebras
(Theorem~\ref{thm:dagger is omega-continuous}). In order to do so, we have to establish a few additional lemmas.

\begin{lem}\label{lem:s-naturality}
Let $\BB$ be an $\omega$-category and let $T: \AAA \times \BB \to \BB$ be an $\omega$-functor. The mapping
\[  S(T(A, -)) \xRightarrow{s^{T(A,-)}} T(A, -) \circ  S(T(A,-)) : \omega \to \BB \]
is natural in $A.$ More specifically, the following diagram:
\cstikz{natural-s.tikz}
commutes for any $f: A \to B$ in $\AAA$.
\end{lem}
\begin{proof}
In Appendix~\ref{proof:s-naturality}.
\end{proof}

\begin{nota}
\label{not:w-diagrams}
Given an $\omega$-diagram $D:\omega\to\AAA$,
we denote objects $D(n)$ by $D_n$, the colimit of $D$ (if it exists) by $D_\omega$ and its colimiting cocone morphisms by $d_{n}:D_n\to D_\omega$. 
\end{nota}

We proceed with two simple lemmas that show that the $\colim: [\omega, \CC] \to \CC$ functor is quite well-behaved on $\omega$-categories and $\omega$-functors, as one would expect.

\begin{lem}\label{lem:whiskering colimit}
Let $T:\AAA\to\BB$ be an $\omega$-functor between $\omega$-categories $\AAA$ and $\BB$. Assume further $D,D':\omega\to\AAA$ are $\omega$-diagrams and $\tau:D\naturalto D'$ a natural transformation. Then the following diagram commutes: 
\cstikz{colimitofwhisker.tikz}
\end{lem}
\begin{proof}
In Appendix~\ref{proof:whiskering colimit}.
\end{proof}

\begin{lem}\label{lem:whiskering colimit2}
Let $T, H:\AAA\to\BB$ be $\omega$-functors between $\omega$-categories $\AAA$ and $\BB$, and let $\tau:T\naturalto H$ be a natural transformation. Given an $\omega$-diagram $D:\omega\to\AAA$, the following diagram commutes: 
	\cstikz{colimitofwhisker2.tikz}
\end{lem}
\begin{proof}
In Appendix~\ref{proof:whiskering colimit2}.
\end{proof}	

We now recall an important lemma that shows that $\omega$-functors on $\omega$-categories are closed under currying.

\begin{lem}[\cite{lehman-smyth}]
Given an $\omega$-category $\CC$ and an $\omega$-functor $T: \AAA \times \BB \to \CC$,
we define a functor ${\lambda B. T(-,B) : \AAA \to [\BB \to_\omega \CC]}$ by:
\begin{align*}
\lambda B.T(A, B)  &\coloneqq T(A,-) : \BB \to \CC,                   &&\text{ for } A \in \mathrm{Ob}(\AAA);\\
\lambda B.T(f, B)  &\coloneqq T(f, -): T(A_1,-) \naturalto T(A_2,-) : \BB \to \CC,  &&\text{ for } f: A_1 \to A_2.
\end{align*}
Then ${\lambda B. T(-,B) : \AAA \to [\BB \to_\omega \CC]}$ is an $\omega$-functor.
\end{lem}

The main result from this subsection follows.

\begin{thm}\label{thm:dagger is omega-continuous}
Let $\BB$ be an $\omega$-category and let $T: \AAA \times \BB \to \BB$ be an $\omega$-functor. Then:
\begin{enumerate}
\item The functor $ T^{\dagger} \coloneqq Y \circ \lambda B.T(-,B) : \AAA \to \BB $ is an $\omega$-functor;
\item There exists a natural isomorphism $\phi^T :  T \circ \langle \Id, T^\dagger \rangle \naturalto T^\dagger : \AAA \to \BBB$ given by:
\[ \phi^T_A \coloneqq \left( T(A, T^\dagger A) =  T(A, Y(T(A,-))) \xrightarrow{y^{T(A,-)}} Y(T(A,-)) = T^\dagger A \right); \]
\item The pair $(T^\dagger, \phi^T)$ is a parameterised initial algebra of $T$.
\end{enumerate}
\end{thm}
\begin{proof}
\hfill%

(1) Because $T^\dagger$ is the composition of two $\omega$-functors.

(2) Every component of $\phi^T$ is an isomorphism by Theorem~\ref{thm:initial-algebra-simple}.
It remains to show naturality. So let $f:A\to B$ be a morphism in $\AAA$. For brevity we write $T_A, T_B, T_f$ for $T(A,-), T(B, -), T(f,-),$ respectively.
 By definition
{%
\small{%
\[ (\phi^T_A)^{-1} = \left( T^\dagger A = \colim(S(T_A)) \xrightarrow{\colim (s^{T(A, -)})} \colim (T_AS(T_A)) = T_A(\colim (S(T_A))) = T(A, T^\dagger A) \right) \] 
}%
}%
and naturality of $\phi^T$ comes down to the commutativity of the outer square of the diagram:
\[ \stikz[0.95]{alpha-dagger-natural2.tikz} \]
This is indeed the case, since squares (1) and (4) commute by definition of $T^\dagger$, and diagram (2) commutes by the functoriality of $\colim$ and Lemma \ref{lem:s-naturality}.
By definition of horizontal composition, 
we have 
$ T_f*S(T_f)=T_fS(T_B)\circ T_A S(T_f). $
Since $\colim$ is a functor, it follows from the latter equality that the commutativity of square (3) comes down to the commutativity of the outer square of the following diagram:
\cstikz{colimitofwhisker3.tikz}
Here the upper square commutes by Lemma \ref{lem:whiskering colimit} and the lower one by Lemma \ref{lem:whiskering colimit2}.

(3) For every $A \in \Ob(\AAA)$, we have $(T^\dagger A, \phi^T_A) = (Y(T_A), y^{T_A})$, which is the initial $T(A,-)$-algebra by Theorem~\ref{thm:initial-algebra-simple}.
\end{proof}

\begin{rem}
In the situation of the above theorem, we sometimes simply write $\phi$ instead of $\phi^T$ whenever $T$ is clear from the context.
Moreover, in the special case when $\AAA=\BB^n$, we see that both $T: \BB^{n+1} \to \BB$ and $T^\dagger : \BB^n \to \BB$ are $\omega$-functors.
In fact, in our semantic treatment, the interpretation of a type
$\Theta \vdash A$ is given by an $\omega$-functor that is of the form $H: \BB^{|\Theta|} \to \BB$.
\end{rem}

We conclude the subsection by showing an important proposition for proving the type substitution lemma.

\begin{prop}\label{prop:dagger-equal}
Let $\CC$ be an $\omega$-category, $T: \BB \times \CC \to \CC$ an $\omega$-functor and $H: \AAA \to \BB$ an $\omega$-functor.
Then $(T \circ (H \times \Id) )^\dagger = T^\dagger \circ H : \AAA \to \CC.$
\end{prop}
\begin{proof}
For any object $A \in \Ob(\AAA)$ and
for any morphism $f \in \Mor(\AAA)$, we have:
\begin{align*}
( T \circ (H \times \Id))^\dagger(A) &= Y(T(HA, -)) \\
&= (Y \circ \lambda B. T(-, B) )(HA) \\
&= T^\dagger(HA) \\
&= (T^\dagger \circ H)(A)
\end{align*}
\begin{align*}
( T \circ (H \times \Id))^\dagger(f) &= Y(T(Hf, -)) \\
&= (Y \circ \lambda B. T(-, B) )(Hf) \\
&= T^\dagger(Hf) \\
&= (T^\dagger \circ H)(f). \qedhere
\end{align*}
\end{proof}

\subsection{Coherence properties of (parameterised) initial algebras}\label{sub:relating}

We model (arbitrary) recursive types within a linear category that is
algebraically compact in a strong sense and that allows us to model recursive
types involving the $\multimap$ connective. The non-linear recursive types
form a subset of our types for which we have to, in addition, provide
categorical structure that allows for them to be copied, discarded and promoted
(the substructural rules of intuitionistic linear logic).  We do this by
providing a non-linear interpretation of these types within a
cartesian category. The two interpretations live in different categories,
but they are strongly related to each other via suitable mediating functors and
a natural isomorphism. In order to show this, we first explain how
parameterised initial algebras can be related to each other in such a strong
sense. The current subsection is devoted to this.

\begin{nota}
For a functor $F: \AAA \times \BB \to \BB$, we define a functor
\[ F^* \coloneqq S \circ \lambda B. F(-,B) : \AAA \to [\omega, \BB] , \]
so that $F^\dagger = \colim \circ F^*.$
\end{nota}

\begin{asm}\label{ass:parameterised-square}
Throughout the remainder of the section, we assume we are given the following data.
Let $\AAA$ and $\CC$ be categories and let $\BB$ and
$\DD$ be $\omega$-categories with initial objects $\varnothing$ and $0$, respectively.
Let $\alpha: T \circ (N \times M) \naturalto M \circ H$ be a natural isomorphism, as in: 
\cstikz{bekic-sequence.tikz}
where $H$ and $T$ are $\omega$-functors and where $M$ is a strict $\omega$-functor with
$z: 0 \to M \varnothing$ the required (unique) isomorphism. The functor $N$ need not be an $\omega$-functor.
\end{asm}

\begin{lem}\label{lem:alpha-star}
The assignment
{
\small{
\begin{align*}
 \alpha^* &: T^* \circ N \naturalto (M  \tr -) \circ H^* : \AAA \to [\omega, \DD] \quad \text{given by} \\
 (\alpha^*_A)_0     &\coloneqq \left( 0 \xrightarrow z M \varnothing \right) \\
 (\alpha^*_A)_{n+1} &\coloneqq \left( T(NA, -)^{n+1} 0 \xrightarrow{T(NA, (\alpha_A^*)_n)} T(NA, MH(A, -)^n \varnothing) \xrightarrow{\alpha_{A, H(A,-)^n \varnothing}} MH(A,-)^{n+1} \varnothing \right) .
\end{align*}
}
}
defines a natural isomorphism.
\end{lem}
\begin{proof}
In Appendix~\ref{proof:alpha-star}.
\end{proof}

The next theorem shows how to extend the action of $(-)^\dagger$ to natural transformations.

\begin{thm}\label{thm:alpha-dagger-def}
The natural isomorphism $\alpha$ induces a natural isomorphism
\begin{align*}
\alpha^\dagger &: T^\dagger \circ N \naturalto M \circ H^\dagger : \AAA \to \DD \quad \text{defined by} \\
\alpha^\dagger_A &\coloneqq \left( T^\dagger NA = \colim(T^*NA)  \xrightarrow{\colim(\alpha^*_A)} \colim(MH^*A) = M \colim(H^*A) = MH^\dagger A \right).
\end{align*}
\end{thm}
\begin{proof}
By Lemma~\ref{lem:alpha-star}, $\alpha^*_A$ is an isomorphism and thus so is $\alpha^\dagger_A$. Naturality follows from:
\[ \stikz{alpha-dagger-natural.tikz} \]
where (1) and (4) commute by definition, (2) commutes by naturality of $\alpha^*$ and functoriality of $\colim$, and (3) commutes because
$M$ preserves $\omega$-colimits (more specifically, Lemma~\ref{lem:whiskering colimit}).
\end{proof}

\begin{rem}
The above construction generalises the operation $(-)^\dagger$ from \cite[
Corollary 7.3.13]{fiore-thesis} to the context of $\omega$-categories.
\end{rem}

\begin{cor}\label{cor:alpha-n}
If $\AAA = \BB^n$, $\CC = \DD^n$, $N = M^{\times n}$ and
$ \alpha: T \circ M^{\times n+1} \naturalto M \circ H: \BB^{n+1} \to \DD $
is a natural isomorphism, then so is $ \alpha^\dagger: T^\dagger \circ M^{\times n} \naturalto M \circ H^\dagger : \BB^n \to \DD. $
\end{cor}

By reading off the proof of Theorem~\ref{thm:alpha-dagger-def}, we obtain another corollary that we use in our adequacy proof.
\begin{cor}\label{cor:alpha-id}
In the special case where $\alpha = id$ and $z = \id$, then $\alpha^\dagger = \id$ and thus $T^\dagger \circ N = M \circ H^\dagger.$
\end{cor}

We need two lemmas that establish some properties of the $(-)^*$ and $(-)^\dagger$ operations.

\begin{lem}\label{lem:alpha-dagger-def}
The operation $(-)^*$ defined in Lemma~\ref{lem:alpha-star} enjoys the following properties:
\begin{enumerate}
\item If $\beta : H \naturalto H' : \AAA \times \BBB \to \BBB$ is a natural isomorphism, then for any $A \in \AAA$ :
\[(M\beta \circ \alpha)^*_A = M\beta^*_A \circ \alpha^*_A . \] 
\item Given functors $P :\EE \to \CC$ and $Q: \EE \to \AAA$
and given a natural isomorphism \\ ${ \gamma: P \naturalto N \circ Q : \EE \to \CC }$, then for any $E \in \EE $ :
\[(\alpha(Q\times \Id) \circ T(\gamma \times M))^*_E = \alpha^*_{QE} \circ T^* \gamma_E . \]
\end{enumerate}
\end{lem}
\begin{proof}
In Appendix~\ref{proof:alpha-dagger-def}.
\end{proof}

\begin{lem}\label{lem:alpha-operations}
The operation $(-)^\dagger$ enjoys the following properties:
\begin{enumerate}
\item \label{list:beta1} If $\beta : H \naturalto H' : \AAA \times \BBB \to \BBB$ is a natural isomorphism, then
$ (M\beta \circ \alpha)^\dagger = M\beta^\dagger \circ \alpha^\dagger . $
\item Given functors $P :\EE \to \CC$ and $Q: \EE \to \AAA$
and given a natural isomorphism \\ ${ \gamma: P \naturalto N \circ Q : \EE \to \CC }$, then
$ (\alpha(Q\times \Id) \circ T(\gamma \times M))^\dagger = \alpha^\dagger Q \circ T^\dagger \gamma . $
\end{enumerate}
\end{lem}
\begin{proof}
(1) We have:
\begin{align*}
(M\beta \circ \alpha)^\dagger_A &= \colim((M\beta \circ \alpha)^*_A) & \text{(Definition)}\\
& = \colim (M\beta^*_A \circ \alpha^*_A) & \text{(Lemma~\ref{lem:alpha-dagger-def})}\\
& = \colim (M\beta^*_A) \circ \colim(\alpha^*_A) &\text{(Functoriality)}\\
& = M\colim (\beta^*_A) \circ \colim(\alpha^*_A) &\text{(Lemma~\ref{lem:whiskering colimit})}\\
& = M\beta^\dagger_A \circ \alpha^\dagger_A &\text{(Definition)}\\
& = (M\beta^\dagger \circ \alpha^\dagger)_A &
\end{align*}

(2) We have:
\begin{align*}
(\alpha(Q\times \Id) \circ T(\gamma \times M))^\dagger_E &= \colim((\alpha(Q\times \Id) \circ T(\gamma \times M))^*_E) &(\text{Definition}) \\
&= \colim(\alpha^*_{QE} \circ T^*\gamma_E) & (\text{Lemma~\ref{lem:alpha-dagger-def}}) \\
&= \colim(\alpha^*_{QE}) \circ \colim(T^*\gamma_E) &(\text{Functoriality}) \\
&= \alpha^\dagger_{QE} \circ T^\dagger\gamma_E &(\text{Definition}) \\
&= (\alpha^\dagger Q \circ T^\dagger\gamma )_E &
\qedhere
\end{align*}
\end{proof}

\begin{rem}\label{rem:dagger}
As a special case, if one takes $N$ and $M$ to be
identity functors, then given a natural isomorphism ${ \alpha: T \naturalto H:
\AAA \times \BBB \to \BBB }$, it follows ${ \alpha^\dagger: T^\dagger \naturalto
H^\dagger : \AAA \to \BBB }$ is also a natural isomorphism and Lemma~\ref{lem:alpha-operations} (\ref{list:beta1}) shows the operation $(-)^\dagger$ is functorial.
\end{rem}

We need one more lemma before we may prove the main coherence property.

\begin{lem}\label{lem:pentagon-sequence}
In the special case where categories $\AAA$, $\CC$ and the functor $N$ are trivial in Assumption~\ref{ass:parameterised-square}, 
the following diagram of natural transformations:
\[ \stikz{initial-sequence-extravaganza.tikz} \]
commutes, where we regard $H$ and $T$ as functors $H: \BB \to \BB$ and $T:\DD \to \DD$ and we regard $\alpha$ as a natural isomorphism $\alpha: T \circ M \naturalto M \circ H.$
\end{lem}
\begin{proof}
In Appendix~\ref{proof:pentagon-sequence}.
\end{proof}

The main result from this section follows. It shows how parameterised initial algebras may be related in the situation of Assumption~\ref{ass:parameterised-square}. We will later use this theorem in order to show our semantics is sound.

\begin{thm}\label{thm:fixpoint-deconstruction}
The following diagram of natural isomorphisms:
\cstikz{fixpoint-deconstruction.tikz}
commutes.
\end{thm}
\begin{proof}
Fix an object $A \in \AAA$. We have to show:
\cstikz{fixpoint-specific.tikz}
where ${T' := T(NA,-) : \DD \to \DD}$ and ${H' := H(A,-) : \BBB \to \BBB.}$ 
Note that we then have a natural isomorphism ${\alpha'  : T' \circ M \naturalto M \circ H'}$
given by $\alpha' := \alpha_{A,-}.$
Using Lemma~\ref{lem:pentagon-sequence}, we get:
\begin{equation}
\colim(Ms^{H'}) \circ \colim(\alpha^*_A) = \colim(\alpha' * S(H')) \circ \colim(T' \alpha^*_A) \circ \colim(s^{T'})
\end{equation}
Using Lemma~\ref{lem:whiskering colimit2} and the fact that $M$ and $T$ preserve $\omega$-colimits, we get:
\begin{equation}
M\colim(s^{H'}) \circ \colim(\alpha^*_A) = \alpha_{A,Y(H')} \circ T' \colim(\alpha^*_A) \circ \colim(s^{T'})
\end{equation}
Since $\colim(s^{H'})$ and $\colim(s^{T'})$ are isomorphisms, the above is equivalent to:
\begin{equation}
\colim(\alpha^*_A)  \circ \colim(s^{T'})^{-1} = M\colim(s^{H'})^{-1} \circ \alpha_{A,Y(H')} \circ T' \colim(\alpha^*_A)  
\end{equation}
Finally, by definition of $\alpha^\dagger$ and $y$, we get:
$
\alpha^\dagger_A  \circ y^{T'} = My^{H'} \circ \alpha_{A,Y(H')} \circ T' \alpha^\dagger_A .
$
\end{proof}

By combining this theorem together with Corollary~\ref{cor:alpha-id}, we obtain another corollary that is important for the adequacy proof.
\begin{cor}\label{cor:phi-adequacy}
In the special case where $\alpha = \id$ and $z = \id$, the 2-categorical diagram
\cstikz{2-categorical-cor.tikz}
commutes, i.e. $\phi^T N = M \phi^H.$
\end{cor}

We also see that $\phi$ is natural in the choice of functor.

\begin{cor}
In the special case where $N = M = \Id$, then $\alpha^\dagger \circ \phi^T = \phi^H \circ \alpha \langle \Id, \alpha^\dagger \rangle,$
i.e., $\phi^X$ is natural in $X$.
\end{cor}

\section{Categorical Model}\label{sec:categorical-model}

In this section we introduce our categorical model for LNL-FPC and describe its
categorical properties. A $\cpo$-LNL model is a $\cpo$-enriched model of
intuitionistic linear logic (also known as linear/non-linear model
\cite{benton-small,benton-wadler}) that has suitable $\omega$-colimits.

\subsection{$\cpo$-categories and Algebraic Compactness}\label{sub:cpo}

A \emph{cpo} (complete partial order) is a poset that has suprema of all
increasing chains. If, in addition, the cpo has a least element, then we say
it is \emph{pointed}.  A monotone function between two cpo's is
\emph{Scott continuous} if it preserves all suprema of increasing chains.
If, in addition, the two cpo's are pointed and the function preserves the least
element, then we say the function is \emph{strict}.

We denote with $\cpo$ the category of cpo's and Scott continuous functions
and with $\cpobs$ the category of pointed cpo's and strict Scott continuous
functions. Both categories are complete and cocomplete, $\cpo$ is
cartesian closed and $\cpobs$ is symmetric monoidal closed when equipped
with the smash product and the strict function space~\cite{abramskyjung:domaintheory}.

Our categorical model makes heavy use of \emph{enriched} category theory. In particular,
we use $\cpo$-enriched and $\cpobs$-enriched categories. We only provide
a brief introduction to $\cpo$-categories and we suggest the reader
consults~\cite[Chapter 2]{fiore-thesis} for a more detailed introduction.

A $\cpo$-category $\CC$ is a category $\CC$ whose homsets have the additional structure of a cpo and
composition of morphisms is a Scott continuous function (in both arguments).
A $\cpo$-functor $F:\CC \to \DD$ between
$\cpo$-categories is a functor whose action on hom-cpo's ${F_{X,Y}:\CC(X,Y)\to\DD(FX,FY)}$
is Scott continuous for each pair of objects $X,Y\in\CC$.
A $\cpo$-adjunction $\Phi : F \dashv G : \DD \to \CC$ is given by $\cpo$-categories $\CC$ and $\DD$,
$\cpo$-functors $F : \CC \to \DD$ and $G: \DD \to \CC,$ together with a natural isomorphism
${\Phi : \CC(-, G-) \cong \DD(F-, -) : \CC^\op \times \DD \to \cpo.}$ 
A $\cpo$-symmetric monoidal closed category $\CC(I, \otimes, \multimap)$ is a $\cpo$-category $\CC$ 
together with  $\cpo$-functors $\otimes: \CC \times \CC \to \CC$ and $\multimap: \CC^\op \times \CC \to \CC$ 
that form a symmetric monoidal closed
category in the usual sense (one can then conclude Currying is a $\cpo$-adjunction).
Similarly, $\cpobs$-enriched versions of these notions may be defined by
requiring the hom-cpo's to be pointed and Scott continuity to be strict.

\begin{defi}
In a $\cpo$-category $\CC$, a morphism $e: A\to B$ is called an
\emph{embedding} if there is a morphism $p:B\to A$, called a \emph{projection},
such that $p\circ e=\id_A$ and $e\circ p\leq \id_B$. 
The pair $(e,p)$ is said to form an \emph{embedding-projection (e-p) pair}.
\end{defi}

It is well-known that one component of an e-p pair uniquely determines the other, so we shall
write $e^\bullet$ to indicate the projection counterpart of an embedding $e$. 
Every isomorphism is an embedding and
embeddings are closed under composition with
$(e_1\circ e_2)^\bullet=e_2^\bullet\circ e_1^\bullet$,
which allows us to make the next definition.

\begin{defi}
	Given a $\cpo$-category $\CC$, its \emph{subcategory of embeddings}, denoted $\CC_e,$ is the full-on-objects subcategory of $\CC$ whose morphisms are exactly the embeddings of $\CC$.
\end{defi}

By duality, one can also define $\CC_p$ to be the subcategory of projections.
The notion of an e-p pair is of fundamental importance for establishing the famous limit-colimit coincidence theorem~\cite{smyth-plotkin:domain-equations}.
Before we can state it, we need to define a few additional notions.

\begin{defi}
In a $\cpo$-category $\CC$, an \emph{$\omega$-diagram over embeddings} is an $\omega$-diagram $D : \omega \to \CC$,
such that all $D(i \leq j)$ are embeddings. Given such a diagram, we define $D^\bullet : \omega^\op \to \CC$,
to be the $\omega^\op$-diagram given by $D^\bullet (n) \coloneqq D(n)$ and $D^\bullet (i \leq j) \coloneqq D(i \leq j)^\bullet.$
Given a cocone $\mu : D \to A$ in $\CC$, we define a cone $\mu^\bullet : A \to D^\bullet$ in $\CC$ by $(\mu^\bullet)_i \coloneqq (\mu_i)^\bullet$.
We shall say that $\CC$ has all
$\omega$-colimits over embeddings if every $\omega$-diagram over embeddings has a colimit in $\CC$.
\end{defi}

By combining several results from~\cite{smyth-plotkin:domain-equations}, one can show the next theorem.

\begin{thm}[Limit-colimit coincidence]\label{thm:locally-determined}
Let $\CC$ be a $\cpo$-category with all $\omega$-colimits over embeddings. Let $D: \omega \to \CC$ be an
$\omega$-diagram over embeddings and let $\mu: D \to A$ be a cocone of $D$. The following are equivalent:
\begin{enumerate}
  \item $\mu$ is a colimiting cocone of $D$ in $\CC$,
  \item $\mu$ is a colimiting cocone of $D$ in $\CC_e$,
  \item Each $\mu_i$ is an embedding, $(\mu_i \circ \mu_i^\bullet)_{i}$ is an increasing chain and $\bigvee_i \mu_i \circ \mu_i^\bullet = \id_A$,
  \item $\mu^\bullet$ is a limiting cone of $D^\bullet$ in $\CC$,
  \item $\mu^\bullet$ is a limiting cone of $D^\bullet$ in $\CC_p$. 
\end{enumerate}
\end{thm}
\begin{proof}
  In Appendix~\ref{proof:locally-determined}.
\end{proof}


\begin{rem}\label{rem:cpo-bad}
This theorem applies to both categories $\cpo$ and $\cpobs$. However, when
constructing an initial algebra for an endofunctor $T : \cpo \to \cpo$, the
theorem is not useful, because the initial map $\iota: \varnothing \to
T\varnothing$ is not an embedding (unless $T\varnothing = \varnothing$ ) and therefore the initial sequence
of $T$ is not necessarily computed over embeddings.
Also, observe that $\cpo_e$ has $\omega$-colimits, but has no initial object and thus it is not a good setting for constructing initial algebras.
\end{rem}

On the other hand, this theorem
is very useful for constructing initial (and final) (co)algebras for
many categories like $\cpobs$ that have some additional structure, which
we explain next.

\begin{defi}[{\cite[Definition 7.1.1]{fiore-thesis}}]
In a $\cpo$-category, an \emph{e-initial} object is an initial object such
that every morphism with it as source is an embedding. The dual notion is
called a \emph{p-terminal} object. An object that is both e-initial and
p-terminal is called an \emph{ep-zero} object.
\end{defi}

The category $\cpo$ has an initial object given by $\varnothing$, but it does not
have an e-initial object. The category $\cpobs$ has an ep-zero object that is
the one-element cpo $\{ \perp \}$.


An endofunctor $T: \CC \to \CC$ is \emph{algebraically compact} if $T$ has an initial $T$-algebra $(\Omega, \tau)$,
such that $(\Omega, \tau^{-1})$ is a final $T$-coalgebra. In this situation we say $(\Omega, \tau)$ is a \emph{free} $T$-algebra.
A $\cpo$-category $\CC$ is \emph{$\cpo$-algebraically compact} if every $\cpo$-endofunctor $T : \CC \to \CC$ has a free $T$-algebra.
$\cpo$-algebraic compactness also implies the existence of \emph{parameterised free algebras} \cite{fiore-thesis}.

\begin{thm}[{\cite[Corollary 7.2.4]{fiore-thesis}}]\label{thm:alg-compact}
A $\cpo$-category with ep-zero and colimits of $\omega$-chains over embeddings is $\cpo$-algebraically compact.
\end{thm}

Therefore, categories that enjoy this structure are an ideal setting for solving recursive domain equations.
In particular, algebraic compactness is crucial for interpreting isorecursive types involving $\multimap$.
We will see
that the linear category of our model is indeed such a category and that we can very easily interpret (arbitrary) types within it.
Non-linear types, however, need to admit an additional non-linear interpretation within a cartesian category that does not
enjoy such a strong property. To resolve this problem, our type interpretations
will be based on the following theorem, which combines results of
\cite{smyth-plotkin:domain-equations}.

\begin{thm}
\label{thm:e-p-trick}
  Let $\AAA, \BBB$ and $\CC$ be $\cpo$-categories where $\AAA$ and $\BB$ have 
  $\omega$-colimits over embeddings. If ${T: \AAA^\op \times \BBB \to \CC}$
  is a $\cpo$-functor, then the \emph{covariant} functor 
  \[ T_e : \AAA_e \times \BBB_e \to \CCe	\quad \text{given by} \quad T_e(A,B) = T(A,B) \quad \text{and} \quad T_e(e_1, e_2) = T((e_1^\bullet)^\op, e_2) \]
	is an $\omega$-functor. 
\end{thm}
\begin{proof}
This follows from \cite[Theorem 3]{smyth-plotkin:domain-equations}, because both categories $\AAA$ and $\BB$ have locally determined $\omega$-colimits of embeddings thanks to Theorem~\ref{thm:locally-determined}.
\end{proof}

Therefore, this theorem allows us to consider mixed-variance $\cpo$-functors on $\cpo$-categories as covariant ones on subcategories of embeddings.
By trivialising the category $\AAA$, we get: 
\begin{cor}\label{cor:cpo functors restrict to omega continuous functors}
Let $T:\BBB \to \CC$ be a $\cpo$-functor between $\cpo$-categories $\BBB$ and $\CC$ where $\BB$ has $\omega$-colimits over embeddings. Then $T$ restricts to an $\omega$-functor $T_e:\BBB_e\to\CC_e$.
\end{cor}

\subsection{Models of Intuitionistic Linear Logic}\label{sub:lnl}

Our type system has both linear and non-linear features, so naturally, it
should be interpreted within a model of Intuitionistic Linear Logic, also known as a linear/non-linear (LNL) model
(see~\cite{benton-wadler,benton-small} for more details).

\begin{defi}\label{def:lnl}
A \emph{model of Intuitionistic Linear Logic} is given by the following data:
a cartesian monoidal category with finite coproducts $(\CC, \vartimes, 1, \amalg, \varnothing)$;
a symmetric monoidal closed category with finite coproducts $(\LL, \otimes, \multimap, I, +, 0)$;
and a symmetric monoidal adjunction $\stikz{ill-model.tikz}$.
We also adopt the following notation. 
The comonad-endofunctor is $! \coloneqq F \circ G : \LL \to \LL$;
the unit and counit of the adjunction are $\eta : \text{Id} \naturalto G \circ F : \CC \to \CC $ and $\epsilon :\ !  \naturalto \text{Id} : \LL \to \LL,$ respectively.
\end{defi}

Moreover, in this situation, the left-adjoint $F$ is both cocontinuous and strong symmetric monoidal. Thus, there exist isomorphisms
$ I \xrightarrow u F(1) $ \text{and} $ 0 \xrightarrow z F(\varnothing) $ in $\LL$,
together with natural isomorphisms
\[
 \otimes \circ (F \times F) \xRightarrow{m\ } F \circ \vartimes : \CC \times \CC \to \LL
\qquad \text{and} \qquad
+ \circ\ (F \times F) \xRightarrow{c\ } F \circ \amalg : \CC \times \CC \to \LL
\]
that satisfy some coherence conditions, which we omit here. Observe, that the
structure of the natural isomorphisms $m$ and $c$ is very similar. In fact, most of our
propositions involving these natural isomorphisms can be proven in a uniform
way, so we introduce some notation for brevity.

\begin{nota}\label{not:odot}
We write $\odot$ for either the
$\otimes$ or $+$ bifunctor; $\boxdot$ for the bifunctor $\vartimes$ or $\amalg$
respectively; $\beta$ for the natural isomorphisms $m$ or $c$, respectively.
Then, we have a natural isomorphism:
\[ 
\odot \circ (F \times F) \xRightarrow{\beta\ } F \circ \boxdot :  \CC \times \CC \to \LL.
\]
\end{nota}

\subsection{Models of LNL-FPC}\label{sub:lnl-fpc}

Before we may define our categorical model, we first introduce the notion of a \emph{pre-embedding}, which is novel
(to the best of our knowledge). Pre-embeddings are of fundamental importance for defining
the non-linear interpretations of non-linear types.

\begin{defi}
Let $T:\AAA\to\BB$ be a $\cpo$-functor between $\cpo$-categories $\AAA$ and $\BB$. Then a morphism $f$ in $\AAA$ is called a \emph{pre-embedding with respect to }$T$ if $T(f)$ is an embedding in $\BB$.
\end{defi}

We can now describe our categorical model by combining results from the previous two subsections. 
Our categorical model is a model of
Intuitionistic Linear Logic with some additional structure for solving
recursive domain equations. Because of that, we shall use the same notation for
it as introduced in the previous subsection.

\begin{defi}\label{def:model}
A \emph{$\cpo$-LNL model} is given by the following data:
\begin{enumerate}
\item[1.] A $\cpo$-symmetric monoidal closed category $(\LL, \otimes, \multimap, I)$ with finite $\cpo$-coproducts $(\LL, +, 0)$;
\item[2.] A $\cpo$-cartesian monoidal category $(\CC, \vartimes, 1)$ with finite $\cpo$-coproducts $(\CC, \amalg, \varnothing)$;
\item[3.] A $\cpo$-symmetric monoidal adjunction $\stikz{lnl-fpc-model.tikz}$, such that:
\item[4.] $\LL$ has an e-initial object $0$ and all $\omega$-colimits over embeddings; $\CC$ has all $\omega$-colimits over pre-embeddings, w.r.t $F$, and the product functor $(- \vartimes -) : \CC \times \CC \to \CC$ preserves them.
\end{enumerate}
\end{defi}

\begin{asm}
Throughout the rest of the paper we assume we are working with an arbitrary, but fixed, $\cpo$-LNL model as in Definition~\ref{def:model}.
\end{asm}

While not immediately obvious, $\LL$ is $\cpo$-algebraically compact. We reason as follows.

\begin{thm}\label{thm:model-compact}
In every $\cpo$-LNL model:
\begin{enumerate}
\item The initial object $0\in\LL$ is an ep-zero object;
\item Each zero morphism $\perp_{A,B}$ is least in $\LL(A,B)$;
\item $\LL$ is $\cpobs$-enriched;
\item $\LL$ is $\cpo$-algebraically compact.
\end{enumerate}
\end{thm}
\begin{proof}
(1) For every object $A\in\LL$, define a morphism ${\perp_A \coloneqq I \xrightarrow{e_I^\bullet} 0 \xrightarrow{e_A} A,}$
where $e_X : 0 \to X$ is unique and $e_X^\bullet: X \to 0$ is its projection counterpart.
Then, for every $h: A \to B$, we have $h\ \circ \perp_A = \perp_B$.
This means the category $\LL$ is \emph{weakly pointed} in the sense of~\cite[Definition 9]{brauner}. One can now define
morphisms
\[ \perp_{A,B} \coloneqq A \xrightarrow{\lambda_A^{-1}} I \otimes A \xrightarrow{\textbf{uncurry}(\perp_{A \multimap B})} B \]
and prove that:
\begin{align*}
g\ \circ \perp_{A,B} &= \perp_{A,C} & \text{for } g: B \to C && \text { \cite[Proposition 11]{brauner} } \\
\perp_{B,C} \circ\ f &= \perp_{A,C} & \text{for } f: A \to B && \text { \cite[Proposition 12]{brauner} }
\end{align*}
In the presence of an initial object, this then implies that $0$ is a zero object; $\perp_A = \perp_{I,A}$
and that the morphisms $\perp_{A,B}$ are zero morphisms~\cite[Lemma 4.8]{eclnl}. Therefore
$0$ is an ep-zero, because $(e_X, e_X^\bullet)$ forms an e-p pair for every object $X$
and also ${\perp_{A,B} = A \xrightarrow{e_A^\bullet} 0 \xrightarrow{e_B} B}$.

(2) For any $f: A \to B$, we have $ \perp_{A,B} = \perp_{B,B} \circ f = e_B \circ e_B^\bullet \circ f \leq \id_B \circ f = f. $

(3) is now immediate from (2) and (4) follows by Theorem~\ref{thm:alg-compact}.
\end{proof}

Every embedding in $\CC$ is a pre-embedding (because $\cpo$-functors preserve embeddings) and pre-embeddings are closed
under composition, so one can form the full-on-objects subcategory of $\CC$ of
all pre-embeddings. We define such a subcategory within our model.

\begin{defi}
Let $\PE$ be the full-on-objects subcategory of $\CC$ of pre-embeddings with respect to the functor $F$.
\end{defi}

The importance of this subcategory is that it allows us to \emph{reflect}
solutions of certain recursive domain equations from the category $\LL_e$ into
the category $\PE$.  As we have previously explained, the construction of
initial algebras of endofunctors on $\CC$\  is not done over embeddings, in
general (e.g. if $\CC = \cpo$), because the initial map $\iota : \varnothing \to X$ usually is not an embedding. Observe, however, that $\iota:\varnothing \to X$ always is a pre-embedding with respect to $F$.
In fact, we shall see that the initial algebras for \emph{all} non-linear type
interpretations are constructed over initial sequences consisting of
pre-embeddings.

We proceed with a simple lemma.

\begin{lem}\label{lem:creating colimits}
	Let $\AAA$ be a $\cpo$-category that has $\omega$-colimits over embeddings. Then $\AAA_e$ has all $\omega$-colimits, and the inclusion functor $\AAA_e\hookrightarrow\AAA$ is an $\omega$-functor that reflects $\omega$-colimits.
\end{lem}	
\begin{proof}
	This follows immediately by Theorem~\ref{thm:locally-determined} (1) -- (2).
\end{proof}	

Our next theorem shows $\PE$ is a suitable setting for constructing initial algebras.

\begin{thm}\label{thm:continuous embeddings}
In every $\cpo$-LNL model:
	\begin{enumerate}
	\item[(1)] $\LL_e$ is an $\omega$-category, and the subcategory inclusion $\LL_e\hookrightarrow\LL$ is a strict $\omega$-functor that also reflects $\omega$-colimits.
	\item[(2)] $\PE$ is an $\omega$-category and the subcategory inclusion $\PE \hookrightarrow \CC$ is a strict $\omega$-functor that also reflects $\omega$-colimits.
	\item[(3)] The subcategory inclusion $\CC_e \hookrightarrow \PE$ preserves and reflects $\omega$-colimits ($\CC_e$ does not necessarily have an initial object).
\end{enumerate}
\end{thm}
\begin{proof}
(1) follows from Lemma \ref{lem:creating colimits} and the fact that $\LL$ has an  e-initial object $0$. 
	
(2) Let $A:\PE\hookrightarrow\CC$ be the inclusion, and let $D:\omega\to\PE$ be an
$\omega$-diagram.
Let $\mu:D\to X$ be a cocone in $\PE$ and assume that it is colimiting in $\CC$.
We aim to show that $\mu$ is also colimiting in $\PE$, since this would imply
that the inclusion $A$ reflects $\omega$-colimits. So let $\nu: D\to Y$ be
another cocone in $\PE$, then $\nu$ is a cocone in $\CC$, hence there is a
unique morphism $f:X\to Y$ in $\CC$ such that $f\circ\mu=\nu$. Now, both
$F\mu:F\circ D\to FX$ and $F\nu:F\circ D\to FY$ are cocones of $F\circ D$ in
$\LL_e$. Since $\mu$ is colimiting in $\CC$, and $F$ (as a left adjoint)
preserves colimits, it follows that $F\mu$ is the colimiting cocone of $F\circ
D$ in $\LL$. Hence there is a unique $e:FX\to FY$ in $\LL$ such that $e\circ
F\mu=F\nu$. Since $f\circ\mu=\nu$, we obtain $Ff\circ F\mu=F\nu$, so $Ff=e$. By
(1), the inclusion $\LL_e\hookrightarrow\LL$ reflects $\omega$-colimits,
hence $e$ must be a morphism in $\LL_e$, hence an embedding. We conclude that
$Ff$ is an embedding, so $f$ is a pre-embedding, which shows that $\mu$ is
indeed colimiting in $\PE$. 

Next, we have to show that $\PE$ has all $\omega$-colimits, so we have to show
that the colimit of $D$ exists. First, since $\CC$ has $\omega$-colimits over pre-embeddings, the
colimiting cocone $\mu:D\to X$ of $D$ in $\CC$ exists. We have to show that
$\mu$ is actually a cocone in $\PE$. Since $F$ is a left adjoint, it preserves
colimits, hence $F\mu:F\circ D\to FX$ is colimiting in $\LL$. Now $F\circ
D:\omega\to\LL_e$ is an $\omega$-diagram, hence by (1), it
follows that $F\mu$ is also the colimiting cocone of $F\circ D$ in $\LL_e$,
meaning that $F\mu$ consists of embeddings. Hence $\mu$ consists of
pre-embeddings, so is a cocone in $\PE$. Since we already showed that the
inclusion $A$ reflects $\omega$-colimits, it follows that $\mu$ is colimiting
in $\PE$, so the colimit of $D$ exists in $\PE$, and is clearly preserved by
$A$. So the inclusion $A$ is an $\omega$-functor.

Let $\varnothing$ be the initial object of $\CC$, and let $X$ be some other
object of $\CC$. Let $i:\varnothing\to X$ be the unique initial map in $\CC$.
As a left adjoint, $F$ preserves initial objects, hence $F\varnothing\in\LL$ is
initial, and $Fi:F\varnothing\to FX$ is unique. Since $F\varnothing\cong 0$ is
e-initial, it follows that $Fi$ is an embedding. Hence $i$ is a pre-embedding,
so a morphism in $\PE$. It follows that also $\PE$ is an $\omega$-category, and
since we showed that $\varnothing$ is also initial in $\PE$, it follows that
$A$ is strict.

(3) We aim to show that the inclusion $B:\CC_e\hookrightarrow \PE$ preserves
and reflects $\omega$-colimits. Note that $A\circ B:\CC_e\hookrightarrow\CC$
is the inclusion into $\CC$.  It follows from Lemma
\ref{lem:creating colimits} that $\CC_e$ has all $\omega$-colimits, and the
inclusion $A\circ B:\CC_e\hookrightarrow\CC$ preserves and reflects
$\omega$-colimits.  Let $D:\omega\to\CC_e$ be a diagram, and let $\mu:D\to X$
be a cocone in $\CC_e$ such that $B\mu$ is colimiting in $\PE$.  Since $A$
preserves $\omega$-colimits, it follows that $(A\circ B)\mu$ is colimiting in
$\CC$. Since $A\circ B$ reflects $\omega$-colimits, it follows that $\mu$ is
colimiting in $\CC_e$ so $B$ indeed reflects $\omega$-colimits.  In order to
show that $B$ also preserves $\omega$-colimits, let $\mu$ be the colimiting
cocone of $D$ in $\CC_e$. Since $A\circ B$ preserves $\omega$-colimits, it
follows that $A(B\mu)=(A\circ B)\mu$ is colimiting in $\CC$. Since $A$
reflects $\omega$-colimits by (2), it follows that $B\mu$ is
colimiting, hence $B$ indeed preserves $\omega$-colimits.
\end{proof}

By definition, the functor $F : \CC \to \LL$ restricts to a functor
$F_{pe}:\PE\to\LL_e$. By Corollary \ref{cor:cpo functors restrict to omega continuous
functors} the functor $G : \LL \to \CC$ restricts to an $\omega$-functor
$G_e:\LL_e\to\CC_e$. Let $G_{pe} : \LL_e \to \PE$ be the functor given by
$G_{pe} \coloneqq \left( \LL_e \xrightarrow{G_e} \CC_e \hookrightarrow \PE \right)$.
The next theorem shows our model has suitable mediating functors for relating initial algebras.

\begin{thm}\label{thm:Fpe and Gpe continuous}
 		In any $\cpo$-LNL model, $F_{pe}:\PE\to\LL_e$ is a strict $\omega$-functor. In addition, ${G_{pe}:\LL_e\to\PE}$ is an $\omega$-functor.
\end{thm}  
\begin{proof}
Let $D: \omega \to \CC_{pe}$ be an $\omega$-diagram and let $\mu: D \to A$ be its colimit in $\CC_{pe}$.
Since $F_{pe}$ is simply a restriction of $F$, then its action on $D$ and $\mu$ is simply $F \circ D $ and $F \mu,$ respectively.
By Theorem~\ref{thm:continuous embeddings} (1), $\mu$ is also a colimit of $D$ in $\CC$ and since $F: \CC \to \LL$
is a left adjoint, it follows $F \mu$ is a colimit of $F \circ D$ in $\LL$. However, $F \circ D$ is a diagram
over embeddings (in $\LL$) and $F \mu$ is its colimit over embeddings (in $\LL$) and therefore by Theorem~\ref{thm:continuous embeddings} (1), $F \mu$ is a colimit of $F \circ D$
in $\LL_e$. Thus, $F_{pe}$ is an $\omega$-functor and strictness follows because $z : 0 \to F\varnothing$ is an
isomorphism.
By Corollary \ref{cor:cpo functors restrict to omega continuous functors} and by Theorem
\ref{thm:continuous embeddings} (3), we see that $G_{pe}$ is the composition of two $\omega$-functors
and is therefore an $\omega$-functor itself.
\end{proof}

By Notation~\ref{not:odot}, 
$ \odot \circ (F \times F) \xRightarrow{\beta\ } F \circ \boxdot :  \CC \times \CC \to \LL $
is a natural isomorphism
and by Corollary~\ref{cor:cpo functors restrict to omega continuous functors}, we see that
$\odot: \LL \times \LL \to \LL$ restricts to an $\omega$-functor $\odot_e : \LL_e \times \LL_e \to \LL_e.$	
The final theorem of this subsection, together with the previous two, show that $\PE$ is a good setting for the non-linear interpretation of non-linear types.
	
\begin{thm}\label{thm:boxdot is continuous}
  In any $\cpo$-LNL model,
  the bifunctor $\boxdot:\CC\times\CC\to\CC$ restricts to an
  $\omega$-bifunctor ${\boxdot_{pe}:\PE\times\PE\to\PE}$. Moreover,
  the natural isomorphism $\beta$ restricts to a natural isomorphism
	$ \odot_e \circ\ (F_{pe} \times F_{pe}) \xRightarrow{\beta^{pe}} F_{pe} \circ \boxdot_{pe} : \PE \times \PE \to \LL_e. $
\end{thm}		
\begin{proof}
  Let $f:X_1\to Y_1$ and $g:X_2\to Y_2$ be morphisms in $\PE$, so $Ff$ and
  $Fg$ are embeddings. Naturality of $\beta$ gives us :
  \[ \beta_{(Y_1,Y_2)}\circ ( Ff \odot  Fg) = F(f\boxdot g) \circ \beta_{(X_1,X_2)},\]
  hence
  \[  F(f\boxdot g)=  \beta_{(Y_1,Y_2)} \circ (Ff\odot Fg) \circ \beta_{(X_1,X_2)}^{-1} . \]
  As a natural isomorphism, each component of $\beta$
  is an embedding. Since $\odot : \CC\times\CC\to\CC$ is a $\DCPO$-functor,
  it preserves embeddings, hence $Ff\odot Fg$ is an embedding. We conclude
  that $F(f\boxdot g)$ is an embedding, so $\boxdot$ restricts to a functor
  $\boxdot_{pe} : \PE\times\PE\to\PE$. 
  
Let $D: \omega \to \CC_{pe} \times \CC_{pe}$ be an $\omega$-diagram and let $\mu: D \to (A,B)$ be its colimiting cocone in $\CC_{pe} \times \CC_{pe}$.
Since $\boxdot_{pe}$ is simply a restriction of $\boxdot$, then its action on $D$ and $\mu$ is simply $\boxdot \circ D $ and $\boxdot \mu,$ respectively.
By Theorem~\ref{thm:continuous embeddings} (2) and the fact that colimits are calculated pointwise in product categories, it follows $\mu$ is also a colimit of $D$ in $\CC \times \CC.$ Since $\boxdot : \CC \times \CC \to \CC$
preserves $\omega$-colimits over pre-embeddings ($\vartimes$ by assumption and $\amalg$ is a colimit), it follows $\boxdot \mu$ is a colimit of $\boxdot \circ D$ in $\CC$.
Therefore by Theorem~\ref{thm:continuous embeddings} (2), it may be reflected, so that $\boxdot \mu$ is a colimit of $\boxdot \circ D$
in $\PE$. Thus, $\boxdot_{pe}$ is an $\omega$-functor.
Finally, every isomorphism is an embedding and therefore every component of $\beta$ is an embedding (in $\LL$), so that indeed $\beta$ restricts to a natural isomorphism
\[ \odot_e \circ\ (F_{pe} \times F_{pe}) \xRightarrow{\beta^{pe}} F_{pe} \circ \boxdot_{pe} : \PE \times \PE \to \LL_e. \qedhere \]
\end{proof}

\subsection{Concrete CPO-LNL Models}\label{sub:concrete}

We now describe some concrete $\cpo$-LNL models.

\begin{thm}\label{thm:clnl-model}
The adjunction
\stikz{clnl-corollary.tikz},
where the left adjoint is given by (domain-theoretic) lifting and the right adjoint $U$ is the forgetful functor,
is a $\cpo$-LNL model.
\end{thm}

We will also show this model is computationally adequate in \secref{sec:adequate}. This concrete model actually arises from
a more general class of $\cpo$-LNL models, as we shall now explain.

Let $\M$ be a small $\cpobs$-symmetric monoidal category and let $\widehat \M$
be the category of $\cpobs$-presheafs, that is, the category whose objects are
the $\cpobs$-enriched functors $T : \M^\op \to \cpobs$ and whose morphisms are natural
transformations between them (the notion of a $\cpobs$-natural transformation
coincides with the ordinary notion). Then, by the enriched Yoneda Lemma, the
category $\widehat \M$ is complete and cocomplete and it has a
$\cpobs$-symmetric monoidal closed structure when equipped with the Day convolution~\cite{day-convolution}.
In this situation there exists an adjunction
\stikz{copower.tikz}, where the left adjoint is given by taking the $\cpobs$-copower with the tensor unit $I$ and the right adjoint is the representable functor (see~\cite[\S 6]{borceux:handbook2}).

\begin{thm}
Composing the two adjunctions
\stikz{big-model.tikz}
yields a $\cpo$-LNL model.
\end{thm}

Observe that any $\cpobs$-functor necessarily preserves least morphisms and therefore also ep-zero objects. Thus, if $\M = \mathbf{1}$, then $\widehat \M \simeq \mathbf 1$ and we recover the trivial (or completely degenerate) $\cpo$-LNL model.
If $\M = \mathbf 1_{\perp}$ (the category with a unique object and two morphisms $\id$ and $\perp,$ where $\perp$ is least), then $\widehat \M \cong \cpobs$ and we recover the $\cpo$-LNL model from Theorem~\ref{thm:clnl-model}.
Choosing $\M$ to be a suitable category of quantum computation, we recover a
model of Proto-Quipper-M, a quantum programming language~\cite{pqm-small}.
Choosing $\M$ to be a suitable category of string diagrams, we recover a model
of ECLNL, a programming language for string diagrams~\cite{eclnl}.
Similar presheaf models have been used to model Ewire, a circuit description language~\cite{ewire-mfps}, and the quantum lambda calculus~\cite{quantum-higher}.
All of these languages have mixed linear/non-linear features and therefore our results presented here could help with design decisions on introducing recursive (or inductive) datatypes.

\section{Denotational Semantics}\label{sec:semantics}

We now describe the denotational semantics of LNL-FPC. Every type $\Theta \vdash A$ admits a (standard) interpretation as an $\omega$-functor $\lrb{\Theta \vdash A } : \LL_e^{|\Theta|} \to \LL_e.$
Every \emph{non-linear} type ${\Theta \vdash P}$ admits an \emph{additional} non-linear interpretation as an $\omega$-functor ${\flrb{\Theta \vdash P} : \PE^{|\Theta|} \to \PE}$ that is related
to its standard interpretation via a natural isomorphism (see Figure~\ref{fig:alpha-def})
\[ \alpha^{\Theta \vdash P}: \lrb{\Theta \vdash P} \circ F_{pe}^{\times |\Theta|} \naturalto F_{pe} \circ \flrb{\Theta \vdash P} : \PE^{|\Theta|} \to \LL_e. \]
It is precisely the existence of this natural isomorphism that allows us to define the substructural operations (copying, discarding and promotion) on non-linear types.

As in FPC, the interpretation of a term $\Theta; \Gamma \vdash m : A$ is a \emph{family} of morphisms of $\LL$
{%
\small{
\[ \lrb{\Theta; \Gamma \vdash m : A} = \left\{ \lrb{\Theta; \Gamma \vdash m : A}_{\vec Z} : \lrb{\Theta \vdash \Gamma}(F^{\times |\Theta|}\vec Z) \to \lrb{\Theta \vdash A}(F^{\times |\Theta|}\vec Z) \ |\ \vec Z \in \text{Ob}(\CC^{|\Theta|}) \right\} \]
}%
}%
indexed by objects $\vec Z$ of $\CC^{|\Theta|}.$
In the special case when $\Theta = \cdot$ the interpretation can be seen as a morphism of $\LL$.
In order to show the semantics is sound, we also provide a non-linear interpretation of non-linear values that is compatible with the substructural operations.
We remark that types are interpreted as functors on $\LL_e$, which is a subcategory of $\LL$, the category in which terms are interpreted. Doing so is crucial for defining the semantics.

\subsection{Interpretation of LNL-FPC Types}\label{sub:type-semantics}

The (standard) interpretation of a well-formed type $\Theta \vdash A$ is a functor
$\lrb{\Theta \vdash A} : \LL_e^{|\Theta|} \to\ \LLe$ defined
by induction on the derivation of $\Theta \vdash A$ (Figure~\ref{fig:type-interpretation}, left).
The non-linear interpretation of a non-linear type $\Theta \vdash P$ is a functor
$\flrb{\Theta \vdash P} : \PE^{|\Theta|} \to\ \PE$ defined
by induction on the derivation of $\Theta \vdash P$ (Figure~\ref{fig:type-interpretation}, right).
The non-linear interpretation of non-linear types depends on the standard interpretation, so the latter needs to be defined first, as we have done.
Both assignments are $\omega$-functors defined on $\omega$-categories, as we show next, and so the assignments are well-defined, i.e.,
the parameterised initial algebras we need to interpret recursive types exist.
\begin{figure}
\centering
\begin{minipage}{0.5\textwidth}
\centering
\begin{align*}
\lrb{\Theta \vdash A} &:\ \LL_{\mathbf e}^{|\Theta|} \to \LLe\\
\lrb{\Theta \vdash \Theta_i} &\coloneqq \Pi_i &&\\
\lrb{\Theta \vdash !A} &\coloneqq\ !_e \circ \lrb{\Theta \vdash A} &&\\
\lrb{\Theta \vdash A + B } &\coloneqq +_e \circ \langle \lrb{\Theta \vdash A}, \lrb{\Theta \vdash B} \rangle &&\\
\lrb{\Theta \vdash A \otimes B } &\coloneqq \otimes_e \circ \langle \lrb{\Theta \vdash A}, \lrb{\Theta \vdash B} \rangle &&\\
\lrb{\Theta \vdash A \multimap B } &\coloneqq\ \multimap_e \circ\ \langle \lrb{\Theta \vdash A}, \lrb{\Theta \vdash B} \rangle &&\\
\lrb{\Theta \vdash \mu X.A} &\coloneqq \lrb{\Theta, X \vdash A}^\dagger &&
\end{align*}
\end{minipage}%
\begin{minipage}{0.5\textwidth}
\centering
	\begin{align*}
		\flrb{\Theta \vdash P} &:\PE^{|\Theta|} \to \PE &&\\
		\flrb{\Theta \vdash \Theta_i} &\coloneqq \Pi_i &&\\
		\flrb{\Theta \vdash !A} &\coloneqq G_{pe} \circ \lrb{\Theta\vdash A} \circ F_{pe}^{\times |\Theta|} &&\\
		\flrb{\Theta \vdash P + Q } &\coloneqq \amalg_{pe} \circ \langle\flrb{\Theta \vdash P}, \flrb{\Theta \vdash Q}\rangle  &&\\
			\flrb{\Theta \vdash P \otimes Q } &\coloneqq \vartimes_{pe} \circ \langle \flrb{\Theta \vdash P}, \flrb{\Theta \vdash Q} \rangle &&\\
		\flrb{\Theta \vdash \mu X.P} &\coloneqq \flrb{\Theta, X \vdash P}^{\dagger} &&
		\end{align*}
\end{minipage}%
\caption{Standard interpretation of types (left) and non-linear interpretation of non-linear types (right).}
\label{fig:type-interpretation}
\end{figure}


\begin{thm}
For any well-formed types $\Theta \vdash A$ and $\Theta \vdash P$, where $P$ is non-linear:
\begin{enumerate}
\item The assignment $\lrb{\Theta \vdash A} : \LL_e^{|\Theta|} \to\ \LLe$ is an $\omega$-functor.
\item The assignment $\flrb{\Theta \vdash P} : \PE^{|\Theta|} \to\ \PE$ is an $\omega$-functor.
\end{enumerate}
\end{thm}
\begin{proof}
Both statements are proved by induction one after the other.
The functor $\Pi_i$ is obviously an $\omega$-functor. The functors $!_e, +_e,
\otimes_e, \multimap_e$ are $\omega$-functors because of
Theorem~\ref{thm:e-p-trick}. The functors $G_{pe}, F_{pe}, \vartimes_{pe}, \amalg_{pe}$ are $\omega$-functors (Theorems \ref{thm:Fpe and Gpe continuous} and \ref{thm:boxdot is continuous}).
Moreover $\omega$-functors are closed under
composition and tupling. Finally, the $\mu$-case follows from Theorem~\ref{thm:dagger is
omega-continuous}. 
\end{proof}

The next theorem shows that the standard and non-linear interpretations of non-linear types are strongly related via a natural isomorphism (see Figure~\ref{fig:alpha-def}).
\begin{thm}\label{thm:alpha-def}
	For any non-linear type $\Theta \vdash P,$
	there exists a natural isomorphism
	\begin{align*}
	 \alpha^{\Theta \vdash P} &: \lrb{\Theta \vdash P} \circ F_{pe}^{\times |\Theta|} \naturalto F_{pe} \circ \flrb{\Theta \vdash P} : \PE^{|\Theta|} \to \LL_e 	& \text{inductively defined by} \\
    \alpha^{\Theta \vdash \Theta_i} & \coloneqq \id & \\
    \alpha^{\Theta \vdash !A} & \coloneqq \id & \\
    \alpha^{\Theta \vdash P \odot Q} & \coloneqq \beta^{pe} \langle \flrb{\Theta \vdash P}, \flrb{\Theta \vdash Q} \rangle \circ \odot_e \langle \alpha^{\Theta \vdash P}, \alpha^{\Theta \vdash Q}\rangle & \text{(see Theorem~\ref{thm:boxdot is continuous})} \\
    \alpha^{\Theta \vdash \mu X.P} & \coloneqq  (\alpha^{\Theta, X \vdash P})^\dagger &(\text{see Theorem~\ref{thm:alpha-dagger-def}}).
	\end{align*}
\end{thm}
\begin{proof}
In Appendix~\ref{proof:alpha-def}.
\end{proof}

\begin{figure}
\centering
\stikz[0.95]{big-conjecture-simple.tikz}
\caption{Relationship between standard and non-linear interpretations of non-linear types.} \label{fig:alpha-def}
\end{figure}

In FPC, one has to prove a substitution lemma that shows the interpretation of types respects type substitution.
In LNL-FPC, we have to prove three such lemmas: one for the standard
interpretation $\lrb{\Theta \vdash A}$ of (arbitrary) types, one for the non-linear interpretation $\flrb{\Theta \vdash P}$
of non-linear types, and one for the natural isomorphism $\alpha^{\Theta \vdash P}$ that relates them.
In order to do this, one first has to show a permutation and a contraction lemma.

\begin{lem}[Permutation]\label{lem:permutation}
Given types $\Theta, X, Y, \Theta' \vdash A$ and $\Theta, X, Y, \Theta' \vdash P$,
with $P$ non-linear:
\begin{enumerate}
\item $\lrb{\Theta, Y, X, \Theta' \vdash A} = \lrb{\Theta, X, Y, \Theta' \vdash A} \circ \swap_{m,m'}$
\item $\flrb{\Theta, Y, X, \Theta' \vdash P} = \flrb{\Theta, X, Y, \Theta' \vdash P} \circ \swap_{m,m'}$
\item $\alpha^{\Theta, Y, X, \Theta' \vdash P} = \alpha^{\Theta, X, Y, \Theta' \vdash P} \swap_{m,m'}$
\end{enumerate}
where
$|\Theta| = m$, $|\Theta'| = m'$
and $\swap_{m,m'} = \langle \Pi_1, \ldots, \Pi_m, \Pi_{m+2}, \Pi_{m+1}, \Pi_{m+3}, \ldots, \Pi_{m+m'+2} \rangle$.
\end{lem}
\begin{proof}
Essentially the same as~\cite[Lemma C.0.1]{fiore-thesis}.
\end{proof}

\begin{lem}[Contraction]\label{lem:contraction}
Given types $\Theta, \Theta' \vdash A$ and $\Theta, \Theta' \vdash P$,
with $P$ non-linear:
\begin{enumerate}
\item $\lrb{\Theta, X, \Theta' \vdash A} = \lrb{\Theta,  \Theta' \vdash A} \circ \drop_{m,m'}$
\item $\flrb{\Theta, X, \Theta' \vdash P} = \flrb{\Theta,  \Theta' \vdash P} \circ \drop_{m,m'}$
\item $\alpha^{\Theta, X, \Theta' \vdash P} = \alpha^{\Theta,  \Theta' \vdash P} \drop_{m,m'}$
\end{enumerate}
where
$X \not\in \Theta \cup \Theta'$,
$|\Theta| = m$, $|\Theta'| = m'$
and $\drop_{m,m'} = \langle \Pi_1, \ldots, \Pi_m, \Pi_{m+2}, \ldots, \Pi_{m+m'+1} \rangle$.
\end{lem}
\begin{proof}
Essentially the same as~\cite[Lemma C.0.2]{fiore-thesis}.
\end{proof}

In all models of FPC, the substitution lemma (and also the permutation and contraction lemmas) hold up to isomorphism. However,
in models where the equality of Proposition~\ref{prop:dagger-equal} holds, this can be strengthened to an equality (cf.~\cite[pp. 181]{fiore-thesis}).
Indeed, in LNL-FPC the same is true for the standard substitution lemma, but
the non-linear substitution lemma holds only up to isomorphism, which is induced by the isomorphism $\alpha$.

\begin{lem}[Substitution]\label{lem:alpha-substitution}
Given types $\Theta, X \vdash A$ and $\Theta \vdash B$ and $\Theta, X \vdash P$ and $\Theta \vdash R$ with $P$ and $R$ non-linear:
\begin{enumerate}
\item $\lrb{\Theta \vdash A[B/X]} = \lrb{\Theta, X \vdash A} \circ \langle \Id, \lrb{\Theta \vdash B}\rangle$
\item
	$\flrb{\Theta \vdash P[R/X]} \cong \flrb{\Theta, X \vdash P} \circ \langle \Id, \flrb{\Theta \vdash R}\rangle,$
	where the natural isomorphism, denoted $\gamma^{\Theta \vdash P[R/X]}$ is given by:
	\begin{align*}
	\gamma^{\Theta \vdash \Theta_i [R/X]} & =\id \\
	\gamma^{\Theta \vdash X [R/X]} & =\id  \\
	\gamma^{\Theta \vdash !A [R/X]} & =G_{pe} \lrb{\Theta,X\vdash A} \langle F_{pe}^{\times |\Theta|},\alpha^{\Theta\vdash R}\rangle\\
	\gamma^{\Theta \vdash (P \odot Q) [R/X]} & = \boxdot_{pe} \langle \gamma^{\Theta\vdash P[R/X]},\gamma^{\Theta\vdash Q[R/X]}\rangle\\
	\gamma^{\Theta \vdash \mu Y.P [R/X]} & = (\gamma^{\Theta,Y\vdash P[R/X]})^\dagger
	\end{align*}
\item \label{list:alpha3}
$ \alpha^{\Theta \vdash P[R/X]} = (F_{pe}\gamma^{\Theta\vdash P[R/X]})^{-1} \circ \alpha^{\Theta, X \vdash P} \langle \Id, \flrb{\Theta \vdash R}\rangle \circ \lrb{\Theta, X\vdash P} \langle F_{pe}^{\times |\Theta|}, \alpha^{\Theta \vdash R} \rangle$,
(see Figure~\ref{fig:substitution}).
\end{enumerate}
\begin{figure}
	\cstikz{alpha-substitution.tikz}
	\caption{The commuting diagram of natural isomorphisms for Lemma~\ref{lem:alpha-substitution} (\ref{list:alpha3}).}\label{fig:substitution}
\end{figure}
\end{lem}
\begin{proof}
In Appendix~\ref{proof:alpha-substitution}.
\end{proof}

The most important type isomorphism in FPC is the folding and unfolding of
recursive types. Denotationally, this is modelled using the properties of
parameterised initial algebras and the substitution lemma. The same is true in
LNL-FPC, but, in addition to the standard folding and unfolding of recursive types, we also have
a non-linear interpretation for folding and unfolding of non-linear recursive types.

\begin{defi}\label{def:fold-unfold}
Given types $\Theta \vdash \mu X. A$ and $\Theta \vdash \mu X. P,$ with $P$ non-linear, let
\begin{align*}
\sfold^{\Theta \vdash \mu X.A} &: \lrb{\Theta \vdash A[\mu X. A/ X]} \naturalto \lrb{\Theta \vdash \mu X. A} : \LL_e^{|\Theta|} \to \LL_e \\
\ifold^{\Theta \vdash \mu X.P} &: \flrb{\Theta \vdash P[\mu X. P/ X]} \naturalto \flrb{\Theta \vdash \mu X. P} : \PE^{|\Theta|} \to \PE
\end{align*}
be the natural isomorphisms given by
\begin{align*}
\sfold^{\Theta \vdash \mu X.A} &\coloneqq \left( \lrb{\Theta \vdash A[\mu X. A/ X]} = \lrb{\Theta, X \vdash A} \circ \langle \Id, \lrb{\Theta \vdash \mu X. A} \rangle  \xRightarrow{\phi} \lrb{\Theta \vdash \mu X. A} \right) \\
\ifold^{\Theta \vdash \mu X.P} &\coloneqq \left( \flrb{\Theta \vdash P[\mu X. P/ X]} \xRightarrow{\gamma} \flrb{\Theta, X \vdash P} \circ \langle \Id, \flrb{\Theta \vdash \mu X. P} \rangle  \xRightarrow{\phi} \flrb{\Theta \vdash \mu X. P} \right) .
\end{align*}
We denote their inverses by $\sunfold^{\Theta \vdash \mu X.A}$ and $\iunfold^{\Theta \vdash \mu X.P},$ respectively.
\end{defi}

We now show that given a non-linear recursive type, the standard and non-linear
interpretations of its folding and unfolding are also strongly related.
We shall use this theorem later in order to prove that values
$\Theta; \Gamma \vdash \fold\ v: P$ with $P$ non-linear admit a non-linear interpretation. 

\begin{thm}\label{thm:main}
Let $\Theta \vdash \mu X. P$ be a non-linear type. Then the diagram of natural isomorphisms
\cstikz{alpha-intuitionistic.tikz}
commutes.
\end{thm}
\begin{proof}
\begin{align*}
&\phantom{bbb} F_{pe} \ifold^{\Theta \vdash \mu X. P} \circ \alpha^{\Theta \vdash P[\mu X.P / X]} \\
&= F_{pe} \phi \circ F_{pe} \gamma^{\Theta \vdash P[\mu X. P / X]} \circ \alpha^{\Theta \vdash P[\mu X.P / X]} &(\text{Definition}) \\
&= F_{pe} \phi \circ \alpha^{\Theta, X \vdash P} \langle \Id, \flrb{\Theta \vdash \mu X.P}\rangle \circ
\lrb{\Theta,X \vdash P} \langle F_{pe}^{\times |\Theta|}, \alpha^{\Theta \vdash \mu X.P} \rangle&(\text{Lemma}~\ref{lem:alpha-substitution}\ (\ref{list:alpha3}))\\
&= F_{pe} \phi \circ \alpha^{\Theta, X \vdash P} \langle \Id, \flrb{\Theta, X \vdash P}^\dagger\rangle \circ \lrb{\Theta,X \vdash P} \langle F_{pe}^{\times |\Theta|}, (\alpha^{\Theta, X \vdash P})^\dagger \rangle&(\text{Definition})\\
&= (\alpha^{\Theta, X \vdash P})^\dagger \circ \phi F_{pe}^{\times |\Theta|} &(\text{Theorem}~\ref{thm:fixpoint-deconstruction})\\
&= \alpha^{\Theta \vdash \mu X. P} \circ \sfold^{\Theta \vdash \mu X. P}
  F_{pe}^{\times |\Theta|} &(\text{Definition})~~~~\quad\qedhere
\end{align*}
\end{proof}

\subsection{Interpretation of LNL-FPC Term Contexts}

The interpretation of term contexts is straightforward. The interpretation of an (arbitrary) term context $\Theta \vdash \Gamma$ is an $\omega$-functor
$\lrb{\Theta \vdash \Gamma} : \LL_e^{|\Theta|} \to \LL_e$ (see Figure~\ref{fig:term-context-interpretation}, left). A non-linear term context $\Theta \vdash \Phi$ admits
an additional non-linear interpretation as an $\omega$-functor $\flrb{\Theta \vdash \Phi} : \PE^{|\Theta|} \to \PE$ (see Figure~\ref{fig:term-context-interpretation}, right).
Just like with the interpretation of types, the two are strongly related.

\begin{thm}\label{thm:alpha-context}
	For any non-linear term context $\Theta \vdash \Phi,$
	there exists a natural isomorphism
	\[ \alpha^{\Theta \vdash \Phi}: \lrb{\Theta \vdash \Phi} \circ F_{pe}^{\times |\Theta|} \naturalto F_{pe} \circ \flrb{\Theta \vdash \Phi} : \PE^{|\Theta|} \to \LL_e \]
	that is given by:
	\begin{align*}
    \alpha^{\Theta \vdash \cdot} & \coloneqq k_u
      \qquad (\text{the constant } u \text{ natural transformation}) \\
    \alpha^{\Theta \vdash \Phi, x : P} & \coloneqq m^{pe} \langle \flrb{\Theta \vdash \Phi}, \flrb{\Theta \vdash P} \rangle \circ \otimes_e \langle \alpha^{\Theta \vdash \Phi}, \alpha^{\Theta \vdash P}\rangle
      \mbox{} \qquad \text{(see Theorem~\ref{thm:boxdot is continuous}).}
	\end{align*}
\end{thm}
\begin{proof}
Straightforward verification.
\end{proof}

\begin{figure}
\centering
\begin{align*}
\lrb{\Theta \vdash \Gamma} &: \LL_e^{|\Theta|} \to \LL_e  &  \flrb{\Theta \vdash \Phi} &: \PE^{|\Theta|} \to \PE \\
\lrb{\Theta \vdash \cdot} &\coloneqq K_I  &  \flrb{\Theta \vdash \cdot} &\coloneqq K_1 \\
\lrb{\Theta \vdash \Gamma, x : A} &\coloneqq \otimes_e \circ \langle \lrb{\Theta \vdash \Gamma} , \lrb{\Theta \vdash A} \rangle &  \flrb{\Theta \vdash \Phi, x : P} &\coloneqq \vartimes_{pe} \circ \langle \flrb{\Theta \vdash \Phi}, \flrb{\Theta \vdash P} \rangle
\end{align*}

\caption{Interpretation of term contexts (left); non-linear interpretation of non-linear term contexts (right). $K_X$ is the constant $X$ functor.}\label{fig:term-context-interpretation}
\end{figure}

\subsection{Interpretation of LNL-FPC Terms}

We introduce some notation for brevity that we use throughout the rest of the section. Given a type context $\Theta,$ we use $\vec{Z}$ to range over $\text{Ob}(\PE^{|\Theta|}) = \Ob(\CC^{|\Theta|})$.
We write $\vv{FZ}$ for $F_{pe}^{\times |\Theta|} \vec Z = F^{\times |\Theta|} \vec Z .$
We use $\Upsilon$ to range over both types and term contexts and we use $\Psi$ to range over non-linear types and non-linear term contexts.
Then, let
\begin{align*}
\lrbtz{\Upsilon}   &\coloneqq \lrb{\Theta \vdash \Upsilon} \left( \vv{FZ} \right) \in \Ob(\LL), &
\flrbtz{\Psi}  &\coloneqq \flrb{\Theta \vdash \Psi} \left( \vec Z \right) \in \Ob(\CC).
\end{align*}
With this notation in place, from the previous two subsections we have isomorphisms (in $\LL$)
\begin{align*}
 \sunfold^{\Theta \vdash \mu X. A}_{\vv{FZ} } : \lrbtz{\mu X. A} &\cong \lrbtz{A[\mu X. A / X]} : \sfold^{\Theta \vdash \mu X. A}_{\vv{FZ}}, &
 \alpha^{\Theta \vdash \Psi}_{\vec Z} : \lrbtz{\Psi} &\cong F \flrbtz{\Psi} . 
\end{align*}
\begin{rem}
The isomorphisms above are natural in $\PE^{|\Theta|}$, which we saw is
crucial for defining the type interpretations. However, as in FPC, this
naturality is irrelevant for the term interpretations -- we only need the fact that
each component is an isomorphism in $\LL$. Thus, there is no danger in working within
$\CC$ and $\LL$, instead of $\PE$ and $\LL_e$.
\end{rem}

Next, we show how to interpret the substructural rules of Intuitionistic Linear Logic in a $\cpo$-LNL model.
The isomorphism $\alpha$ plays a fundamental part.

\begin{defi}\label{def:discard}
We define discarding ($\diamond$), copying ($\triangle$) and promotion ($\Box$) morphisms:
{%
\small{%
\begin{align*}
  \diamond_{\vec Z}^{\Theta \vdash \Psi}  &\coloneqq \lrbtz{\Psi} \xrightarrow{\alpha} F \flrbtz{\Psi} \xrightarrow{F1} F1 \xrightarrow{u^{-1}} I \\
  \triangle_{\vec Z}^{\Theta \vdash \Psi} &\coloneqq \lrbtz{\Psi} \xrightarrow{\alpha} F \flrbtz{\Psi} \xrightarrow{F \langle \id, \id \rangle } F \left( \flrbtz{\Psi} \vartimes \flrbtz{\Psi} \right) \xrightarrow{m^{-1}} F \flrbtz{\Psi} \otimes F \flrbtz{\Psi}
      \xrightarrow{\alpha^{-1} \otimes \alpha^{-1}} \lrbtz{\Psi} \otimes \lrbtz{\Psi} \\
  \Box_{\vec Z}^{\Theta \vdash \Psi}      &\coloneqq \lrbtz{\Psi} \xrightarrow{\alpha} F \flrbtz{\Psi} \xrightarrow{F\eta}\, !F \flrbtz{\Psi} \xrightarrow{!\alpha^{-1}}\ ! \lrbtz{\Psi}
\end{align*}
}%
}%
where, for brevity, we have written $\alpha$ instead of $\alpha^{\Theta \vdash \Psi}_{\vec Z}.$
\end{defi}

\begin{prop}\label{prop:comonoid}
The triple $\left( \lrbtz \Psi, \triangle_{\vec Z}^{\Theta \vdash \Psi}, \diamond_{\vec Z}^{\Theta \vdash \Psi} \right)$ forms a cocommutative
comonoid in $\LL$.
\end{prop}
\begin{proof}
By the axioms of monoidal adjunctions and the cartesian structure of $\CC$.
\end{proof}

We can now define the interpretation of LNL-FPC terms. A term $\Theta; \Gamma \vdash m :A$ is interpreted as a family morphisms of $\LL$ indexed by $\vec Z.$
That is 
\[ \lrb{\Theta ; \Gamma \vdash m :A} := \left\{ \lrbtz{\Gamma \vdash m : A} : \lrbtz{\Gamma} \to \lrbtz{A} \ |\ \vec Z \in \Ob(\CC^{|\Theta|}) \right\}, \]
where $\lrbtz{\Gamma \vdash m : A}$ is defined by induction on the derivation of $\Theta; \Gamma \vdash m : A$ in Figure~\ref{fig:semantics-terms} (notice the remark).
Just as in FPC, the interpretation of terms is invariant with respect to both
$\Theta$ and $\vec Z$, so omitting the subscript and superscript annotations
(as we have done) should not lead to confusion, but on the other hand it improves readability. For brevity, we also informally write $\lrb{m}$ instead of $\lrbtz{\Gamma \vdash m :A}$, which
can be inferred from context.

\begin{figure}
\begin{align*}
&\lrb{\Phi, x: A \vdash x : A} \coloneqq
\lrb{\Phi} \otimes \lrb{A} \xrightarrow{\diamond \otimes \id} I \otimes \lrb{A} \xrightarrow{\cong} \lrb{A} \\
&\lrb{\Gamma \vdash \lleft_{A,B} m : A+B} \coloneqq
\lrb{\Gamma}\xrightarrow{\lrb{m}} \lrb{A} \xrightarrow{\mathrm{left}} \lrb{A}+\lrb{B}=\lrb{A+B} \\
&\lrb{\Gamma \vdash \rright_{A,B} m : A+B} \coloneqq
\lrb{\Gamma}\xrightarrow{\lrb{m}} \lrb{B} \xrightarrow{\mathrm{right}} \lrb{A}+\lrb{B}=\lrb{A+B} \\
&\lrb{\Phi, \Gamma, \Sigma \vdash \ccase\ m\ \texttt{of}\ \{\lleft\ x \to n\ |\ \rright\ y \to p\} : C} \coloneqq
\lrb \Phi \otimes \lrb{\Gamma} \otimes \lrb{\Sigma} 
\xrightarrow{\triangle  \otimes \id} \\
& \quad
\lrb \Phi \otimes \lrb \Phi \otimes \lrb{\Gamma} \otimes \lrb{\Sigma}
\xrightarrow \cong
\lrb \Phi \otimes \lrb{\Sigma} \otimes \lrb \Phi \otimes \lrb{\Gamma}
\xrightarrow{\id \otimes \lrb m}
\lrb \Phi \otimes \lrb{\Sigma} \otimes \lrb{A+B}
\xrightarrow \cong \\
& \quad \left( \lrb \Phi \otimes \lrb{\Sigma} \otimes \lrb A \right)
+
\left( \lrb \Phi \otimes \lrb{\Sigma} \otimes \lrb B \right)
\xrightarrow{\left[ \lrb n, \lrb p \right]}
\lrb C \\
&\lrb{\Phi, \Gamma, \Sigma \vdash \langle m,n \rangle: A \otimes B} \coloneqq
      \lrb{\Phi}\otimes\lrb{\Gamma}\otimes\lrb{\Sigma}
        \xrightarrow{\triangle  \otimes \id}
      \lrb{\Phi} \otimes \lrb{\Phi}\otimes\lrb{\Gamma}\otimes\lrb{\Sigma} \xrightarrow \cong\\
     &\quad \lrb{\Phi}\otimes\lrb{\Gamma}\otimes\lrb{\Phi}\otimes\lrb{\Sigma}
      \xrightarrow{\lrb m \otimes \lrb n}
        \lrb A \otimes \lrb B = \lrb{A \otimes B} \\
&\lrb{\Phi, \Gamma, \Sigma \vdash\llet\ \langle x,y \rangle=m\ \text{in}\ n:C} := \lrb{\Phi}\otimes\lrb{\Gamma}\otimes\lrb{\Sigma}
  			\xrightarrow{\triangle \otimes\id}\lrb{\Phi}\otimes\lrb{\Phi}\otimes \lrb{\Gamma}\otimes\lrb{\Sigma}\xrightarrow \cong \\
        &\quad \lrb{\Phi}\otimes\lrb{\Gamma}\otimes\lrb{\Phi}\otimes\lrb{\Sigma}
  			\xrightarrow{\lrb{m}\otimes \id}
  			\lrb{A\otimes B}\otimes \lrb{\Phi}\otimes\lrb{\Sigma} \xrightarrow \cong
       	\lrb{\Phi}\otimes\lrb{\Sigma}\otimes\lrb{A}\otimes\lrb{B}
  			\xrightarrow{\lrb{n}}
  			\lrb{C}\\
&\lrb{\Gamma \vdash \lambda x^A . m : A \multimap B} \coloneqq
\lrb{\Gamma} \xrightarrow{\textbf{curry}({\lrb m})} (\lrb A \multimap \lrb B) = \lrb{A \multimap B} \\
&\lrb{\Phi, \Gamma, \Sigma \vdash mn : B} \coloneqq
\lrb{\Phi} \otimes \lrb{\Gamma} \otimes \lrb{\Sigma}
\xrightarrow{\triangle \otimes\id}
\lrb{\Phi}\otimes\lrb{\Phi}\otimes
\lrb{\Gamma}\otimes\lrb{\Sigma}
\xrightarrow \cong \\
& \quad \lrb{\Phi}\otimes\lrb{\Gamma}\otimes\lrb{\Phi}\otimes\lrb{\Sigma}
\xrightarrow{\lrb{m}\otimes \lrb{n}}
\left( \lrb{A} \multimap \lrb{B} \right) \otimes \lrb{A}\xrightarrow{\eval}\lrb{B} \\
&\lrb{\Phi \vdash \lift\ m :\ !A} \coloneqq \lrb \Phi \xrightarrow{\Box}\ !\lrb \Phi \xrightarrow{! \lrb m}\ ! \lrb A = \lrb{!A} \\
&\lrb{\Gamma \vdash \force\ m :A} \coloneqq \lrb \Gamma \xrightarrow{\lrb m}\ ! \lrb A \xrightarrow{\epsilon} \lrb A \\
&\lrb{\Gamma \vdash \fold{}_{\mu X.A} m: \mu X. A} \coloneqq \lrb \Gamma \xrightarrow{\lrb m} \lrb{A[\mu X. A / X]} \xrightarrow{\sfold} \lrb{\mu X. A} \\
&\lrb{\Gamma \vdash \unfold\ m: A [\mu X. A / X]} \coloneqq \lrb \Gamma \xrightarrow{\lrb m} \lrb{\mu X. A} \xrightarrow{\sunfold} \lrb{A[\mu X. A / X]}
\end{align*}
where, for brevity, we write $\lrb{\Upsilon}$ for $\lrbtz{\Upsilon}$ for all types or term contexts $\Upsilon$ and where we omit the subscripts and type superscripts of some morphisms (which can be easily inferred).
\caption{Interpretation of \calc{} terms.}\label{fig:semantics-terms}
\end{figure}

The careful reader might have noticed that the above notation is not completely
precise, because terms do not have unique derivations. Technically, the
definition should be read as providing an interpretation for a derivation $D$
of $\Theta; \Gamma \vdash m : A$. However, the next theorem justifies our
notation.

\begin{thm}\label{thm:derivations}
Let $D_1$ and $D_2$ be two derivations of a judgement $\Theta; \Gamma \vdash m: A.$
Then $\lrbtz{D_1} = \lrbtz{D_2}.$
\end{thm}
\begin{proof}
As previously explained, two derivations may only differ in whether some
non-linear variables are propagated up into both premises or just one.
Variables that are propagated up but not used are always discarded. The proof follows by induction using
Proposition~\ref{prop:comonoid}.
\end{proof}

In order to prove our semantics sound, we have to first formulate a substitution lemma for terms.
This lemma, in turn, requires that the interpretation of \emph{non-linear
values} behave well with respect to the substructural
morphisms of Definition~\ref{def:discard}. To show this, we reason as follows.

\begin{defi}
Given a non-linear term context $\Theta \vdash \Phi$ and a non-linear type $\Theta \vdash P$,
a morphism $f: \lrbtz{\Phi} \to \lrbtz{P}$ in $\LL$ is called \emph{non-linear}, whenever there exists a morphism ${ f': \flrbtz{\Phi} \to \flrbtz{P} }$ in $\CC$, with
$ f = \left( \lrbtz{\Phi} \xrightarrow{\alpha^{\Theta \vdash \Phi}_{\vec Z}} F \flrbtz{\Phi} \xrightarrow{Ff'} F \flrbtz{P} \xrightarrow{ \left( \alpha^{\Theta \vdash P}_{\vec Z} \right)^{-1} } \lrbtz{P} \right) . $
\end{defi}

\begin{prop}\label{prop:non-linear-morphisms}
If $f: \lrbtz{\Phi} \to \lrbtz{P}$ is non-linear, then:
\begin{align*}
  \diamond_{\vec Z}^{\Theta \vdash P} \circ f &= \diamond^{\Theta \vdash \Phi}_{\vec Z} ; \\
  \triangle^{\Theta \vdash P}_{\vec Z} \circ f &= (f \otimes f) \circ \triangle^{\Theta \vdash \Phi}_{\vec Z} ; \\
  \Box^{\Theta \vdash P}_{\vec Z} \circ f &=\ !f \circ \Box^{\Theta \vdash \Phi}_{\vec Z}.
\end{align*}
In particular, non-linear morphisms are comonoid homomorphisms (with respect to Proposition~\ref{prop:comonoid}).
\end{prop}

Recall that by Lemma~\ref{lem:values-syntax}, non-linear values necessarily have a non-linear term context.

\begin{prop}\label{prop:values-non-linear}
Let $\Theta; \Phi \vdash v : P$ be a non-linear value. Then $\lrbtz{\Phi \vdash v : P}$ is non-linear.
\end{prop}
\begin{proof}
We define a non-linear interpretation \[\flrbtz{\Phi \vdash v : P} : \flrbtz{\Phi} \to \flrbtz P\] in $\CC$ by induction on the derivation of $\Theta; \Phi \vdash v : P$ as follows:
\begin{align*}
\flrb{\Phi, x : P \vdash x : P} &\coloneqq \flrb \Phi \times \flrb P \xrightarrow{\pi_2} \flrb P \\
\flrb{\Phi \vdash \lleft_{P,R} v : P +R } &\coloneqq \flrb \Phi \xrightarrow{\flrb v} \flrb P \xrightarrow{\mathrm{inl}} \flrb P \amalg \flrb R = \flrb{P + R} \\
\flrb{\Phi \vdash \rright_{P,R} v : P +R } &\coloneqq \flrb \Phi \xrightarrow{\flrb v} \flrb R \xrightarrow{\mathrm{inr}} \flrb P \amalg \flrb R = \flrb{P + R} \\
\flrb{\Phi \vdash \left\langle v, w \right\rangle : P \otimes R } &\coloneqq \flrb \Phi \xrightarrow{ \left\langle \flrb v, \flrb w \right\rangle } \flrb P \times \flrb R = \flrb{P \otimes R} \\
\flrb{\Phi \vdash \lift\ m :\ !A } &\coloneqq \flrb \Phi \xrightarrow{\eta} GF\flrb \Phi \xrightarrow{G\alpha^{-1}} G \lrb \Phi \xrightarrow{G \lrb m} G \lrb A = \flrb{!A} \\
\flrb{\Phi \vdash \fold_{\mu X. P} v : \mu X. P} &\coloneqq \flrb \Phi \xrightarrow{\flrb v} \flrb{P[\mu X. P / X]} \xrightarrow{\ifold} \flrb{\mu X. P},
\end{align*}
where $\pi_2$ is the second projection in $\CC$; $\mathrm{inl}$ and $\mathrm{inr}$ are the coproduct injections in $\CC$;  $\langle f, g \rangle$ is the unique map induced by the product in $\CC$ and $\ifold$ is defined in Definition~\ref{def:fold-unfold}.
The definition is invariant with respect to $\Theta$ and $\vec Z$, so again, we omit them from the subscript and superscript to improve readability.


To complete the proof, one has to show
\[ \lrbtz{\Phi \vdash v : P } = \left( \lrbtz{\Phi} \xrightarrow{\alpha^{\Theta \vdash \Phi}_{\vec Z}} F \flrbtz{\Phi} \xrightarrow{F \flrbtz{\Phi \vdash v : P} } F \flrbtz{P} \xrightarrow{ \left( \alpha^{\Theta \vdash P}_{\vec Z} \right)^{-1} } \lrbtz{P} \right) , \]
which follows by simple verification using the available categorical data.
The $\fold$ case is the most complicated (in general), but because of the preparatory work we have done, it may be established simply by using Theorem~\ref{thm:main}.
\end{proof}

Next, as usual, we also have to formulate a substitution lemma for terms. Before we do so, recall Lemma~\ref{lem:syntax-sub}.

\begin{lem}[Substitution]
Let $\Theta; \Phi, \Gamma, x:A \vdash m: B$ be a term and $\Theta; \Phi, \Sigma \vdash v : A$ a value with $\Gamma \cap \Sigma = \varnothing.$
Then
{\small{%
\begin{align*}
\lrb{m[v/x]} = \left(
      \lrb{\Phi}\otimes\lrb{\Gamma}\otimes\lrb{\Sigma}
        \xrightarrow{\cong \circ (\triangle  \otimes \id)}
     \lrb{\Phi}\otimes\lrb{\Gamma}\otimes\lrb{\Phi}\otimes\lrb{\Sigma}
      \xrightarrow{\id \otimes \lrb v}
        \lrb \Phi \otimes \lrb \Gamma \otimes \lrb A
      \xrightarrow{\lrb m}
        \lrb B
      \right)
\end{align*}
}}%
where we have omitted the obvious subscript and superscript annotations, for brevity.
\end{lem}
\begin{proof}
By induction using Lemma~\ref{lem:syntax-sub} and Propositions~\ref{prop:non-linear-morphisms} - \ref{prop:values-non-linear}.
\end{proof}

With this in place, we may now show our semantics is sound.

\begin{thm}[Soundness]\label{thm:soundness}
If $\Theta; \Gamma \vdash m :A$ and $m \Downarrow v,$ then
$\lrbtz{\Gamma \vdash m: A} = \lrbtz{\Gamma \vdash v: A}$.
\end{thm}
\begin{proof}
Straightforward induction.
\end{proof}

\subsection{Notation for closed types}\label{sub:notation-closed}

In the special case when $\Theta = \cdot$, the notation for terms and term contexts can be simplified, as we showed in~\secref{sub:lnl-terms}. The same is also true for the denotational semantics and we will use this simplified notation for the proof of
adequacy in the next section.
We denote with $*$ the unique object of the terminal category $\mathbf 1.$

Let $\lrb \Upsilon \coloneqq \lrb{\cdot \vdash \Upsilon}(*) \in \text{Ob}(\LL).$ 
Let $\flrb \Psi \coloneqq \flrb{\cdot \vdash \Psi}(*) \in \text{Ob}(\CC)$.
For $\chi \in \{\alpha, \diamond, \triangle, \Box \}$, we write $\chi^\Psi \coloneqq \chi^{\cdot \vdash \Psi}_{*}$.
Then, we have isomorphisms
$ \alpha^\Psi : \lrb \Psi \cong F \flrb \Psi$, 
and substructural morphisms $\diamond^\Psi : \lrb \Psi \to I,\ \triangle^\Psi : \lrb \Psi \to \lrb \Psi \otimes \lrb \Psi $ and $\Box^\Psi : \lrb \Psi \to ! \lrb \Psi.$

Writing $\sfold^{\mu X. A}$ for $\sfold^{\cdot \vdash \mu X. A}_*$ and $\sunfold^{\mu X. A}$ for its inverse, we get an isomorphism
$\sunfold^{\mu X. A} : \lrb{\mu X. A} \cong \lrb{A[\mu X. A / X]} : \sfold^{\mu X. A} . $

For the interpretation of terms, we write $\lrb{\Gamma \vdash m : A} \coloneqq \lrb{\Gamma \vdash m : A}^{\cdot}_* : \lrb \Gamma \to \lrb A$, which is a morphism in $\LL$ and its definition
can be simply read off from Figure~\ref{fig:semantics-terms}, disregarding the remark at the bottom.

\section{Computational Adequacy}\label{sec:adequate}

In this section we show that computational adequacy holds at
non-linear types for a class of $\cpo$-LNL models that satisfy some additional
conditions (Theorem ~\ref{thm:adequacy}). In the process, we also provide sufficient
conditions for program termination at \emph{any} type (Theorem ~\ref{thm:termination}).
The main difficulty is showing that the formal approximation relations
exist (Lemma~\ref{lem:formal-relations}), the proof of which is presented in
$\secref{sub:existence}$.

We begin by specifying the additional properties of $\cpo$-LNL models that
we use in the proof of adequacy.

\begin{defi}\label{def:adequate-lnl-fpc}
A \emph{computationally adequate} $\cpo$-LNL model is a $\cpo$-LNL model where:
\begin{enumerate}
\item $\id_I \neq\ \perp_{I,I}$ (or equivalently $I \not \cong 0$).
\item For any two pairs of parallel morphisms $f_1, g_1: I \to A$ and $f_2, g_2: I \to B$, such that $f_1 \neq \perp \neq g_1$ and $f_2 \neq \perp \neq g_2:$
$ f_1 \otimes f_2 \leq g_1 \otimes g_2 \Rightarrow f_1 \leq g_1 \text{ and } f_2 \leq g_2. $
\end{enumerate}
\end{defi}

For instance, the $\cpo$-LNL model of Theorem \ref{thm:clnl-model} is computationally adequate.

\begin{rem}
Condition (1) is clearly necessary for adequacy. If $I \cong 0$, then $\LL
\simeq \mathbf 1,$ and since $\LL$ is completely degenerate, then clearly
adequacy cannot hold.  Condition (2) is a strong assumption that requires the
tensor product of $\LL$ to behave like the smash product of domain theory. This is a
sufficient condition, but we do not know if it is necessary. However, it
appears to be unavoidable when using the standard proof strategy based on
formal approximation relations, just as we do (see Remark~\ref{rem:augmented-functors}). In principle, there could be a
different proof strategy that does not make this assumption, but we do not
know of such a strategy that works in our setting.
\end{rem}

\begin{asm}
Throughout the rest of the section we consider an arbitrary, but fixed,
computationally adequate $\cpo$-LNL model.
\end{asm}

We begin by showing that non-linear values correspond to non-zero morphisms.

\begin{lem}\label{lem:values-total}
Let $\cdot \vdash v: P$ be a non-linear value.
Then ${\lrb{v} \neq\ \perp}.$
\end{lem}
\begin{proof}
Assume for contradiction $\lrb v =\ \perp$.
Therefore $\diamond^P \circ \lrb v = \perp_{I,I}$.
Since $\cdot \vdash v :P$ is a non-linear value, it may be discarded so that
$\diamond^P \circ \lrb v = \diamond^\cdot = \id_I$ (Propositions~\ref{prop:non-linear-morphisms}-\ref{prop:values-non-linear}).
But then $\perp_{I,I}\ = \id_I,$ which contradicts Definition~\ref{def:adequate-lnl-fpc}.
\end{proof}

\begin{nota}
Given morphisms $f_1 : I \to A_1, \ldots, f_n: I \to A_n$, we write $\llangle f_1, \ldots, f_n \rrangle$
for the morphism \[ \llangle f_1, \ldots, f_n \rrangle := \left( I \xrightarrow \cong I \otimes \cdots \otimes I \xrightarrow{f_1 \otimes \cdots \otimes f_n} A_1 \otimes \cdots \otimes A_n \right) . \]
For example, ${\llangle f_1, f_2 \rrangle = (f_1 \otimes f_2) \circ \lambda_I^{-1}}$.
\end{nota}

Let
$
\text{Values}(A)    \coloneqq \{v\ |\ v\text{ is a value and } \cdot \vdash v:A\}
$
and
$
\text{Programs}(A)  \coloneqq \{p\ |\ \cdot \vdash p : A\} 
$
for any closed type $\cdot \vdash A$.
We prove adequacy using the standard method based on 
\emph{formal approximation relations}, a notion first devised in~\cite{plotkin-85}.
\begin{lem}\label{lem:formal-relations}
For each closed type $\cdot \vdash A$, there exist \emph{formal approximation relations:}
\begin{align*}
\tleq_{A} &\subseteq (\LL(I,\sem A)-\{\perp\})\times \text{Values}(A) \\
\sleq_{A} &\subseteq \LL(I , \lrb A) \times \text{Programs}(A)
\end{align*}
with the properties:
  \begin{description}[style=multiline,labelwidth=\widthof{\bfseries (A1.2)~~~:}, leftmargin=\labelwidth, align=right]
	\item[(A1.1)] $f\tleq_{A+B} \lleft\ v$ iff $\exists f'.\ f=\mathrm{left}\circ f'$ and $f'\tleq_{A}v$;
	\item[(A1.2)] $f\tleq_{A+B} \rright\ v$ iff $\exists f'.\ f=\mathrm{right}\circ f'$ and $f'\tleq_{B}v$;
	\item[(A2)] $f\tleq_{A\otimes B}\langle v,w\rangle$ iff $\exists f',f''.$
    $f= \llangle f', f'' \rrangle$ and $f'\tleq_{A}v$ and $f''\tleq_{B}w$;
	\item[(A3)] $f\tleq_{A \multimap B} \lambda x^A.\ p$ iff $\forall f' \in \LL(I, \lrb A),
    \forall v \in \text{Values}(A).$
    $ f' \tleq_{A} v \Rightarrow \mathrm{eval} \circ \llangle f, f' \rrangle \sleq_{B} p[v/x];$
	\item[(A4)] $f\tleq_{!A} \lift\ p$ iff\ $\exists g \in \CC(1, G \lrb A).\ f = I \xrightarrow u F(1) \xrightarrow {F(g)} !\lrb A$ and
              $\epsilon \circ f \sleq_{A} p;$
	\item[(A5)] $f \tleq_{\mu X.A} \fold\ v$ iff $\sunfold^{\mu X. A} \circ f \tleq_{A[\mu X.A / X]} v$; 
  \item[(B)] $f \sleq_{A} p \text{ iff } f \not = \perp\ \Rightarrow\ \exists v \in \text{Values}(A).\ p \Downarrow v \text{ and } f \tleq_{A} v.$
\end{description}
\end{lem}
\begin{proof}
Because of (A5), the relations are defined by recursion and proving their existence requires considerable effort.
The full proof is presented in \secref{sub:existence}.
\end{proof}
\begin{lem}\label{lem:relation-non-linear}
If $f \tleq_{P} v$, where $P$ is a non-linear type, then $f$ is a non-linear morphism.
\end{lem}
\begin{proof}
By induction on $\cdot \vdash v :P$, essentially the same as Proposition~\ref{prop:values-non-linear}.
\end{proof}

The next proposition is fundamental for the proof of adequacy.

\begin{prop}\label{prop:adequacy-subst}
Let $\Gamma \vdash m: A,$ where
$\Gamma = x_1: A_1, \ldots, x_n: A_n.$
Assume further we are given $v_i$ and $f_i$, such that
$f_i \tleq_{A_i} v_i.$
Then $\lrb m \circ \llangle f_1, \ldots, f_n \rrangle \sleq_{A} m[\vec v\ /\ \vec x].$
\end{prop}
\begin{proof}
By induction on $ \Gamma \vdash m : A,$ using Lemma~\ref{lem:relation-non-linear} where necessary.
\end{proof}

The next theorem establishes sufficient conditions for termination at \emph{any} type.

\begin{thm}[Termination]\label{thm:termination}
Let $\cdot \vdash p: A\ $ be a well-typed program.
If $\ \lrb{p} \not = \perp$, then $p \Downarrow.$
\end{thm}
\begin{proof}
This is a special case of the previous proposition when $\Gamma = \cdot$.
We get $\lrb{\cdot \vdash p: A} \sleq_{A} p,$ and
thus $p \Downarrow$ by definition of $\sleq_{A}$.
\end{proof}


\begin{thm}[Adequacy]\label{thm:adequacy}
Let $\cdot \vdash p: P$ be a well-typed program, with $P$ non-linear. Then
\[p \Downarrow \text{ iff }\ \ \lrb{p} \not = \perp.\]
\end{thm}
\begin{proof}
$(\Rightarrow)$ By soundness and Lemma~\ref{lem:values-total}; $(\Leftarrow)$ by Theorem~\ref{thm:termination}.
\end{proof}

This formulation of adequacy immediately implies another one that some readers might be more familiar with.
We shall say that a closed type $\cdot \vdash P$ is \emph{ground} if it is formed without any usage of $!$
or $\multimap$ (note that such a type is necessarily non-linear)\footnote{To make ground types inhabited, we also have to hardcode the unit type $I$ in the language (see Example~\ref{ex:types}, where we defined $I \equiv\ !(0 \multimap 0)$),
so that $I$ has a unique value. This is trivial and causes no problems.}.
By using soundness, adequacy, Lemma \ref{lem:values-total} and injectivity of $\lrb{-}$ on values of ground type, we get:

\begin{cor}
\label{cor:adequacy-ground}
Let $\cdot \vdash p : P$ and $\cdot \vdash v : P$ be a term and a value of ground type $P$. Then \[ p \Downarrow v \text{ iff } \lrb p = \lrb v. \]
\end{cor}

\subsection{Existence of the formal approximation relations}\label{sub:existence}

Our proof strategy is very
similar to the one in~\cite[Chapter 9]{fiore-thesis} and we now provide an
overview. We define a category where the objects are pairs $(X,\tleq)$ of
objects $X$ of $\LL$ and binary relations $\tleq \subseteq (\LL(I, X) - \{\perp\})
\times $ Values$(A)$ together with a suitable notion of morphism that respects
these relations (Definitions~\ref{def:relation} -- \ref{def:adequacy-category}).
We show that the constructed category has sufficient structure to solve
recursive domain equations (Theorem~\ref{thm:adequacy-compactness}) and 
that the functors we use to interpret types can be lifted to this category
(Theorem~\ref{thm:augmented-functors}), while respecting the obvious forgetful
functor. Then, we provide a non-standard interpretation of types that contains
not only the standard interpretation of types, but also the logical relations
needed for the adequacy proof (Proposition~\ref{prop:augmented}). We also show
that this \emph{augmented interpretation} is appropriate by showing that
(parameterised) initial algebras are computed in the same way as for the
standard interpretation (Corollary~\ref{cor:augmented-folding}) and that it
also respects substitution (Lemma~\ref{lem:adequacy-substitution}). This then
allows us to prove the existence of the required relations.

\begin{defi}\label{def:relation}
Given a closed type $\cdot \vdash A,$ an object $X \in \text{Ob}(\LL)$ and given a relation $\tleq \subseteq (\LL(I, X) - \{\perp\}) \times \text{Values}(A)$, we define another relation
$\sleq \subseteq \LL(I,X) \times $ Programs$(A)$ by
$g \sleq p$ iff $g \not = \perp \Rightarrow \exists\ v \in $ Values$(A).\ p \Downarrow v \land g \tleq v$.
\end{defi}

\begin{defi}\label{def:temp-category}
For any closed type $\cdot \vdash A$, we define a $\Poset$-category $\Rel(A)$ with:
\begin{itemize}
\item Objects $(X, \tleq)$, where $X \in \text{Ob}(\LL)$ and $\tleq \subseteq (\LL(I, X) - \{\perp\}) \times $ Values$(A)$ is a relation;
\item Morphisms $f: (X, \tleq_X) \to (Y, \tleq_Y)$ are exactly the morphisms $f : X \to Y$ in $\LL$, such that:
\[g \tleq_X v \Rightarrow f \circ g \sleq_Y v;\]
\item Composition and identities as in $\LL$;
\item Order inherited from $\LL$, i.e. $f \leq g $ (in $\Rel(A)$) iff $f \leq g$ (in $\LL$).
\end{itemize}
\end{defi}

Unfortunately, the category $\Rel(A)$ is not $\cpo$-enriched (in general) and
therefore it is unsuitable for our purposes. Because of this, we consider a
full subcategory that has the desired structure.

\begin{defi}\label{def:adequacy-category}
Let $\RR(A)$ be the full subcategory of $\Rel(A)$ consisting of objects $(X, \tleq)$, such that
for any $v \in \text{Values}(A)$, the predicate $(- \tleq v)$ is closed under suprema of ascending chains in $\LL(I,X) - \{\perp \}.$
More specifically, if $\{g_i\}_{i=1}^\omega$ is an increasing sequence and $g_i \tleq v$ for every $i \in \omega$,
then $\bigvee g_i \tleq v.$
\end{defi}

Next, we list some useful properties of the categories $\RR(A)$.

\begin{lem}\label{lem:adequacy-basic-properties}
Let $\cdot \vdash A$ be a closed type. In the category $\RR(A):$
\begin{enumerate}
\item For any object $(X, \tleq)$ and $v \in \text{Values}(A)$ and $g \not = \perp$, it follows $g \tleq v$ iff $g \sleq v$. 
\item For any morphism $f : (X, \tleq_X) \to (Y, \tleq_Y)$, if $g \sleq_X p$ then $f \circ g \sleq_Y p$.
\item For any object $(X, \tleq)$ and $p \in \text{Programs}(A)$, the predicate $(- \sleq p)$ is closed under taking suprema of increasing sequences in $\LL(I, X)$.
\end{enumerate}
\end{lem}
\begin{proof}
\textbf{(1)} Because $v \Downarrow v$ and then by definition.

\textbf{(2)} Assume $f \circ g \not = \perp$. Then $g \not = \perp$ and therefore $\exists v.\ p \Downarrow v\ \land\ g \tleq_X v.$
Since ${f : (X, \tleq_X) \to (Y, \tleq_Y)}$ is a morphism, it follows
$f \circ g \sleq_Y v$ and therefore $f \circ g \tleq_Y v$.
This now implies $f \circ g \sleq_Y p.$

\textbf{(3)}
Let $(f_i)_{i \in \omega}$ be an increasing sequence in $\LL(I,X),$ such that for each $i \in \omega$ we have $f_i \sleq p$.
Assume $\bigvee_i f_i \not = \perp.$ Then, there exists $k$, such that for $l \geq k: f_l \not = \perp$
and $f_l \sleq p.$ 
Therefore, $\exists v \in $ Values$(A)$, such that $p \Downarrow v$ and $f_l \tleq v,$ for every $l \geq k.$
Since $(- \tleq v)$ is closed under taking suprema of increasing sequences, we conclude
$\bigvee_{i \geq k} f_i \tleq v$. But then also $\bigvee_{i \geq k} f_i \sleq p$, by definition. To finish the proof, observe that $\bigvee_{i \geq 0} f_i = \bigvee_{i \geq k} f_i.$
\end{proof}

There is an obvious forgetful functor $U^A: \RR(A) \to \LL$, defined by:
\begin{equation}\label{eq:forgetful-adequacy}
U^A(X, \tleq) := X \text{\qquad and \qquad} U^A(f) := f
\end{equation}

The next lemma establishes some crucial properties of the categories $\RR(A).$

\begin{lem}\label{lem:adequacy-properties}
For any closed type $\cdot \vdash A:$
\begin{enumerate}
\item The category $\RR(A)$ is $\cpo$-enriched.
\item $\RR(A)$ has an ep-zero object $(0, \varnothing).$
\item $\RR(A)$ is $\cpobs$-enriched.
\item $U^A: \RR(A) \to \LL$ is a $\cpobs$-functor.
\item $\RR(A)$ has all $\omega$-colimits of embeddings.
\item $U^A_e: \RR(A)_e \to \LL_e$ is a strict $\omega$-functor (cf. Theorem~\ref{thm:e-p-trick}).
\end{enumerate}
\end{lem}
\begin{proof}
\textbf{(1)}
It is obvious that $\RR(A)$ is $\Poset$-enriched. We will show that suprema of ascending chains in $\RR(A)$ exist and coincide with those in $\LL$.
Assume that $\{f_i\}_{i=1}^\omega$ is an increasing sequence, where ${f_i: (X, \tleq_X) \to (Y, \tleq_Y)}$.
Assume $g$ and $v$ are arbitrary, such that $g \tleq_X v$. We have to show $(\bigvee_{i \geq 0} f_i) \circ g \sleq_Y v$. Towards that end, assume ${(\bigvee_{i \geq 0} f_i) \circ g \not = \perp.}$
Therefore there exists $k \in \omega$, such that for every ${l \geq k: f_l \circ g \not = \perp.}$ By definition, $f_l \circ g \sleq_Y v$
and since $f_l \circ g \not = \perp$ and $v \Downarrow v$, it follows $f_l \circ g \tleq v$. Since $(- \tleq_Y v)$ is closed under suprema, it follows $\bigvee_{i \geq k} f_i \circ g = (\bigvee_{i \geq k} f_i) \circ g = (\bigvee_{i \geq 0} f_i) \circ g \tleq_Y v$
and therefore it follows $(\bigvee_{i \geq 0} f_i) \circ g \sleq_Y v$.

Finally, observe that composition in $\RR(A)$ is a Scott-continuous operation,
because it coincides with composition in $\LL,$ which is itself a Scott-continuous
operation.

\textbf{(2)} Let $(X, \tleq_X)$ be an arbitrary object of $\RR(A)$.
$(0, \varnothing)$ is initial, because the unique $\LL$-morphism ${\perp_{0, X} : 0 \to X}$ is also a $\RR(A)$-morphism (the additional required condition is trivially satisfied for the empty relation).
Moreover, $(0, \varnothing)$ is terminal, because the unique $\LL$-morphism
$\perp_{X, 0} : X \to 0$ is also a $\RR(A)$-morphism (if $g \tleq_X v$ then
$\perp_{X,0} \circ g = \perp_{I,0} \sleq_\varnothing v$, where $\sleq_\varnothing$
is constructed via Definition~\ref{def:relation} when $\tleq = \varnothing$).
Finally, $\perp_{X,0} \circ \perp_{0, X} = \perp_{0,0} = \id_0$ and 
$\perp_{0,X} \circ \perp_{X, 0} = \perp_{X,X} \leq \id_X$ because the order in $\RR(A)$ is inherited from $\LL$. Thus, $(0, \varnothing)$ is an ep-zero.

\textbf{(3)} Combine (1), (2) and the fact that $\perp_{X,Y}$ is least in its hom-cpo.

\textbf{(4)} True by construction (the order and suprema in $\RR(A)$ coincide with those in $\LL$).

\textbf{(5)} Let $D = ((X_k, \tleq_k), e_k)_{k \in \omega}$ be an $\omega$-diagram, s.t. each $e_k: (X_k, \tleq_k) \to (X_{k+1}, \tleq_{k+1})$ is an embedding.
Let $\mu: UD \to X_\omega$ be the colimiting cocone in $\LL$ of the $\omega$-diagram $UD = (X_k, e_k)_{k \in \omega}$.
Define a relation $\tleq_\omega \subseteq (\LL(I,X_\omega) - \{\perp\}) \times \text{Values}(A)$ by:
\[f \tleq_{\omega} v \text{ iff } \forall k \in \omega.\ \mu_k^\bullet \circ f \sleq_k v , \]
where $\mu_k: X_k \to X_\omega$ is the colimiting embedding and $\mu_k^\bullet: X_\omega \to X_k$ is its corresponding projection.
We will show that $(X_\omega, \tleq_\omega)$ is the colimit of $D$ in $\RR(A)$ with colimiting cocone given by $\mu$.

We have to show $(- \tleq_\omega v)$ is closed under suprema of ascending chains, for arbitrary $v$. Let $(f_i)_{i \in \omega}$ be an increasing
sequence, s.t. $f_i \tleq_\omega v.$
Then ${ \forall k. \forall i.\ \mu_k^\bullet \circ f_i \sleq_k v. }$
But $(- \sleq_k v)$ is closed under suprema of increasing sequences, thus ${ \forall k.\ \bigvee_i \mu_k^\bullet \circ f_i = \mu_k^\bullet \circ \bigvee_i f_i \sleq_k v. }$
Then, by definition $\bigvee_i f_i \tleq_\omega v$.

Next, we show that $\mu_k: (X_k, \tleq_k) \to (X_\omega, \tleq_\omega)$ is a morphism of $\RR(A).$ Towards that end, assume that $g \tleq_k v$. We have to show $\mu_k \circ g \sleq_\omega v$.
To do this, assume further $\mu_k \circ g \neq \perp.$ Now, $\mu_k \circ g \sleq_\omega v$ is equivalent to $\mu_k \circ g \tleq_\omega v$,
which is equivalent to $\forall l.\ \mu_l^\bullet \circ \mu_k \circ g \sleq_l v.$
For $l \leq k$, we have 
\[ \mu_l^\bullet \circ \mu_k \circ g = e_{l,k}^\bullet \circ \mu_k^\bullet \circ \mu_k \circ g = e_{l,k}^\bullet \circ g \]
and since $e_{l,k}^\bullet : (X_k, \tleq_k) \to (X_l, \tleq_l)$ is a morphism, then it follows $e_{l,k}^\bullet \circ g \sleq_l v.$
For $l \geq k$, we have:
\[ \mu_l^\bullet \circ \mu_k \circ g = \mu_l^\bullet \circ \mu_l \circ e_{k,l} \circ g = e_{k,l} \circ g\]
and since $e_{k,l} : (X_k, \tleq_k) \to (X_l, \tleq_l)$ is a morphism, then it follows $e_{k,l} \circ g \sleq_l v.$
Combining both cases, we conclude that $\mu_k: (X_k, \tleq_k) \to (X_\omega, \tleq_\omega)$ is a morphism of $\RR(A)$.

Next, observe that $\mu_k^\bullet: (X_\omega, \tleq_\omega) \to (X_k, \tleq_k)$ is a morphism of $\RR(A),$ by definition of~$\tleq_\omega$.

Therefore, $\mu: D \to (X_\omega, \tleq_\omega)$ is a cocone over embeddings in $\RR(A).$
Clearly, $(\mu_i^\bullet \circ \mu_i)_{i \in \omega}$ forms an increasing sequence of morphisms in $\RR(A)$
and $\bigvee_i \mu_i^\bullet \circ \mu_i = \id$, because the order and suprema in $\RR(A)$ coincide with those of $\LL$.
By the limit-colimit coincidence theorem~\cite{smyth-plotkin:domain-equations}, it follows $\mu : D \to (X_\omega, \tleq_\omega)$ is a colimiting cocone, as required.

\textbf{(6)} Combining (4), (5) and Theorem~\ref{thm:e-p-trick} shows $U^A_e$ is an $\omega$-functor. Strictness is obvious.
\end{proof}

Note that by the above lemma, the categories $\RR(A)_e$ are $\omega$-categories.
Even better, we can show they are suitable categories for interpreting recursive types.

\begin{thm}\label{thm:adequacy-compactness}
Let $\cdot \vdash A$ be a closed type. Then $\RR(A)$ is $\cpo$-algebraically compact.
\end{thm}
\begin{proof}
Combine Theorem~\ref{thm:alg-compact} with Lemma~\ref{lem:adequacy-properties} (1) (2) (5).
\end{proof}

Next, we show that the functors we use to interpret type expressions may be
defined on the categories $\RR(A)$ by extending their action to relations as well.

\begin{defi}\label{def:adequacy-functors}
Let $\cdot \vdash A$ and $\cdot \vdash B$ be two closed types. We define functors as follows:
\begin{enumerate}
\item $!^A : \RR(A) \to \RR(!A)$ given by:
  \[ !^{A} (X, \tleq) := (!X, !\tleq) \qquad \text{and} \qquad !^{A} f := !f, \]
  where the relation $(!\tleq) \subseteq (\LL(I, !X) - \{\perp\}) \times \text{Values}(!A)$ is given by:
  \[ g\ (!\tleq)\ \lift\ p \quad \text{ iff } \quad \exists h \in \CC(1, GX).\ g = I \xrightarrow u F(1) \xrightarrow {F(h)} !X \text{ and } \epsilon \circ g \sleq p\]
\item  $\otimes^{A,B} : \RR(A) \times \RR(B) \to \RR(A \otimes B)$ given by :
  \[ (X, \tleq_X) \otimes^{A,B} (Y, \tleq_Y) := (X \otimes Y, \tleq_X \otimes \tleq_Y) \qquad \text{and} \qquad f \otimes^{A,B} h := f \otimes h, \]
  where the relation $(\tleq_X \otimes \tleq_Y) \subseteq (\LL(I, X \otimes Y) - \{\perp\}) \times \text{Values}(A \otimes B)$ is given by:
  \[ g\ (\tleq_X \otimes \tleq_Y)\ \langle v, w \rangle \quad \text{ iff }\quad \exists g', g''.\ g = (g' \otimes g'') \circ \lambda_I^{-1}\ \land\ g' \tleq_X v\ \land\ g'' \tleq_Y w\]
\item $\multimap^{A,B} : \RR(A)^\op \times \RR(B) \to \RR(A \multimap B)$ given by :
  \[ (X, \tleq_X) \multimap^{A,B} (Y, \tleq_Y) := (X \multimap Y, \tleq_X \multimap \tleq_Y) \qquad \text{and} \qquad f^\op \multimap^{A,B} h := f^\op \multimap h, \]
  where the relation $(\tleq_X \multimap \tleq_Y) \subseteq (\LL(I, X \multimap Y) - \{\perp\}) \times \text{Values}(A \multimap B)$ is given by:
  \[ g\ (\tleq_X \multimap \tleq_Y)\ \lambda x. p \quad \text{ iff }\quad \forall g' \in \LL(I, X),
    \forall v \in \text{Values}(A):\\
    g' \tleq_{X} v \Rightarrow \text{eval} \circ \llangle g, g' \rrangle \sleq_{Y} p[v/x] \]
\item  $+^{A,B} : \RR(A) \times \RR(B) \to \RR(A+B)$ given by :
  \[ (X, \tleq_X) +^{A,B} (Y, \tleq_Y) := (X+Y, \tleq_X + \tleq_Y) \qquad \text{and} \qquad f +^{A,B} h := f+h, \]
  where the relation $(\tleq_X + \tleq_Y) \subseteq (\LL(I, X+Y) - \{\perp\}) \times \text{Values}(A+B)$ is given by:
  \[ g\ (\tleq_X + \tleq_Y)\ \lleft\ v \quad \text{ iff } \quad \exists g'.\ g = \text{left} \circ g'\ \land\ g' \tleq_X v \]
  and
  \[ g\ (\tleq_X + \tleq_Y)\ \rright\ v \quad \text{ iff } \quad \exists g'.\ g = \text{right} \circ g'\ \land\ g' \tleq_Y v \]
\end{enumerate}
\end{defi}

Before we prove our next theorem, we need a couple of lemmas that allow us to
establish some order-reflection properties of our models that will be used
in its proof.

\begin{lem}\label{lem:injection-split-mono}
The coproduct injections in $\LL$ are split monomorphisms.
\end{lem}
\begin{proof}
Let $A \xrightarrow{\mathrm{left}} A+B \xleftarrow{\mathrm{right}} B$ be
two coproduct injections. Since the category $\LL$ is pointed:
\[ [\id_A , \perp_{B,A}] \circ \mathrm{left} = \id_A \qquad \text{ and } \qquad [\perp_{A,B}, \id_B] \circ \mathrm{right} = \id_B. \qedhere \]
\end{proof}

\begin{lem}[Order-reflection via adjunction]\label{lem:adjoint-reflection}
Let $f: X \to GA$ and $h: X \to GA$ be two parallel morphisms in $\CC$, where
$A \in \text{Ob}(\LL)$ and $X \in \text{Ob}(\CC)$. Then $F(f) \leq F(h)$ iff $f \leq h.$
\end{lem}
\begin{proof}
The right-to-left direction follows because $F$ is a $\cpo$-functor. For the other direction:
\begin{align*}
  F(f) \leq F(h)
& \Rightarrow G\epsilon \circ GF(f) \circ \eta \leq G\epsilon \circ GF(h) \circ \eta & (G \text{ and $(- \circ -)$ are Scott-continuous} )\\
& \Leftrightarrow G\epsilon \circ \eta \circ f \leq G\epsilon \circ \eta \circ h & (\text{Naturality of } \eta)\\
& \Leftrightarrow f \leq h & (\text{counit-unit equation})~~~~\quad\qedhere
\end{align*}
\end{proof}

\begin{thm}\label{thm:augmented-functors}
All functors from Definition~\ref{def:adequacy-functors} are well-defined and are moreover $\cpo$-functors. In addition, the following diagrams:
\cstikz{adequacy-forgetful.tikz}
commute, where $\star \in \{+, \otimes, \multimap \}$.
\end{thm}
\begin{proof}
\textbf{Case} $!:$

We first show $!^{A}$ is a well-defined functor. We begin by showing that ${(-\ (!\tleq)\ \lift\ p)}$ is closed under suprema of increasing sequences.
Let $(g_i)_{i \in \omega}$ be an increasing sequence, such that ${g_i\ (!\tleq)\ \lift\ p.}$ Then there exist $h_i \in \CC(1, GX)$, such that $g_i = F(h_i) \circ u$ and $\epsilon \circ g_i \sleq p.$
By Lemma~\ref{lem:adequacy-basic-properties} (3) it follows $ \epsilon \circ \bigvee_i g_i = \bigvee_i \epsilon \circ g_i \sleq p.$
Since $u$ is an isomorphism and $(g_i)_{i \in \omega}$ forms an increasing sequence, then $(Fh_i)_{i \in \omega}$ also forms an increasing sequence.
The functor $F$ reflects the order of this sequence (Lemma~\ref{lem:adjoint-reflection}) and therefore $(h_i)_{i \in \omega}$ also forms an increasing sequence. 
Moreover, $\bigvee_i g_i = \bigvee_i F(h_i) \circ u = (\bigvee_i F(h_i)) \circ u = F(\bigvee_i h_i) \circ u$
and it follows by definition that $\bigvee_i g_i (!\tleq )\ \lift\ p.$
Therefore $!^{A}$ is well-defined on objects.

Next, we show $!^{A}$ is well-defined on morphisms.
Let
$f : (X, \tleq_X) \to (Y, \tleq_Y)$
be a morphism in $\RR(A)$.
Assume that $g\ (!\tleq_{X} )\ \lift\ p$ and that ${!f  \circ g \not = \perp.}$
We have to show $!f \circ g\  (!\tleq_{Y})\ \lift\ p$. 
Then $\exists h.\ g = Fh \circ u$ and $\epsilon \circ g \sleq_{X} p.$
We have
\[!f \circ g = FGf \circ Fh \circ u = F(Gf \circ h) \circ u\]
and also
\[\epsilon \circ !f \circ g = f \circ \epsilon \circ g \sleq_Y p\]
because of Lemma~\ref{lem:adequacy-basic-properties} (2). Therefore, by definition $!f \circ g\ (!\tleq_Y)\ \lift\ p$.
Therefore $!^{A}$ is well-defined on morphisms.

Finally, it immediately follows by definition of $!^{A}$ that it preserves identities and composition, that it is a $\cpo$-functor and that the required diagram commutes.

\textbf{Case} $\star = \otimes :$

We first show $\otimes^{A,B}$ is a well-defined functor. We show ${ (-\ (\tleq_X \otimes \tleq_Y)\ \langle v, w \rangle) }$ is closed under suprema of increasing sequences.
Let $(g_i)_{i \in \omega}$ be an increasing sequence, such that $g_i\ (\tleq_X \otimes \tleq_Y)\ \langle v, w \rangle.$ Then there exist $g_i', g_i''$, such that $g_i = (g_i' \otimes g_i'') \circ \lambda_I^{-1}$ and such that $g_i' \tleq_X v$
and $g_i'' \tleq_Y w.$
Since $\lambda_I^{-1}$ is an isomorphism, $(g_i' \otimes g_i'')_{i \in \omega}$
is an increasing sequence and by Definition~\ref{def:adequate-lnl-fpc}(2) it
follows that both $(g_i')_{i \in \omega}$ and $(g_i'')_{i \in \omega}$ are
increasing sequences. Therefore $\bigvee_i g_i' \tleq_X v$ and $\bigvee_i g_i'' \tleq_Y w$.
Also:
\[\bigvee_i g_i = \bigvee_i (g_i' \otimes g_i'') \circ \lambda_I^{-1} = ((\bigvee_i g_i') \otimes (\bigvee_i g_i'')) \circ \lambda_I^{-1} \]
so that by definition $\bigvee_i g_i\ (\tleq_X \otimes \tleq_Y)\ \langle v, w \rangle.$
Therefore $\otimes^{A,B}$ is well-defined on objects.

Next, we show $\otimes^{A,B}$ is well-defined on morphisms.
Let
\begin{align*}
f_1 &: (X_1, \tleq_{X_1}) \to (Y_1, \tleq_{Y_1})\\
f_2 &: (X_2, \tleq_{X_2}) \to (Y_2, \tleq_{Y_2})
\end{align*}
be two morphisms in $\RR(A)$ and $\RR(B)$ respectively. Assume that $g\ (\tleq_{X_1} \otimes \tleq_{X_2})\ \langle v, w \rangle$ and that ${(f_1 \otimes f_2) \circ g \not = \perp.}$
We have to show $(f_1 \otimes f_2) \circ g\  (\tleq_{Y_1} \otimes \tleq_{Y_2})\ \langle v, w \rangle$. 
Then there exist $g',g''$, such that $g = (g' \otimes g'') \circ \lambda_I^{-1}$ and such that $g' \tleq_{X_1} v$
and $g'' \tleq_{X_2} w.$
Therefore, $f_1 \circ g' \sleq_{Y_1} v$ and $f_2 \circ g'' \sleq_{Y_2} w$.
It's easy to see that $f_1 \not = \perp \not = g'$ and $f_2 \not = \perp \not = g''.$
Therefore, $f_1 \circ g' \tleq_{Y_1} v$ and $f_2 \circ g'' \tleq_{Y_2} w$.
Also:
\[(f_1 \otimes f_2) \circ g = (f_1 \otimes f_2) \circ (g' \otimes g'') \circ \lambda_I^{-1} = ((f_1 \circ g') \otimes (f_2 \circ g'')) \circ \lambda_I^{-1} \]
By definition, ${(f_1 \otimes f_2) \circ g\ (\tleq_{Y_1} \otimes \tleq_{Y_2})\ \langle v, w \rangle.}$
Therefore, $\otimes^{A,B}$ is well-defined on morphisms.

Finally, it immediately follows by definition of $\otimes^{A,B}$ that it preserves identities and composition, that it is a $\cpo$-functor and that the required diagram commutes.

\textbf{Case} $\star = \multimap :$

We first show $\multimap^{A,B}$ is a well-defined functor. We show $(-\ (\tleq_X \multimap \tleq_Y)\ \lambda x. p)$ is closed under suprema of increasing sequences.
Let $(g_i)_{i \in \omega}$ be an increasing sequence, such that $g_i\ (\tleq_X \multimap \tleq_Y)\ \lambda x. p.$
Let $g'$ and $v$ be arbitrary, such that $g' \tleq_X v$. It now follows that $\text{eval} \circ \llangle g_i, g' \rrangle \sleq_Y p[v/x]$.
From Lemma~\ref{lem:adequacy-basic-properties} we get:
\[ \text{eval} \circ \llangle \bigvee_i g_i, g' \rrangle = \bigvee_i \text{eval} \circ \llangle g_i, g' \rrangle \sleq_Y p[v/x] \]
so that by definition $\bigvee_i g_i\ (\tleq_X \multimap \tleq_Y)\ \lambda x. p.$
Therefore $\multimap^{A,B}$ is well-defined on objects.

Next, we show $\multimap^{A,B}$ is well-defined on morphisms.
Let
\begin{align*}
f_1^\op &: (X_1, \tleq_{X_1}) \to (Y_1, \tleq_{Y_1})\\
f_2 &: (X_2, \tleq_{X_2}) \to (Y_2, \tleq_{Y_2})
\end{align*}
be two morphisms in $\RR(A)^\op$ and $\RR(B)$ respectively. Assume that $g\ (\tleq_{X_1} \multimap \tleq_{X_2})\ \lambda x. p$ and that ${(f_1^\op \multimap f_2) \circ g \not = \perp.}$
We have to show $(f_1^\op \multimap f_2) \circ g\ (\tleq_{Y_1} \multimap \tleq_{Y_2})\ \lambda x. p$. 
Let $g'$ and $v$ be arbitrary such that $g' \tleq_{Y_1} v.$ Since
\[ f_1 : (Y_1, \tleq_{Y_1}) \to (X_1, \tleq_{X_1})\]
it follows $f_1 \circ g' \sleq_{X_1} v.$

If $f_1 \circ g' = \perp,$ then it trivially follows that $\text{eval} \circ \llangle g, f_1 \circ g' \rrangle \sleq_{X_2} p[v/x]$.

If $f_1 \circ g' \neq \perp,$ then $f_1 \circ g' \tleq_{X_1} v$ and therefore $\text{eval} \circ \llangle g, f_1 \circ g' \rrangle \sleq_{X_2} p[v/x]$.

Therefore, in both cases, $\text{eval} \circ \llangle g, f_1 \circ g' \rrangle \sleq_{X_2} p[v/x]$. This then implies:
\[ f_2 \circ \text{eval} \circ \llangle g, f_1 \circ g' \rrangle \sleq_{Y_2} p[v/x], \]
because $f_2 : (X_2, \tleq_{X_2}) \to (Y_2, \tleq_{Y_2})$. And therefore:
\begin{align*}
& \text{eval} \circ \llangle (f_1^\op \multimap f_2) \circ g, g' \rrangle = \\
&= \text{eval} \circ (((f_1^\op \multimap f_2) \circ g) \otimes g') \circ \lambda_I^{-1}  \\
&= \text{eval} \circ ((f_1^\op \multimap f_2) \otimes \id) \circ (g \otimes g') \circ \lambda_I^{-1}  \\
&= \text{eval} \circ ((\id^\op \multimap f_2) \otimes \id) \circ ((f_1^\op \multimap \id) \otimes \id) \circ (g \otimes g') \circ \lambda_I^{-1}  \\
&= f_2 \circ \text{eval} \circ ((f_1^\op \multimap \id) \otimes \id) \circ (g \otimes g') \circ \lambda_I^{-1} \qquad (\text{Naturality of eval})\\
&= f_2 \circ \text{eval} \circ (\id \otimes f_1) \circ (g \otimes g') \circ \lambda_I^{-1} \qquad ({\text{Parameterised adjunction \cite[pp. 102]{maclane}}})\\
&= f_2 \circ \text{eval} \circ \llangle g, f_1 \circ g' \rrangle\\
&\sleq_{Y_2} p[v/x]
\end{align*}

By definition, ${(f_1^\op \multimap f_2) \circ g\ (\tleq_{Y_1} \multimap \tleq_{Y_2})\ \lambda x. p.}$
Therefore, $\multimap^{A,B}$ is well-defined on morphisms.

Finally, it immediately follows by definition of $\multimap^{A,B}$ that it preserves identities and composition, that it is a $\cpo$-functor and that the required diagram commutes.

\textbf{Case} $\star = + :$

We first show $+^{A,B}$ is a well-defined functor. We show $(-\ (\tleq_X + \tleq_Y)\ \lleft\ v)$ is closed under suprema of increasing sequences.
Let $(g_i)_{i \in \omega}$ be an increasing sequence, such that $g_i\ (\tleq_X + \tleq_Y)\ \lleft\ v.$ Then there exist $g_i'$, such that $g_i = \text{left}\ \circ g_i'\ \land\ g_i' \tleq_X v.$
Since left is a split monomorphism (Lemma~\ref{lem:injection-split-mono}), it follows $(g_i')_{i \in \omega}$ is an increasing sequence and therefore $\bigvee_i g_i' \tleq_X v.$ Since $\bigvee_i g_i = \text{left}\ \circ \bigvee_i g_i'$ then
by definition $\bigvee_i g_i\ (\tleq_X + \tleq_Y)\ \lleft\ v.$ Similarly, it follows that $(-\ (\tleq_X + \tleq_Y)\ \rright\ v)$ is also closed under suprema of increasing sequences. Therefore
$+^{A,B}$ is well-defined on objects.

Next, we show $+^{A,B}$ is well-defined on morphisms.
Let
\begin{align*}
f_1 &: (X_1, \tleq_{X_1}) \to (Y_1, \tleq_{Y_1})\\
f_2 &: (X_2, \tleq_{X_2}) \to (Y_2, \tleq_{Y_2})
\end{align*}
be two morphisms in $\RR(A)$ and $\RR(B)$ respectively. Assume that $g\ (\tleq_{X_1} + \tleq_{X_2})\ v$ and that ${(f_1 + f_2) \circ g \not = \perp.}$
We have to show $(f_1 + f_2) \circ g\  (\tleq_{Y_1} + \tleq_{Y_2})\ v$. There are two cases for $v$: either $v= \lleft\ w$ or $v = \rright\ w$. We
show the first one and the other one follows by complete analogy.
Thus, assume $v= \lleft\ w$. Then $\exists g'.\ g = \text{left} \circ g'\ \land\ g' \tleq_{X_1} w.$
Therefore, $f_1 \circ g' \sleq_{Y_1} w.$
Since $\text{left} \circ f_1 \circ g' = (f_1 + f_2) \circ \text{left} \circ g' = (f_1 + f_2) \circ g \not = \perp$ it follows 
$f_1 \circ g' \not = \perp$ and therefore
$f_1 \circ g' \tleq_{Y_1} w.$
By definition, ${(f_1 + f_2) \circ g\ (\tleq_{Y_1} + \tleq_{Y_2})\ \lleft\ w.}$
Therefore, $+^{A,B}$ is well-defined on morphisms.

Finally, it immediately follows by definition of $+^{A,B}$ that it preserves identities and composition, that it is a $\cpo$-functor and that the required diagram commutes.
\end{proof}

\begin{rem}
\label{rem:augmented-functors}
In the proof above, condition (2) of Definition~\ref{def:adequate-lnl-fpc} is
used to show $\otimes^{A,B}$ is well-defined. This is the only place where this
assumption is used in the paper.
\end{rem}

So far, we have defined operations (and functors) that show how to construct
appropriate logical relations for types of the form $!A, A+B, A \otimes B$ and $A
\multimap B$. Next, we show how to fold and unfold logical relations via
the functors $\FOLD{}$ and $\UNFOLD{}$ below. $\FOLD{}$ should be
thought of as the \emph{introduction} (fold) and $\UNFOLD{}$ as the
\emph{elimination} (unfold) for logical relations involving recursive types.

\begin{lem}\label{lem:adequacy-iso}
For every closed type $\cdot \vdash \mu X.A,$
there exists an isomorphism of categories \[\FOLD{\mu X.A} : \RR(A[\mu X.A / X]) \cong \RR(\mu X. A) : \UNFOLD{\mu X. A}\] given by:
\begin{align*}
&\FOLD{\mu X.A} :  \RR(A[\mu X.A / X]) \to \RR(\mu X. A)   & & \UNFOLD{\mu X.A} : \RR(\mu X. A) \to \RR(A[\mu X.A / X]) \\
&\FOLD{\mu X.A}(Y, \tleq) = (Y, \FOLD{\mu X.A} \tleq) & & \UNFOLD{\mu X.A}(Y, \tleq) = (Y, \UNFOLD{\mu X.A} \tleq)\\
&\FOLD{\mu X.A}(f) = f & & \UNFOLD{\mu X.A}(f) = f
\end{align*} 
where the folded relation $(\FOLD{\mu X.A} \tleq) $ is defined by: 
\begin{align*}
  (\FOLD{\mu X.A} \tleq) & \subseteq (\LL(I, Y) - \{\perp\}) \times \text{Values}(\mu X.A) \\
  g\ (\FOLD{\mu X.A} \tleq)\ \fold\ v \quad & \iff \quad g \tleq v 
\end{align*}
and the unfolded relation $(\UNFOLD{\mu X.A} \tleq)$ is defined by:
\begin{align*}
(\UNFOLD{\mu X.A} \tleq) & \subseteq (\LL(I, Y) - \{\perp\}) \times \text{Values}(A [\mu X.A / X]) \\
g\ (\UNFOLD{\mu X.A} \tleq)\ v \quad & \iff \quad g \tleq \fold\ v
\end{align*}
Moreover, the functors $\FOLD{\mu X.A}$ and $\UNFOLD{\mu X.A}$ are $\cpobs$-functors.
\end{lem}
\begin{proof}
We first show that the functor $\FOLD{\mu X.A}$ is well-defined on objects. This
is equivalent to showing $(-\ (\FOLD{\mu X.A} \tleq)\ \fold\ v)$ is closed under
suprema of increasing sequences in ${\LL(I, Y) - \{\perp\}}$. This, by definition, is equivalent to
showing that $(-\ \tleq\ v)$ is also closed under the same suprema, which is true by assumption. 

Next, we show the functor $\FOLD{\mu X.A}$ is well-defined on morphisms. Let
${f: (Y, \tleq_Y) \to (Z, \tleq_Z)}$
be a morphism in $\RR(A [\mu X.A / X])$. Assume $g\ (\FOLD{\mu X.A}
\tleq_Y)\ \fold\ v$ and assume $f \circ g \neq \perp.$ From the former it
follows $g \tleq_Y v$ and therefore $f \circ g \sleq_Z v.$ Since $f \circ g \neq \perp$ it follows $f \circ g \tleq_Z v$ and by definition
$f \circ g\ (\FOLD{\mu X.A} \tleq_Z)\ \fold\ v$. Therefore, $\FOLD{\mu X.A}$ is well-defined on morphisms.

Next, we show that the functor $\UNFOLD{\mu X.A}$ is well-defined on objects. This
is equivalent to showing $(-\ (\UNFOLD{\mu X.A} \tleq)\ v)$ is closed under
suprema of increasing sequences in $\LL(I, Y) - \{\perp\}$. This, by definition, is equivalent to
showing that $(-\ \tleq\ \fold\ v)$ is also closed under the same suprema, which is true by assumption. 

Next, we show $\UNFOLD{\mu X.A}$ is well-defined on morphisms. Let
$f: (Y, \tleq_Y) \to (Z, \tleq_Z)$
be a morphism in $\RR(\mu X.A)$. Assume $g\ (\UNFOLD{\mu X.A}
\tleq_Y)\ v$ and assume $f \circ g \neq \perp.$ From the former it
follows $g \tleq_Y \fold\ v$ and therefore $f \circ g \sleq_Z \fold\ v.$ Since $f \circ g \neq \perp$ it follows $f \circ g \tleq_Z \fold\ v$ and by definition it follows
$f \circ g\ (\UNFOLD{\mu X.A} \tleq_Z)\ v$. Therefore, $\UNFOLD{\mu X.A}$ is well-defined on morphisms.

Obviously, both functors preserve identities and composition.

Next, we show that for any relation $\tleq$, we have that $\FOLD{\mu X.A} (\UNFOLD{\mu X.A} \tleq ) = \tleq$.
This is true, because:
\[g\ (\FOLD{\mu X.A} (\UNFOLD{\mu X.A} \tleq))\ \fold\ v \quad \iff \quad g\ (\UNFOLD{\mu X. A}  \tleq)\ v \quad \iff \quad g \tleq \fold\ v.\]

Also, we show that for any relation $\tleq$, we have that $\UNFOLD{\mu X.A} (\FOLD{\mu X.A} \tleq ) = \tleq$.
This is true, because:
\[g\ (\UNFOLD{\mu X.A} (\FOLD{\mu X.A} \tleq))\ v \quad \iff \quad g\ (\FOLD{\mu X. A}  \tleq)\ \fold\ v \quad \iff \quad g \tleq v.\]
Therefore $\FOLD{\mu X.A} \circ \UNFOLD{\mu X.A} = \Id_{\RR(\mu X. A)}$ and $\UNFOLD{\mu X.A} \circ \FOLD{\mu X.A} = \Id_{\RR(A[\mu X.A / X])}$,
which shows the two functors are isomorphisms. Since both functors are identity on morphisms, it follows that they are both $\cpobs$-functors.
\end{proof}

This finishes the categorical development of the categories $\RR(A)$. Next, we
define a \emph{non-standard type interpretation} on the categories $\RR(A)_e$.
When we compose this interpretation with the forgetful functor, we get exactly
the standard type interpretation (Proposition~\ref{prop:augmented}). Since the new interpretation carries
additional information (the logical relations needed for the proof of
adequacy), we shall refer to it as the \emph{augmented interpretation}.

\begin{nota}\label{not:vector-types}
In what follows, given a type context $\Theta = X_1, \ldots , X_n$ and closed types $\cdot \vdash C_i$ $(1 \leq i \leq n)$,
we write $\vec C$ for $C_1, \ldots , C_n$ and $[\vec C / \Theta]$ for $[C_1 / X_1, \ldots , C_n / X_n]$.
\end{nota}

\begin{defi}
Given a well-formed type $\Theta \vdash A$ and closed types $\vec C$ as in Notation~\ref{not:vector-types}, we define the \emph{augmented interpretation} to be the functor
\[ \elrbc{\Theta \vdash A} : \RR(C_1)_e \times \cdots \times \RR(C_n)_e \to \RR(A[ \vec C / \Theta ])_e \]
defined inductively by:
\begin{align*}
\elrbc{\Theta \vdash \Theta_i} &:= \Pi_i &&\\
\elrbc{\Theta \vdash !A} &:=\ !_e^{A[\vec C / \Theta]} \circ \elrbc{\Theta \vdash A} &&\\
\elrbc{\Theta \vdash A \star B } &:= \star_e^{A[\vec C / \Theta], B[\vec C / \Theta]} \circ \langle \elrbc{\Theta \vdash A}, \elrbc{\Theta \vdash B} \rangle &(\text{where }\star \in \{ +, \otimes, \multimap \})&\\
\elrbc{\Theta \vdash \mu X.A} &:= \left(\FOLD{\mu X. A[\vec C / \Theta]}_e \circ \elrb{\Theta, X \vdash A}{\vec C, \mu X. A[\vec C / \Theta]} \right)^\dagger &&
\end{align*}
\end{defi}

\begin{prop}\label{prop:augmented}
Every functor $\elrbc{\Theta \vdash A}$ is an $\omega$-functor and the following diagram:
\cstikz{extended-diagram.tikz}
commutes.
\end{prop}
\begin{proof}
By induction on $\Theta \vdash A$.

\textbf{Case } $\Theta \vdash \Theta_i.$ Obvious.

\textbf{Case } $\Theta \vdash !A.$

By induction $\elrbc{\Theta \vdash A}$ is an
$\omega$-functor. $\omega$-functors are closed under composition and
$!_e^{A[\vec C / \Theta]}$ is an $\omega$-functor because of
Theorem~\ref{thm:e-p-trick} and Theorem~\ref{thm:augmented-functors}. Therefore $\elrbc{\Theta \vdash !A}$ is an $\omega$-functor. Moreover:
\begin{align*}
   U^{!A[\vec C/ \Theta]}_e \circ \elrbc{\Theta \vdash !A}
&= U^{!A[\vec C/ \Theta]}_e \circ !_e^{A[\vec C / \Theta]} \circ \elrbc{\Theta \vdash A} & (\text{Definition})\\
&=\ !_e \circ U^{A[\vec C/ \Theta]}_e \circ \elrbc{\Theta \vdash A} & (\text{Theorem~\ref{thm:augmented-functors}})\\
&=\ !_e \circ \lrb{\Theta \vdash A} \circ (U^{C_1}_e \times \cdots \times U^{C_n}_e) & (\text{IH})\\
&=\ \lrb{\Theta \vdash !A} \circ (U^{C_1}_e \times \cdots \times U^{C_n}_e), & (\text{Definition})
\end{align*}
which shows the required diagram commutes.

\textbf{Case } $\Theta \vdash A \star B.$

By induction $\elrbc{\Theta \vdash A}$ and $\elrbc{\Theta \vdash B}$ are
$\omega$-functors. $\omega$-functors are closed under composition and pairing (\secref{sub:initial})
and $\star_e^{A[\vec C / \Theta], B[\vec C / \Theta]}$ is an $\omega$-functor because of
Theorem~\ref{thm:e-p-trick} and Theorem~\ref{thm:augmented-functors}. Therefore $\elrbc{\Theta \vdash A \star B}$ is an $\omega$-functor. Moreover:
\begin{align*}
& U^{A \star B[\vec C/ \Theta]}_e \circ \elrbc{\Theta \vdash A \star B} = \\
&= U^{A \star B[\vec C/ \Theta]}_e \circ \star_e^{A[\vec C / \Theta], B[\vec C / \Theta]} \circ \langle \elrbc{\Theta \vdash A}, \elrbc{\Theta \vdash B} \rangle & (\text{Definition})\\
&=\ \star_e \circ (U^{A[\vec C/ \Theta]}_e \times U^{B[\vec C/ \Theta]}_e) \circ \langle \elrbc{\Theta \vdash A}, \elrbc{\Theta \vdash B} \rangle & (\text{Theorem~\ref{thm:augmented-functors}})\\
&=\ \star_e \circ \langle U^{A[\vec C/ \Theta]}_e \circ \elrbc{\Theta \vdash A}, U^{B[\vec C/ \Theta]}_e \circ \elrbc{\Theta \vdash B} \rangle & \\
&=\ \star_e \circ \langle \lrb{\Theta \vdash A} \circ (U^{C_1}_e \times \cdots \times U^{C_n}_e), \lrb{\Theta \vdash B} \circ (U^{C_1}_e \times \cdots \times U^{C_n}_e) \rangle & (\text{IH})\\
&=\ \star_e \circ \langle \lrb{\Theta \vdash A}, \lrb{\Theta \vdash B} \rangle \circ (U^{C_1}_e \times \cdots \times U^{C_n}_e) & \\
&=\ \lrb{\Theta \vdash A \star B} \circ (U^{C_1}_e \times \cdots \times U^{C_n}_e), & (\text{Definition})
\end{align*}
which shows the required diagram commutes.

\textbf{Case } $\Theta \vdash \mu X. A.$

By induction $\elrb{\Theta, X \vdash A}{\vec C, \mu X.A[\vec C / \Theta] }$ is an
$\omega$-functor. $\omega$-functors are closed under composition and
$\FOLD{\mu X. A[\vec C / \Theta]}_e$ is an $\omega$-functor because of
Theorem~\ref{thm:e-p-trick} and Theorem~\ref{thm:augmented-functors}. Therefore $\elrbc{\Theta \vdash \mu X. A}$ is an $\omega$-functor, because of Theorem~\ref{thm:dagger is omega-continuous}.

Next, observe that for any closed type $\cdot \vdash \mu Y.B,$ we have:
\begin{equation}\label{eq:fold-forget}
U^{\mu Y.B} \circ \FOLD{\mu Y.B} = U^{B[\mu Y.B / Y]}
\end{equation}
which follows immediately from the definition of $\FOLD{}.$ Then:
\begin{align*}
& U^{\mu X. A[\vec C/ \Theta]}_e \circ \FOLD{\mu X. A[\vec C / \Theta]}_e \circ \elrb{\Theta, X \vdash A}{\vec C, \mu X. A[\vec C / \Theta]} = & \\
& = U^{A[\vec C/ \Theta][\mu X. A[\vec C/ \Theta] / X]}_e \circ \elrb{\Theta, X \vdash A}{\vec C, \mu X. A[\vec C / \Theta]} & \eqref{eq:fold-forget}\\
& = U^{A[\vec C/ \Theta, \mu X. A[\vec C/ \Theta] / X]}_e \circ \elrb{\Theta, X \vdash A}{\vec C, \mu X. A[\vec C / \Theta]} & (\text{Type identity}) \\
& = \lrb{\Theta, X \vdash A}  \circ (U^{C_1}_e \times \cdots \times U^{C_n}_e \times U^{\mu X.A[\vec C / \Theta]}_e),  & (\text{IH})
\end{align*}
which shows that all of the conditions of Assumption~\ref{ass:parameterised-square} are satisfied.
Therefore, by Corollary~\ref{cor:alpha-id} :
\begin{align*}
 U^{\mu X. A[\vec C/ \Theta]}_e \circ \left( \FOLD{\mu X. A[\vec C / \Theta]}_e \circ \elrb{\Theta, X \vdash A}{\vec C, \mu X. A[\vec C / \Theta]} \right)^\dagger =
 \lrb{\Theta, X \vdash A}^\dagger  \circ (U^{C_1}_e \times \cdots \times U^{C_n}_e )
\end{align*}
and then by definition:
\begin{align*}
 U^{\mu X. A[\vec C/ \Theta]}_e \circ \elrbc{\Theta \vdash \mu X.A} =
 \lrb{\Theta \vdash \mu X. A}  \circ (U^{C_1}_e \times \cdots \times U^{C_n}_e ),
\end{align*}
which shows the required diagram commutes.
\end{proof}

By considering the proof of the above proposition for type $\Theta \vdash \mu X. A,$ we get an important corollary which
shows that for any recursive type, the parameterised initial
algebra for the augmented interpretation is exactly the same as
the one for the standard interpretation.

\begin{cor}\label{cor:augmented-folding}
The following (2-categorical) diagram:
\cstikz{augmented-fixpoint.tikz}
commutes (see Theorem~\ref{thm:dagger is omega-continuous} for $\phi$).
\end{cor}
\begin{proof}
In Proposition~\ref{prop:augmented}, we have shown
\begin{align*}
U^{\mu X. A[\vec C/ \Theta]}_e \circ \left( \FOLD{\mu X. A[\vec C / \Theta]}_e \circ \elrb{\Theta, X \vdash A}{\vec C, \mu X. A[\vec C / \Theta]} \right) = \\
 = \lrb{\Theta, X \vdash A}  \circ (U^{C_1}_e \times \cdots \times U^{C_n}_e \times U^{\mu X.A[\vec C / \Theta]}_e),
\end{align*}
which satisfies the conditions of Assumption~\ref{ass:parameterised-square} (with $M$ and $N$ given by the (products of) forgetful functors). Therefore, by
Corollary~\ref{cor:phi-adequacy} we complete the proof.
\end{proof}

Next, we prove a substitution lemma for the augmented interpretations which
shows that they behave as expected. Its proof depends on a
permutation and a contraction lemma.

\begin{lem}[Permutation]\label{lem:adequacy-permutation}
For any well-formed type $\Theta, X, Y, \Theta' \vdash A$ and closed types $C_1, \ldots, C_{m+m'+2}:$
\[ \elrb{\Theta, Y, X, \Theta' \vdash A}{\vec C, C_{m+2}, C_{m+1}, \vec{C'}} = \elrb{\Theta, X, Y, \Theta' \vdash A}{\vec C, C_{m+1}, C_{m+2}, \vec{C'}} \circ \swap_{m,m'} \]
where $|\Theta| = m$, $|\Theta'| = m'$,
$\swap_{m,m'} = \langle \Pi_1, \ldots, \Pi_m, \Pi_{m+2}, \Pi_{m+1}, \Pi_{m+3}, \ldots, \Pi_{m+m'+2} \rangle$
and where $\vec C = C_1, \ldots, C_m$ and $\vec{C'} = C_{m+3}, \ldots, C_{m+m'+2}$.
\end{lem}
\begin{proof}
Straightforward induction, essentially the same as Lemma~\ref{lem:permutation} (1).
\end{proof}

\begin{lem}[Contraction]\label{lem:adequacy-contraction}
For any well-formed type $\Theta, \Theta' \vdash A,$ such that $X \not\in \Theta \cup \Theta',$ and closed types $C_1, \ldots, C_{m+m'+1}:$
\[ \elrb{\Theta, X, \Theta' \vdash A}{\vec C, C_{m+1}, \vec{C'}} = \elrb{\Theta,  \Theta' \vdash A}{\vec C, \vec{C'}} \circ \drop_{m,m'} \]
where
$|\Theta| = m$, $|\Theta'| = m'$,
$\drop_{m,m'} = \langle \Pi_1, \ldots, \Pi_m, \Pi_{m+2}, \ldots, \Pi_{m+m'+1} \rangle$
and where we have $\vec C = C_1, \ldots, C_m$ and $\vec{C'} = C_{m+2}, \ldots, C_{m+m'+1}$.
\end{lem}
\begin{proof}
Straightforward induction, essentially the same as Lemma~\ref{lem:contraction} (1).
\end{proof}

\begin{lem}[Substitution]\label{lem:adequacy-substitution}
For any well-formed types $\Theta, X \vdash A$ and $\Theta \vdash B$ and closed types $C_1, \ldots, C_{m}$:
\[ \elrbc{\Theta \vdash A[B/X]} = \elrb{\Theta, X \vdash A}{\vec C, B[\vec C / \Theta]} \circ \langle \Id, \elrbc{\Theta \vdash B} \rangle. \]
\end{lem}
\begin{proof}
Straightforward induction, essentially the same as Lemma~\ref{lem:alpha-substitution} (1).
\end{proof}

We may now finally prove the existence of the formal approximation relations as promised in Lemma~\ref{lem:formal-relations}.
First, we introduce some additional notation for convenience when dealing with closed types and their augmented interpretations.
Recall that the augmented interpretation of a closed type $\cdot \vdash A$ is a functor $\elrb{\cdot \vdash A }{\cdot}: \mathbf 1 \to \RR(A).$
Therefore let $\elrbs{A} := \elrb{\cdot \vdash A}{\cdot}(*)$, where $*$ is the unique object of $\mathbf 1$.
Then $\elrbs A$ is an object of $\RR(A)$. For the standard interpretation of types we use the notation from~\secref{sub:notation-closed}.

\begin{proof}[Proof of Lemma~\ref{lem:formal-relations}]\mbox{}\\
First, observe that by Proposition~\ref{prop:augmented}, $\elrbs A = (\lrb A, \tleq_A)$ for some
\[ \tleq_A \subseteq (\LL(I, \lrb A) - \{\perp\}) \times \text{Values}(A) . \]
We will show that these relations $\tleq_A$ have the required properties.
Next, recall that by Theorem~\ref{thm:e-p-trick}, $T_e(A,B) = T(A,B),$ for any (bi)functor $T$ and objects $A,B$. Then:

Property (B) holds by construction (cf. Definitions~\ref{def:relation}, \ref{def:temp-category} and \ref{def:adequacy-category}).

Properties (A1.1), (A1.2), (A2) and (A3) hold, because:
\[ \elrbs{A \star B} = \elrbs A \star^{A,B} \elrbs B = (\lrb {A \star B}, \tleq_A \star^{A,B} \tleq_B) \qquad\quad (\star \in \{+, \otimes, \multimap \})\]
and then by Definition~\ref{def:adequacy-functors}.

Property (A4) holds, because:
$ \elrbs{!A} = !^A (\elrbs A) = (\lrb{!A}, !\tleq_A) $
and then by Definition~\ref{def:adequacy-functors}.

To prove property (A5), we reason as follows. By using Corollary~\ref{cor:augmented-folding} and Lemma~\ref{lem:adequacy-substitution} in the special case where $\Theta = \cdot$, we get an isomorphism (in the category $\RR(\mu X.A)$)
\[ \phi_*^{-1} : \elrbs{\mu X.A} \cong \FOLD{\mu X.A} \left( \elrb{X \vdash A}{\mu X.A} \left( \elrbs{\mu X. A} \right) \right) = \FOLD{\mu X.A} \left( \elrbs{A[ \mu X. A / X]} \right) : \phi_*, \]
where $*$ is the unique object of the category $\mathbf 1.$
But by Corollary~\ref{cor:augmented-folding} and the notation from~\secref{sub:notation-closed}, it follows $\phi_*^{-1} = \sunfold^{\mu X.A}$ and so we have an isomorphism in $\RR(\mu X. A)$
\begin{equation}\label{eq:adequacy-final}
\begin{split}
\sunfold^{\mu X.A } : (\lrb{\mu X.A}, \tleq_{\mu X.A})
& = \elrbs{\mu X.A}\\
& \cong \FOLD{\mu X.A} \left( \elrbs{A[ \mu X. A / X]} \right)\\
& = (\lrb{A[ \mu X. A / X]}, \FOLD{\mu X.A} \tleq_{A[ \mu X. A / X]}) : \sfold^{\mu X. A}
\end{split}
\end{equation}
We may now easily complete the proof.
For the left-to-right direction of (A5):
\begin{align*}
&f \tleq_{\mu X. A} \fold\ v &\\
\Rightarrow \quad & \sunfold^{\mu X.A} \circ f\ (\FOLD{\mu X.A} \tleq_{A[ \mu X. A / X]} )\ \fold\ v & (\eqref{eq:adequacy-final} \text{ and Lemma~\ref{lem:adequacy-basic-properties} (1)})\\
\Rightarrow \quad & \sunfold^{\mu X.A} \circ f\ \tleq_{A[ \mu X. A / X]} v & (\text{Lemma~\ref{lem:adequacy-iso}})
\end{align*}
And for the right-to-left direction of (A5):
\begin{align*}
& \sunfold^{\mu X.A} \circ f\ \tleq_{A[ \mu X. A / X]} v & \\
\Rightarrow \quad & \sunfold^{\mu X.A} \circ f\ (\FOLD{\mu X.A} \tleq_{A[ \mu X. A / X]} )\ \fold\ v & (\text{Lemma~\ref{lem:adequacy-iso}})\\
\Rightarrow \quad & f = \sfold^{\mu X. A} \circ \sunfold^{\mu X.A} \circ f \tleq_{\mu X. A} \fold\ v & (\eqref{eq:adequacy-final} \text{ and Lemma~\ref{lem:adequacy-basic-properties} (1)})~~~~\quad\qedhere
\end{align*}
\end{proof}

\section{Related Work}\label{sec:related}

LNL-FPC can be seen as an extension with recursive types of the circuit-free
fragment of Proto-Quipper-M~\cite{pqm-small}, which corresponds to the CLNL
calculus in~\cite{eclnl}. The present paper adds a syntactic, operational and
denotational treatment of recursive types to these type systems.

Our work
is also closely related to the LNL calculus~\cite{benton-small} (also known as the adjoint calculus~\cite{benton-wadler}).
The LNL type assignment system has two kinds of term judgements -- a linear one and a non-linear one, which are interpreted in a monoidal category and a cartesian category, respectively.
However, LNL-FPC differs syntactically, because it has only one kind of
term judgement, which we believe results in a more convenient
syntax for programming.

Eppendahl~\cite{eppendahl-thesis} extends LNL with recursive types only for the
linear judgements. In that work, recursive types are treated linearly and are
interpreted in the linear category only. In LNL-FPC, we also show how to add
non-linear recursive types. We are able to treat them non-linearly, because we
present new techniques for solving recursive domain equations within $\cpo$-enriched cartesian categories
that we use for their interpretation.
Moreover, in LNL-FPC, both term-level recursion and type-level recursion are handled elegantly and agnostically with respect to the linearity of the type/term in question, which we think is a point in favor of our syntax.
In contrast, term recursion in LNL and the adjoint calculus seems to be less elegant~\cite[pp. 11]{benton-wadler}.
In an LNL-style syntax, it is also natural to add type recursion over the non-linear types, including non-linear function space, but it is an open problem how to interpret this in the cartesian category.
Indeed, even if we assume $(\CC, \times, [- \to - ], 1)$ is cartesian closed, then the methods presented here do not explain how to compute fixpoints involving the $[- \to -]$ bifunctor, because it is unclear how to see it as a
covariant functor on $\CC$ while still being able to form initial algebras.
In LNL-FPC, we do not assume our cartesian category is closed and so we do not
expose any syntax for the internal-hom of the cartesian category and we have
shown how to interpret the remaining non-linear recursive types (in both
categories, see Figure~\ref{fig:type-interpretation}).

Some mixed linear/non-linear type systems, such as the adjoint
calculus~\cite{benton-wadler}, also support additive conjunction. We did not
include additive conjunction in our type system, because we do not think it is
interesting from a programming perspective. However, it appears it could be
included without causing problems -- on the denotational side we merely have to
assume that our linear category has $\cpo$-products, which is true in all of
our concrete models.

The Lily language~\cite{lily1} is a polymorphic linear/non-linear lambda
calculus where recursive types may be encoded using polymorphism. The
non-linear types in Lily are those of the form $!A$, whereas in LNL-FPC the
non-linear types are also closed under sums, pairs and formation of recursive
types, which is a major focus of our paper and that we believe improves the
usability of the language.  A version of Lily where the function space $A
\multimap B$ is \emph{strict}, rather than linear, was considered
in~\cite{rosolini-simpson} and the language was equipped with a denotational
model based on synthetic domain theory. The authors prove two computational
adequacy results (one for call-by-name and one for call-by-value evaluation)
that hold at all types of the form $!A$.  Computational adequacy for Lily at
types $!A$ was established in~\cite{lily3}, again using a model based on
synthetic domain theory.  These are the only adequacy results, that we know of,
for languages with a similar feature set to LNL-FPC. Our adequacy result,
however, covers a larger range of types (all non-linear types, rather than
only those of the form $!A$) and our categorical model is considerably
simpler and easier to understand. LNL-FPC does not support polymorphism,
however.

LNL-FPC is also related to Barber's Dual
Intuitionistic Linear Logic (DILL)~\cite{dill}, which, like LNL-FPC, also has
one kind of term judgement, but it enforces a strict separation between
linear and non-linear contexts. As a result, for a lambda abstraction
$\lambda x^A. p$ in DILL, the variable $x$ is always treated as a linear
variable, whereas LNL-FPC allows it to be treated as non-linear, provided
$A$ is, resulting in increased convenience for programming.

Other type systems like $\lambda^q_{\to}$~\cite{Bernardy:linear_haskell},
$^a\lambda_{ms}$~\cite{Tov:practical} and Quill~\cite{Morris:best} support
linear or affine types where the substructural operations are implicit, whereas
in LNL-FPC only the weakening and contraction rules are implicit. However, none
of them have so far been equipped with a categorical semantics (which is the
main contribution of this paper).

Fiore and Plotkin model recursive types within FPC using categories
that are algebraically compact in an enriched
sense~\cite{fiore-thesis,fiore-plotkin}. Any category that is
algebraically compact in such a sense and is also cartesian closed is
necessarily degenerate. In a mixed linear/non-linear setting, we also have
to solve recursive domain equations within the cartesian closed category, and because
of this we cannot model recursive types using these techniques.

Another approach to modeling recursive types involves \emph{bilimit compact
categories}~\cite{levy-cbpv} where type expressions are modelled as functors $T
: \AAA^\op \times \BB^\op \times \AAA \times \BB \to \BB$ for which one can
solve recursive domain equations by using domain-theoretic
methods~\cite{pitts-relation}. We do not know if these techniques would also
work in a mixed linear/non-linear setting like ours, but we think that our
presentation is simpler and more compact, because we simply interpret type
expressions as functors $T: \AAA_e \times \BB_e \to \BB_e$.  In addition, our
categorical models can be easily simplified for languages that have only
inductive datatypes (recursive types allowing only $\otimes$ and $+$).  In such
a case, $\cpo$-enrichment
from the model may be removed and our semantics can be adapted quite easily.

In \cite{qpl-tcs,qpl-fossacs}, the authors present a denotational semantics for
the first-order quantum programming language QPL extended with inductive
datatypes (i.e., datatypes formed using $\otimes$ and $+$, but not $\multimap$
and $!$). This type system is not equipped with $\multimap$ and $!$ types and
in order to model copying of classical information, the authors use methods
based on the work reported here (and in \cite{icfp19}) that shows how to
construct the required comonoids.

\section{Conclusion and Future Work}

We introduced a new language called LNL-FPC, which can be seen as extending FPC
with a mixed linear/non-linear typing discipline (\secref{sec:syntax}). An
interesting feature of our type system is that it has implicit weakening and
contraction rules -- the users of the language do not have to explicitly
specify when to copy and discard non-linear variables, while still enjoying
linear type checking for the linear variables.

We presented a big-step call-by-value operational semantics for the language
and showed that a previously known canonical fixpoint operator on terms~\cite{eclnl} may be
derived via the isorecursive type structure, thereby recreating an important result from FPC (\secref{sec:operational}).

In preparation for the denotational semantics of the language, we first had
to establish some new coherence properties between (parameterised) initial algebras defined in potentially different categories and related by appropriate mediating functors
(\secref{sec:w-categories}).

We then described our categorical model,
which is a $\cpo$-enriched model of Intuitionistic Linear Logic with suitable
$\omega$-colimits (\secref{sec:categorical-model}). A crucial insight of our approach and of our
model, is that non-linear types that correspond to mixed-variance
functors can be seen as \emph{covariant} functors both on the linear category
of embeddings, but also on the non-linear category of
\emph{pre-embeddings}. We also presented an entire class of concrete models
that have found applications not only in classical functional programming, but also in
programming languages for certain emerging fields (quantum computation and
diagrammatic languages).

We proved our denotational semantics sound (\secref{sec:semantics}) and
cemented our results by demonstrating a computational adequacy result
(\secref{sec:adequate}). In doing so, we provided two kinds of denotational
interpretations for non-linear types, terms and contexts -- a standard
interpretation within the linear category, and one within the cartesian category
-- that we showed are very strongly related to one another via a coherent
natural isomorphism. This isomorphism, in turn, allowed us to elegantly
interpret the substructural operations of Intuitionistic Linear Logic and in
particular we precisely described the canonical comonoid structure of the
non-linear (recursive) types.  This, in turn, required developing new
techniques for solving recursive domain equations within a cartesian category,
which is perhaps the main contribution of this paper. More specifically, we
showed that the solution of many (mixed-variance) recursive domain equations on
the linear side (constructed over embeddings) can be reflected onto the
cartesian side (by constructing them over pre-embeddings) and this result can
be of interest even without our sound and adequate semantics.

As part of future work we would like to investigate what additional categorical
structure (if any) is required in order to design and model a type system
similar to ours, but where the substructural rules for promotion and
dereliction are also implicit. In practical terms, this would result in a programming
language that is more convenient for use, while still retaining the benefits of
linear type checking. In theoretical terms, the categorical treatment could lead to a
better understanding of the underlying design principles.

Another line of future work is to consider a more abstract model of LNL-FPC
where the $\cpo$-enrichment of our model is axiomatised away. This is a challenging problem, but if resolved, this could help to discover additional models of LNL-FPC.

\section*{Acknowledgments}
We thank Samson Abramsky, Chris Heunen, Mathys Rennela, Francisco Rios and
Peter Selinger for discussions regarding this paper. We also thank the
anonymous ICFP and LMCS referees for their feedback, which led to multiple improvements of
the paper.
We thank the Simons Institute for the Theory of Computing
where the initial part of this work took place. We also thank Schloss Dagstuhl
- Leibniz Center for Informatics for hosting us during the Quantum Programming
Languages seminar where another part of the work took place.

This work was partially funded by the AFOSR under the MURI
grant number FA9550-16-1-0082 entitled, "Semantics, Formal Reasoning, and Tool
Support for Quantum Programming". Vladimir Zamdzhiev is also supported by the
French projects ANR-17-CE25-0009 SoftQPro and PIA-GDN/Quantex.

\newpage

\bibliographystyle{alpha}
\bibliography{refs}

\newpage

\appendix

\section{Omitted Proofs}\label{app:model}

\subsection{Proof of Lemma~\ref{lem:s-naturality}}\label{proof:s-naturality}

\begin{proof}
For brevity we write $T_A, T_B, T_f$ for $T(A,-), T(B, -), T(f,-)$ respectively.

First, observe that $(T_f^{*n})_X = T(f, T(f, \cdots T(f, X) \cdots ).$

Therefore, $(S(T_f))_n = (T_f^{*n})_\varnothing = T(f, T(f, \cdots T(f, \varnothing) \cdots ).$

Next, observe that $T_A^n(g) = T(A, T(A, \cdots T(A, g) \cdots ).$

Next, we have
\begin{align*}
(T_f * S(T_f))_n = ( (T_f * S(T_B)) \circ (T_A * S(T_f)) )_n &= T(f, T_B^n(\varnothing)) \circ T(A, (T^{*n}_f)_\varnothing ) \\
 &= T(f, (T^{*n}_f)_\varnothing) = (T_f^{*n+1})_\varnothing.
\end{align*}
And therefore
\begin{align*}
((T_f * S(T_f))  \circ s^{T_A} )_n &= (T_f^{*n+1})_\varnothing \circ T_A^n(\iota_{T_A \varnothing}) \\
  & = T(f, T(f, \cdots T(f, T(f, \varnothing)) \cdots ) \circ T(A, T(A, \cdots T(A, \iota_{T_A \varnothing}) \cdots )\\
  & = T(f, T(f, \cdots T(f,  T(f, \varnothing) \circ \iota_{T_A \varnothing}) \cdots ) \\
  & = T(f, T(f, \cdots  T(f, \iota_{T_B \varnothing}) \cdots ) \\
  & = T_B^n(\iota_{T_B \varnothing}) \circ (T_f^{*n})_\varnothing \\
  & = (s^{T_B} \circ S(T_f))_n\qedhere
\end{align*}
\end{proof}

\subsection{Proof of Lemma~\ref{lem:whiskering colimit}}\label{proof:whiskering colimit}

\begin{proof}
Composing $\tau_n:D_n\to D_n'$ with $d_{n}':D_n'\to D_\omega'$ yields a morphism $d_{n}'\circ\tau_n:D_n\to D_\omega'$ for each $n\in\omega$, which forms a cocone of the diagram $D$.
 Since $D_\omega$ is the colimit of $D$, then $\colim(\tau) : D_\omega\to D_\omega'$ is by definition the unique map such that
	\[ 	\colim(\tau)\circ d_n=d'_n\circ\tau_n,\]
	hence by functoriality of $T$, we obtain
	\begin{equation}\label{eq:colim-tau}
	T \colim(\tau)\circ T d_n=Td'_n\circ T\tau_n
	\end{equation}
	for each $n\in\omega$.
	Now, $T\circ D : \omega\to\BB$ is an $\omega$-diagram and since $T$ preserves the colimit of $D$, it follows $TD_\omega$ is the colimit of $T\circ D$ with colimiting morphisms $Td_n:TD_n\to TD_\omega$ for each $n\in\omega$.
  Similarly, $TD_\omega'$ is the colimit of $T\circ D'$ with colimiting morphisms ${Td_n':TD_n'\to TD_\omega'}$.
  Then $Td'_n\circ T\tau_n:TD_n\to TD_\omega'$ forms a cocone of the diagram $T\circ D$, and by definition $\colim(T \tau) : TD_\omega\to TD_\omega'$ is the unique morphism such that
	\[  	\colim(T \tau)\circ T d_n=Td'_n\circ T\tau_n.  \]
	It follows now from (\ref{eq:colim-tau}) that $T \colim(\tau)=\colim(T \tau)$.	
\end{proof}

\subsection{Proof of Lemma~\ref{lem:whiskering colimit2}}\label{proof:whiskering colimit2}

\begin{proof}
	By Notation \ref{not:w-diagrams}, $(D_\omega,d_n)_{n\in\omega}$ is the colimiting cocone of $D$.
	By definition of $\omega$-functors, it follows that $(TD_\omega,Td_n)_{n\in\omega}$ and $(HD_\omega,Hd_n)_{n\in\omega}$ are the colimiting cocones of $T\circ D$ and $H\circ D$, respectively.
	Then $(HD_\omega,Hd_{n}\circ\tau_{D_n})_{n\in\omega}$ is a cocone of $TD$, and by definition $\colim(\tau D) :  TD_\omega\to HD_\omega$ is the unique morphism such that
	\[\colim(\tau D)\circ Td_n=Hd_n\circ\tau_{D_n}\]
	for each $n\in\omega$. 
	On the other hand, by naturality of $\tau:T\naturalto H$, we have 
	\[\tau_{D_\omega}\circ Td_n=Hd_n\circ\tau_{D_n},\]
	for each $n\in\omega$, hence we must have $\colim(\tau D)=\tau_{D_\omega}$.
\end{proof}

\subsection{Proof of Lemma~\ref{lem:alpha-star}}\label{proof:alpha-star}
%
\begin{proof}
	First, we have to show that for each $A \in \AAA$, we have that
  \[ \alpha^*_A : S( T(NA, -) ) \naturalto M \circ S( H(A, -) : \omega \to \DD \]
  is natural in $\omega$, i.e., we have to show that
	\begin{equation}\label{eq:betaAnatural}
 M S ( H(A,-)) (n \leq n+1) \circ (\alpha_A^*)_{n}=(\alpha^*_A)_{n+1}\circ 	S(T(NA,-)(n\leq n+1))
	\end{equation}
	for each $n\in\omega$. Let $\iota_1 : 0\to T(NA,0)$ and $\iota_2 : \varnothing \to H(A,\varnothing)$ be the indicated unique morphisms.
  By definition, equation (\ref{eq:betaAnatural}) is then equivalent to:
	\begin{equation}\label{eq:betaAnatural2}
	M H(A,-)^n\iota_2\circ(\alpha^*_A)_{n}=(\alpha^*_A)_{n+1}\circ T(NA,-)^n\iota_1.
	\end{equation}
	When $n=0$, the equality becomes $ M \iota_2 \circ z=(\alpha^*_A)_1\circ\iota_1, $ which is true by initiality.
	We proceed by induction, so assume that (\ref{eq:betaAnatural2}) holds for $n$. Expanding the equality for $n+1$, by definition of $\alpha^*_A$, we get following diagram:
		\cstikz{lem320-5.tikz}
	Here the outer square is just equation (\ref{eq:betaAnatural2}) for $n+1$. Subdiagram (2) commutes by naturality of $\alpha$, whereas (1) commutes because it is $T(NA,-)$ applied to equation (\ref{eq:betaAnatural2}), which is the induction step. We conclude that $(\alpha^*_A)_n$ is indeed natural in $n$.
  
  Next, we have to show that $\alpha^*_A$ is also natural in $A$. Hence let $f:A_1\to A_2 $ be a morphism in $\AAA$. Then we need to show that 
	\[   MS(H(f,-)) \circ \alpha^*_{A_1}=\alpha^*_{A_2}\circ S(T(Nf,-)). \]
	This equality holds if and only it holds for all its components $n$, i.e., if 
	\begin{equation}\label{eq:betaAnatural3}
	 MH(f,-)^n\varnothing\circ(\alpha^*_{A_1})_n=(\alpha^*_{A_2})_n\circ T(Nf,-)^n0
	\end{equation}
	holds for each $n\in\omega$. For $n=0$, it reduces to $\id \circ z = z \circ \id $, which indeed holds. We proceed by induction. Assume that (\ref{eq:betaAnatural3}) holds for $n$. Consider the following diagram:
		\cstikz[0.89]{lem320-6.tikz}
	Here the outer square is just (\ref{eq:betaAnatural3}) for $n+1$.
  Subdiagram (1) commutes because it is $TN(A_1,-)$ acting on equation (\ref{eq:betaAnatural3}) for $n$, the induction step.  Subdiagram (2) because $T$ is a bifunctor. 
  Subdiagram (3) commutes by naturality of $\alpha$.
  We conclude that (\ref{eq:betaAnatural3}) also holds for $n+1$ and therefore $\alpha^*$ is a natural transformation.

  Finally, a simple induction argument shows that every component $\alpha^*_A$ is an isomorphism and therefore $\alpha^*$ is a natural isomorphism.
\end{proof}

\subsection{Proof of Lemma~\ref{lem:alpha-dagger-def}}\label{proof:alpha-dagger-def}
\begin{proof}
For (1),
notice $\beta$ can be seen as a natural transformation $\beta:H\circ(\Id\times\Id)\naturalto\Id\circ H'$, hence $\beta^*$ becomes a natural transformation $\beta^*:H^*\circ\Id\naturalto(\Id\triangleright-)\circ(H')^*$, defined by
\begin{align*}
(\beta^*_A)_0     &\coloneqq \left( \varnothing \xrightarrow{\id_\varnothing} \varnothing \right) \\
(\beta^*_A)_{n+1} &\coloneqq \left( H(A, -)^{n+1} \varnothing \xrightarrow{H(A, (\beta_A^*)_n)} H(A, H'(A, -)^n \varnothing) \xrightarrow{\beta_{A, H'(A,-)^n \varnothing}} H'(A,-)^{n+1} \varnothing \right) .
\end{align*}
Furthermore, we have a natural transformation $M\beta\circ\alpha:T\circ(N\times M)\naturalto M\circ H'$, hence $(M\beta\circ\alpha)^*:T^*\circ N\naturalto (M\triangleright-)\circ (H')^*$ is a natural transformation defined by
{%
\small{%
\begin{align*}
& ((M\beta\circ\alpha)^*_A)_0     \coloneqq \left( 0 \xrightarrow z M \varnothing \right) \\
& ((M\beta\circ\alpha)^*_A)_{n+1} \coloneqq \\
& \qquad \left( T(NA, -)^{n+1} 0 \xrightarrow{T(NA, (\delta_A^*)_n)} T(NA, MH'(A, -)^n \varnothing) \xrightarrow{ \delta_{A, H'(A,-)^n \varnothing}} MH'(A,-)^{n+1} \varnothing \right) .
\end{align*}
}%
}%
where $\delta=M\beta\circ\alpha$.
We have to show that we have $(M\beta\circ\alpha)_A^*=M\beta_A^*\circ\alpha_A^*$. Base case:
\[ ((M\beta\circ\alpha)^*_A)_0=z=M(\id_\varnothing)\circ z=M(\beta_A^*)_0\circ(\alpha_A^*)_0. \] 
Assume as an induction hypothesis that $ ((M\beta\circ\alpha)^*_A)_n=(M\beta^*_A \circ \alpha^*_A)_n, $
and consider the following diagram.
\cstikz{lemC41.tikz}
Here (1) commutes, because it is $T(NA, -)$ applied to the induction hypothesis, (2) commutes by definition of $\alpha_A^*$, (3) commutes by naturality of $\alpha$, and (4) commutes by definition of $\beta_A^*$.
The outer right column of the diagram equals by definition $((M\beta\circ\alpha)_A^*)_{n+1}$. It follows that 
 \[((M\beta\circ\alpha)^*_A)_{n+1}=(M\beta^*_A \circ \alpha^*_A)_{n+1},\]
 which concludes the proof of the first statement.

For the second statement,
base case:
\begin{align*}
((\alpha(Q\times \Id) \circ T(\gamma \times M))^*_E)_0 &= z &(\text{Definition of } (-)^*)\\
&= z \circ \id_0           & \\
&= (\alpha^*_{QE})_0 \circ (T^*\gamma_E)_0 &(\text{Definition of } (-)^* \text{ and } T^*)
\end{align*}
For the induction step, assume that
$ ((\alpha(Q\times \Id) \circ T(\gamma \times M))^*_E)_n = (\alpha^*_{QE})_n \circ (T^*\gamma_E)_n .$
Using the induction hypothesis, then $((\alpha(Q\times \Id) \circ T(\gamma \times M))^*_E)_{n+1}$ is given by the composition of the right and top sides of the diagram:
\[ \stikz[0.85]{star-sequence-proof2.tikz} \]
Triangle (1) commutes by definition of $T^*$ and the Godement product:
{%
\small{%
\begin{align*}
(T^*\gamma_E)_{n+1} = T(\gamma_E,-)^{*(n+1)} = T(\gamma_E,-) * T(\gamma_E,-)^{*n} = T(\gamma_E, T(NQE,-)^n 0) \circ T(PE, (T^*\gamma_E)_n) .
\end{align*}
}%
}%
Square (2) commutes by bifunctoriality of $T$ and (3) commutes by definition of $\alpha^*$.
\end{proof}

\subsection{Proof of Lemma~\ref{lem:pentagon-sequence}}\label{proof:pentagon-sequence}

\begin{proof}
The two compositions are natural transformations of $\omega$-diagrams in $\DD,$ so we prove this by induction on $n \in \omega$.
We first note that if $\AAA$ is taken to be trivial, then $T^*=S(T)$ and $H^*=S(H)$. Recall that $S(T)(n)=T^n0$, and $S(H)(n)=H^n\varnothing$. Let $\iota_{T0}:0\to T0$ and $\iota_{H\varnothing}:\varnothing\to H\varnothing$ be the unique initial maps. Then we have $(s^T)_n=T^n(\iota_{T0})$ and $(s^H)_n=H^n(\iota_{H\varnothing})$. Then	$\alpha^*: T^* \naturalto MH^* : \omega\to\DD$ is a natural transformation given by: 
\begin{align*}
	\alpha^*_0     &\coloneqq \left( 0 \xrightarrow z M \varnothing \right) \\
	\alpha^*_{n+1} &\coloneqq \left( T^{n+1} 0 \xrightarrow{T\alpha_n^*} TMH^n \varnothing \xrightarrow{\alpha_{H^n \varnothing}} MH^{n+1} \varnothing \right) .
\end{align*}
In order to check that the diagram of natural transformations commutes, we first consider the $0$-component as a basis step. Then we have to show that
\[ (Ms^H\circ\alpha^*)_0=(\alpha*S(H)\circ T\alpha^*\circ s^T)_0.\]
Since the object $S(T)(0)$ that is the source of both sides of the last equality is the initial object $0$, it follows that the last equality indeed holds.
Next, as an induction hypothesis, we assume that the $n$-component of the diagram commutes, hence
\begin{equation}\label{eq:pentagon-sequence}
 (Ms^H\circ\alpha^*)_n=(\alpha*S(H)\circ T\alpha^*\circ s^T)_n,
 \end{equation}
 and we aim to show that the $(n+1)$-component commutes as well.
 Note that 
 \begin{align*}
 ( \alpha*S(H)\circ T\alpha^*\circ s^T)_n &= (\alpha*S(H))_n\circ T\alpha^*_n\circ s^T_n=\alpha_{S(H)(n)}\circ T\alpha^*_n\circ s^T_n \\
  &=\alpha_{H^n\varnothing}\circ T\alpha_n^*\circ s^T_n=\alpha_{n+1}^*\circ s^T_n,
 \end{align*}
 hence (\ref{eq:pentagon-sequence}) becomes
 \begin{equation}\label{eq:pentagon-sequence2}
 Ms^H_n\circ\alpha^*_n=\alpha_{n+1}^*\circ s^T_n.
 \end{equation}
 Consider now the following diagram:
\[ \stikz{pentagon-lemma.tikz} \]
Here the outside of the diagram is exactly the identity we aim to show, namely
the $(n+1)$-component. (1) commutes by definition of $\alpha_{n+1}^*$, (2)
commutes, since this is $T$ applied to equation (\ref{eq:pentagon-sequence2}),
and finally (3) commutes since $\alpha$ is natural.
\end{proof}

\subsection{Proof of Theorem~\ref{thm:locally-determined}}\label{proof:locally-determined}
\begin{proof}
  Using \cite[Theorem 2]{smyth-plotkin:domain-equations}, we get the following implications:
  \[ (4) \iff (3) \iff (1) \Rightarrow (2) \iff (5). \]
  To finish the proof we will show $(2) \Rightarrow (3),$ which is stated in a slightly weaker form in the same paper (but its proof is the same).

  Assume $\mu$ is a colimiting cocone of $D$ in $\CC_e$.
  Let $\nu: D \to B$ be a colimiting cocone of $D$ in $\CC$. Since $(1) \Rightarrow (3)$ and $(1) \Rightarrow (2)$
  it follows that:
  \begin{enumerate}
    \item[1.] Each $\nu_i$ is an embedding, $(\nu_i \circ \nu_i^\bullet)_i$ is an increasing sequence and $\bigvee_i \nu_i \circ \nu^\bullet_i = \id_B$;
    \item[2.] $\nu$ is a colimiting cocone of $D$ in $\CC_e$.
  \end{enumerate}

  Therefore, $\mu$ and $\nu$ are isomorphic as cocones and therefore there exists an isomorphism $h:B \to A$, such that $\mu_i = h \circ \nu_i.$
  Therefore $\mu_i$ is an embedding (for every $i$) and
  for $i \leq j$, we have:
  \[\mu_i \circ \mu_i^\bullet = h \circ \nu_i \circ \nu_i^\bullet \circ h^\bullet \leq h \circ \nu_j \circ \nu_j^\bullet \circ h^\bullet = \mu_j \circ \mu_j^\bullet \]
  so that $(\mu_i \circ \mu_i^\bullet)_i$ forms an increasing sequence and:
  \[ \bigvee_i \mu_i \circ \mu_i^\bullet = \bigvee_i h \circ \nu_i \circ \nu_i^\bullet \circ h^\bullet =  h \circ \left( \bigvee_i \nu_i \circ \nu_i^\bullet \right) \circ h^\bullet = h \circ \id_B \circ h^{-1} = \id_A. \qedhere\]
\end{proof}

\subsection{Proof of Theorem~\ref{thm:alpha-def}}\label{proof:alpha-def}
\begin{proof}
By induction on the derivation of $\Theta \vdash P.$ Let $|\Theta| = n $.
The case $\Theta \vdash \Theta_i$ is given by:
\begin{align*}
\lrb{\Theta \vdash \Theta_i} \circ F_{pe}^{\times n} = \Pi_i \circ F_{pe}^{\times n} = F_{pe} \circ \Pi_i = F_{pe} \circ \flrb{\Theta \vdash \Theta_i} .
\end{align*}
So that $\alpha^{\Theta \vdash \Theta_i} = \id$, as required.
The case $\Theta \vdash !A$ is given by:
\begin{align*}
\lrb{\Theta \vdash !A} \circ F_{pe}^{\times n} &=\, !_e \circ \lrb{\Theta \vdash A} \circ F_{pe}^{\times n}
= F_{pe} \circ G_{pe} \circ \lrb{\Theta \vdash A} \circ F_{pe}^{\times n}
= F_{pe} \circ \flrb{\Theta \vdash !A}
\end{align*}
So that $\alpha^{\Theta \vdash !A} = \id$, as required.

The case $\Theta \vdash P \odot Q$ is given by:
\begin{align*}
\lrb{\Theta \vdash P \odot Q} \circ F_{pe}^{\times n}
&=                                                                                          \odot_e \circ \langle \lrb{\Theta \vdash P}, \lrb{\Theta \vdash Q} \rangle \circ F_{pe}^{\times n}  & (\text{Definition}) \\
&=                                                                                          \odot_e \circ \langle \lrb{\Theta \vdash P}\circ F_{pe}^{\times n}, \lrb{\Theta \vdash Q} \circ F_{pe}^{\times n}\rangle  & \\
&\cong                                                                                      \odot_e \circ \langle F_{pe} \circ \flrb{\Theta \vdash P}, F_{pe} \circ \flrb{\Theta \vdash Q} \rangle  & (\text{IH}) \\
&=                                                                                          \odot_e \circ (F_{pe} \times F_{pe}) \circ \langle \flrb{\Theta \vdash P}, \flrb{\Theta \vdash Q} \rangle  & \\
&\cong                                                                                      F_{pe} \circ \boxdot_{pe} \circ \langle \flrb{\Theta \vdash P}, \flrb{\Theta \vdash Q} \rangle & (\text{Theorem~\ref{thm:boxdot is continuous}}) \\
&=                                                                                          F_{pe} \circ \flrb{\Theta \vdash P \odot Q} & (\text{Definition})
\end{align*}
Reading off the morphism, we get $\alpha^{\Theta \vdash P \odot Q} = \beta^{pe} \langle \flrb{\Theta \vdash P}, \flrb{\Theta \vdash Q} \rangle \circ \odot_e \langle \alpha^{\Theta \vdash P}, \alpha^{\Theta \vdash Q}\rangle$, as required.

The case $\Theta \vdash \mu X.P$ is given by:
\begin{equation*}
\begin{array}{rclr}
\lrb{\Theta \vdash \mu X. P} \circ F_{pe}^{\times n}
&=                                                   &\lrb{\Theta, X \vdash P}^\dagger \circ F_{pe}^{\times n}  &\\
&\cong                                               &F_{pe} \circ \flrb{\Theta, X \vdash P}^\dagger            & (\text{Theorem } \ref{thm:alpha-dagger-def})\\
&=                                                   &F_{pe} \circ \flrb{\Theta \vdash \mu X. P}                &
\end{array}
\end{equation*}
Reading off the morphism, we get $\alpha^{\Theta \vdash \mu X.P} = (\alpha^{\Theta, X \vdash P})^\dagger,$ as required.
\end{proof}

\subsection{Proof of Lemma~\ref{lem:alpha-substitution}}\label{proof:alpha-substitution}
\begin{proof}

\textbf{For (1):} Essentially the same as~\cite[Lemma C.0.3]{fiore-thesis}.
  
\textbf{For (2):}
  By induction on the derivation of $\Theta, X \vdash P.$ Let $|\Theta| = n.$

	Case $\Theta, X \vdash \Theta_i$:
	\[ 	\flrb{\Theta\vdash \Theta_i[R/X]}=	\flrb{\Theta\vdash \Theta_i}=\Pi_i=\Pi_i\circ\langle\Pi_1,\ldots,\Pi_n,	\flrb{\Theta\vdash R}\rangle=	\flrb{\Theta, X\vdash \Theta_i}\circ\langle\Id,	\flrb{\Theta\vdash R}\rangle.   \]
  so that $\gamma^{\Theta \vdash \Theta_i [R/X]} =\id,$ as required.

	Case $\Theta, X \vdash X$:
	\[ 	\flrb{\Theta\vdash X[R/X]} =	\flrb{\Theta\vdash R} = \Pi_{n+1} \circ \langle \Id, \flrb{\Theta \vdash R}\rangle = \flrb{\Theta, X \vdash X} \circ \langle \Id, \flrb{\Theta \vdash R}\rangle\]
  so that $\gamma^{\Theta \vdash X [R/X]} =\id,$ as required.
	
	Case $\Theta, X \vdash !A$:
	\begin{align*}
	\flrb{\Theta\vdash !A[R/X]}
  & = \flrb{\Theta\vdash ! (A[R/X])}                                                                          &\\
	& = G_{pe}\circ\lrb{\Theta\vdash A[R/X]}\circ F_{pe}^{\times n}                                                       &\\
	& = G_{pe}\circ\lrb{\Theta,X\vdash A}\circ\langle\Id,\lrb{\Theta\vdash R}\rangle\circ F_{pe}^{\times n}               &(\text{Lemma~\ref{lem:alpha-substitution} (1)})\\
	& = G_{pe}\circ\lrb{\Theta,X\vdash A}\circ \langle F_{pe}^{\times n},\lrb{\Theta\vdash R}\circ F_{pe}^{\times n}\rangle    &\\
	& \cong  G_{pe}\circ\lrb{\Theta,X\vdash A}\circ\langle  F_{pe}^{\times n}, F_{pe}\circ \flrb{\Theta\vdash R}\rangle        &(\text{Theorem~\ref{thm:alpha-def}})\\
	& = G_{pe}\circ\lrb{\Theta,X\vdash A}\circ F_{pe}^{\times (n+1)}\circ\langle \Id,\flrb{\Theta\vdash R}\rangle         &\\
	& = \flrb{\Theta,X\vdash !A}\circ\langle \Id,\flrb{\Theta\vdash R}\rangle                                   &
	\end{align*}
  so that $\gamma^{\Theta \vdash !A [R/X]}  = G_{pe} \lrb{\Theta,X\vdash A} \langle F_{pe}^{\times n},\alpha^{\Theta\vdash R}\rangle$.
	
  Case $\Theta, X \vdash P \odot Q$:
	\begin{align*}
	\flrb{\Theta\vdash (P\odot Q)[R/X]}
  & = \flrb{\Theta\vdash P[R/X]\odot Q[R/X]}                                                                                                                                             &\\
	& = \boxdot_{pe} \circ\langle 	 \flrb{\Theta\vdash P[R/X]},	 \flrb{\Theta\vdash  Q[R/X]}\rangle                                                                                      &\\
	& \cong \boxdot_{pe} \circ\langle 	 \flrb{\Theta,X\vdash P}\circ\langle\Id,\flrb{\Theta\vdash R}\rangle, \flrb{\Theta,X\vdash  Q}\circ\langle\Id,\flrb{\Theta\vdash R}\rangle\rangle  &(\text{IH})\\
	& = \boxdot_{pe} \circ\langle 	 \flrb{\Theta,X\vdash P}, \flrb{\Theta,X\vdash  Q}\rangle\circ\langle\Id,\flrb{\Theta\vdash R}\rangle                                                  &\\
	& = \flrb{\Theta,X\vdash P\odot Q}\circ\langle\Id,\flrb{\Theta\vdash R}\rangle                                                                                                         &
	\end{align*}
  so that $\gamma^{\Theta \vdash (P \odot Q) [R/X]} = \boxdot_{pe} \langle \gamma^{\Theta\vdash P[R/X]},\gamma^{\Theta\vdash Q[R/X]}\rangle$.
	
  Case $\Theta, X \vdash \mu Y.P$:
\begin{align*}
&\phantom{bbb} \flrb{\Theta\vdash \mu Y.P[R/X]} \\
&= \flrb{\Theta,Y\vdash P[R/X]}^\dagger                                                                                              &\\
&\cong \left(\flrb{\Theta, Y, X \vdash P} \circ \langle \Id, \flrb{\Theta, Y \vdash R} \rangle \right)^\dagger                       &\text{(IH and Remark~\ref{rem:dagger})}\\
&= \left(\flrb{\Theta, X, Y \vdash P} \circ \swap_{n,0} \circ \langle \Id, \flrb{\Theta, Y \vdash R} \rangle \right)^\dagger         &(\text{Lemma~\ref{lem:permutation} (2)})\\
&= \left(\flrb{\Theta, X, Y \vdash P} \circ \langle \Id, \flrb{\Theta, Y \vdash R}, \Pi_{n+1} \rangle \right)^\dagger                &\\
&= \left(\flrb{\Theta, X, Y \vdash P} \circ \langle \Id, \flrb{\Theta \vdash R} \circ \drop_{n,0}, \Pi_{n+1} \rangle \right)^\dagger &(\text{Lemma~\ref{lem:contraction} (2)})\\
&= \left(\flrb{\Theta, X, Y \vdash P} \circ (\langle \Id, \flrb{\Theta \vdash R}\rangle \times \Id)  \right)^\dagger                 &\\
&= \flrb{\Theta, X, Y \vdash P}^\dagger \circ \langle \Id, \flrb{\Theta \vdash R}\rangle                                             &(\text{Proposition~\ref{prop:dagger-equal}})\\
&= \flrb{\Theta, X \vdash \mu Y.P} \circ \langle \Id, \flrb{\Theta \vdash R}\rangle                                                  &
\end{align*}
so that $\gamma^{\Theta \vdash \mu Y.P [R/X]} = (\gamma^{\Theta,Y\vdash P[R/X]})^\dagger$.

\textbf{For (3):}
Let $|\Theta| = n$ and let $\vec Z := (Z_1, \ldots, Z_n) \in Ob(\PE^{\times n})$ be an arbitrary object.
We shall prove the equivalent statement 

\[ F_{pe}\gamma^{\Theta\vdash P[R/X]} \circ\alpha^{\Theta \vdash P[R/X]} =  \alpha^{\Theta, X \vdash P} \langle \Id, \flrb{\Theta \vdash R}\rangle \circ \lrb{\Theta, X\vdash P} \langle F_{pe}^{\times n}, \alpha^{\Theta \vdash R} \rangle . \]
We do so by induction on the derivation of $\Theta,X \vdash P.$ 

Case $\Theta, X \vdash \Theta_i$:
\begin{align*}
(F_{pe}\gamma^{\Theta\vdash \Theta_i[R/X]} \circ \alpha^{\Theta \vdash \Theta_i[R/X]})_{\vec Z}
&= \id_{F(Z_i)} \\
&= (\alpha^{\Theta, X \vdash \Theta_i} \langle \Id, \flrb{\Theta \vdash R}\rangle \circ \lrb{\Theta, X\vdash \Theta_i} \langle F_{pe}^{\times n}, \alpha^{\Theta \vdash R} \rangle)_{\vec Z}
\end{align*}

Case $\Theta, X \vdash X$:
\begin{align*}
(F_{pe}\gamma^{\Theta\vdash X[R/X]} \circ \alpha^{\Theta \vdash X[R/X]})_{\vec Z}
&= \alpha^{\Theta \vdash R}_{\vec Z} \\
&= (\alpha^{\Theta, X \vdash X} \langle \Id, \flrb{\Theta \vdash R}\rangle \circ \lrb{\Theta, X\vdash X} \langle F_{pe}^{\times n}, \alpha^{\Theta \vdash R} \rangle)_{\vec Z}
\end{align*}

Case $\Theta, X \vdash !A$:
\begin{align*}
 F_{pe} \gamma^{\Theta\vdash !A[R/X]} \circ \alpha^{\Theta \vdash !A[R/X]} & = F_{pe} G_{pe} \lrb{\Theta,X\vdash A} \langle F_{pe}^{\times n}, \alpha^{\Theta\vdash R} \rangle \circ \id \\
 & = !_e \lrb{\Theta,X\vdash A} \langle F_{pe}^{\times n}, \alpha^{\Theta\vdash R} \rangle \\
 & = \lrb{\Theta,X\vdash !A} \langle F_{pe}^{\times n}, \alpha^{\Theta\vdash R} \rangle \\
 & = \id \circ \lrb{\Theta,X\vdash !A} \langle F_{pe}^{\times n}, \alpha^{\Theta\vdash R} \rangle \\
 & = \alpha^{\Theta, X \vdash !A} \langle \Id, \flrb{\Theta \vdash R}\rangle \circ \lrb{\Theta, X\vdash !A} \langle F_{pe}^{\times n}, \alpha^{\Theta \vdash R} \rangle.
\end{align*}

Case $\Theta, X \vdash P\odot Q:$

Let $|\Theta| = n$ and let $\vec Z := (Z_1, \ldots, Z_n) \in Ob(\PE^{\times n})$ be an arbitrary object. Then the $\vec Z$ component of $\gamma^{\Theta\vdash P[R/X]}$ is a morphism
\[ \gamma^{\Theta\vdash P[R/X]}_{\vec Z}:\flrb{\Theta \vdash P[R/X]}\vec Z\to\flrb{\Theta,X\vdash P}\langle\Id,\flrb{\Theta\vdash R}\rangle\vec Z \]
hence we obtain from Theorem~\ref{thm:boxdot is continuous}:

\begin{align*}
& \phantom{bbb} F_{pe}\big(\gamma_{\vec Z}^{\Theta\vdash P[R/X]}\boxdot_{pe} \gamma_{\vec Z}^{\Theta\vdash Q[R/X]}\big)\circ\beta^{pe}_{\flrb{\Theta\vdash P[R/X]}\vec Z,\flrb{\Theta\vdash Q[R/X]}\vec Z }\\
&=\beta^{pe}_{ \flrb{\Theta,X\vdash P} \langle\Id,\flrb{\Theta\vdash R}\rangle\vec Z,\flrb{\Theta,X\vdash Q} \langle\Id,\flrb{\Theta\vdash R}\rangle\vec Z }\circ\big(F_{pe}\gamma_{\vec Z}^{\Theta\vdash P[R/X]}\odot_e F_{pe}\gamma_{\vec Z}^{\Theta\vdash Q[R/X]}\big),
\end{align*}
hence we find
\begin{align}\label{eq:beta3}
\nonumber & \phantom{bbb} F_{pe}\boxdot\langle\gamma^{\Theta\vdash P[R/X]},\gamma^{\Theta\vdash Q[R/X]}\rangle\circ\beta^{pe}\langle\flrb{\Theta\vdash P[R/X]},\flrb{\Theta\vdash Q[R/X]}\rangle \\
\nonumber &=\beta^{pe}\langle \flrb{\Theta,X\vdash P}\circ\langle\Id,\flrb{\Theta\vdash R}\rangle,\flrb{\Theta,X\vdash Q}\circ\langle\Id,\flrb{\Theta\vdash R}\rangle   \rangle \circ \\
\nonumber &\qquad \odot_e\langle F_{pe}\gamma^{\Theta\vdash P[R/X]},F_{pe}\gamma^{\Theta\vdash Q[R/X]}\rangle\\
&=\beta^{pe}\langle \flrb{\Theta,X\vdash P},\flrb{\Theta,X\vdash Q}\rangle\langle\Id,\flrb{\Theta\vdash R}\rangle\circ\odot_e\langle F_{pe}\gamma^{\Theta\vdash P[R/X]},F_{pe}\gamma^{\Theta\vdash Q[R/X]}\rangle,
\end{align}

Then, surpressing superscripts to improve readability:
\begin{align*}
	& \phantom{bbb} F_{pe}\gamma \circ \alpha \\
  & =F_{pe} \boxdot_{pe} \langle \gamma,\gamma\rangle\circ \beta^{pe} \langle \flrb{\Theta \vdash P[R/X]}, \flrb{\Theta \vdash Q[R/X]} \rangle \circ \odot_e \langle \alpha, \alpha\rangle  \\
	& =  	\beta^{pe}\langle \flrb{\Theta,X\vdash P},\flrb{\Theta,X\vdash Q}\rangle\langle\Id,\flrb{\Theta\vdash R}\rangle\circ\odot_e\langle F_{pe}\gamma,F_{pe}\gamma\rangle \circ\odot_e \langle \alpha, \alpha\rangle\\
	& =  	\beta^{pe}\langle \flrb{\Theta,X\vdash P},\flrb{\Theta,X\vdash Q}\rangle\langle\Id,\flrb{\Theta\vdash R}\rangle \circ\odot_e\langle F_{pe}\gamma\circ \alpha,F_{pe}\gamma\circ \alpha\rangle\\
	& =  	\beta^{pe}\langle \flrb{\Theta,X\vdash P},\flrb{\Theta,X\vdash Q}\rangle\langle\Id,\flrb{\Theta\vdash R}\rangle\circ\\
	& \qquad \odot_e\big\langle \alpha\langle \Id, \flrb{\Theta \vdash R}\rangle \circ \lrb{\Theta, X\vdash P} \langle F_{pe}^{\times n}, \alpha \rangle,\alpha \langle \Id, \flrb{\Theta \vdash R}\rangle \circ \lrb{\Theta, X\vdash Q} \langle F_{pe}^{\times n}, \alpha \rangle\big\rangle\\
	& =  	\beta^{pe}\langle \flrb{\Theta,X\vdash P},\flrb{\Theta,X\vdash Q}\rangle\langle\Id,\flrb{\Theta\vdash R}\rangle\circ\\
	& \qquad \odot_e\big\langle \alpha\langle \Id, \flrb{\Theta \vdash R}\rangle,\alpha\langle \Id, \flrb{\Theta \vdash R}\rangle\big\rangle \circ \odot_e\big\langle \lrb{\Theta, X\vdash P}\langle F_{pe}^{\times n}, \alpha \rangle,\lrb{\Theta, X\vdash Q}\langle F_{pe}^{\times n}, \alpha\rangle\big\rangle \\
	& =  	\beta^{pe}\langle \flrb{\Theta,X\vdash P},\flrb{\Theta,X\vdash Q}\rangle\langle\Id,\flrb{\Theta\vdash R}\rangle\circ\odot_e\langle \alpha,\alpha \rangle\langle \Id, \flrb{\Theta \vdash R}\rangle\circ\\
	& \qquad \odot_e\langle \lrb{\Theta, X\vdash P},\lrb{\Theta, X\vdash Q}\rangle\langle F_{pe}^{\times n}, \alpha \rangle\\
	& =  	\left(\beta^{pe}\langle \flrb{\Theta,X\vdash P},\flrb{\Theta,X\vdash Q}\rangle\circ\odot_e\langle \alpha,\alpha \rangle\right)\langle \Id, \flrb{\Theta \vdash R}\rangle\circ\\
	& \qquad \odot_e\langle \lrb{\Theta, X\vdash P},\lrb{\Theta, X\vdash Q}\rangle\langle F_{pe}^{\times n}, \alpha \rangle \\
	& = \alpha\langle \Id, \flrb{\Theta \vdash R}\rangle\circ\lrb{\Theta,X\vdash P\odot Q}\langle F_{pe}^{\times n}, \alpha \rangle
\end{align*}

where the first and the last equalities are by definition, the second equality
is (\ref{eq:beta3}), the third and the fifth equalities are by functoriality of
$\odot$,  the fourth equality is the induction hypothesis, and the penultimate
equality follows from
$(\rho\circ\sigma)K=\rho K\circ\sigma K$ for any
composable functor $K$ and natural transformations $\rho$ and $\sigma$.

Case $\Theta, X \vdash \mu Y.P$:
\begin{align*}
& \phantom{bbb} F_{pe} \gamma^{\Theta \vdash \mu Y.P[R/X]} \circ \alpha^{\Theta \vdash \mu Y. P[R/X]} \\
&= F_{pe} (\gamma^{\Theta, Y \vdash P[R/X]})^\dagger \circ (\alpha^{\Theta, Y \vdash P[R/X]})^\dagger                                                                                               &(\text{Definition})\\
&= (F_{pe} \gamma^{\Theta, Y \vdash P[R/X]} \circ \alpha^{\Theta, Y \vdash P[R/X]})^\dagger                                                                                                         &(\text{Lemma~\ref{lem:alpha-operations} (1)})\\
&= (\alpha^{\Theta, Y, X \vdash P} \langle \Id, \flrb{\Theta, Y \vdash R}\rangle \circ \lrb{\Theta, Y, X \vdash P} \langle F_{pe}^{\times n+1}, \alpha^{\Theta, Y \vdash R}\rangle)^\dagger         &(\text{IH})\\
&= (\alpha^{\Theta, X, Y \vdash P}\swap_{n,0} \langle \Id, \flrb{\Theta, Y \vdash R}\rangle \circ \\
&\qquad \lrb{\Theta, X, Y \vdash P}\swap_{n,0} \langle F_{pe}^{\times n+1}, \alpha^{\Theta, Y \vdash R}\rangle)^\dagger         &(\text{Lemma~\ref{lem:permutation}})\\
&= (\alpha^{\Theta, X, Y \vdash P} \langle \Pi_1, \ldots, \Pi_n, \flrb{\Theta, Y \vdash R}, \Pi_{n+1}\rangle \circ &\\
&\qquad \lrb{\Theta, X, Y \vdash P}\langle F_{pe} \Pi_1, \ldots, F_{pe} \Pi_n, \alpha^{\Theta, Y \vdash R}, F_{pe} \Pi_{n+1}\rangle )^\dagger         &\\
&= (\alpha^{\Theta, X, Y \vdash P} (\langle \Id, \flrb{\Theta \vdash R} \rangle \times \Id) \circ \lrb{\Theta, X, Y \vdash P} (\langle F_{pe}^{\times n} , \alpha^{\Theta \vdash R} \rangle \times F_{pe}) )^\dagger         &(\text{Lemma~\ref{lem:contraction}})\\
&= (\alpha^{\Theta, X, Y \vdash P})^\dagger \langle \Id, \flrb{\Theta \vdash R} \rangle  \circ \lrb{\Theta, X, Y \vdash P}^\dagger \langle F_{pe}^{\times n} , \alpha^{\Theta \vdash R} \rangle          &(\text{Lemma~\ref{lem:alpha-operations} (2)})\\
&= \alpha^{\Theta, X \vdash \mu Y. P} \langle \Id, \flrb{\Theta \vdash R} \rangle  \circ \lrb{\Theta, X \vdash \mu Y. P} \langle F_{pe}^{\times n} , \alpha^{\Theta \vdash R} \rangle          &(\text{Definition})~~~~\quad\qedhere
\end{align*}

\end{proof}

\end{document}